\documentclass[a4paper,UKenglish,cleveref,nameinlink, autoref, thm-restate]{lipics-v2021}

\hideLIPIcs

\usepackage[rightcaption]{sidecap}

\usepackage{graphics} 
\usepackage{complexity}
\usepackage{amsthm}
\usepackage{amsmath}
\usepackage{amssymb}
\usepackage{xspace}
\usepackage{microtype}
\usepackage{todonotes}
\usepackage{algorithm}
\usepackage[noend]{algpseudocode}
\usepackage{boxedminipage}
\usepackage{soul}
\usepackage{paralist}
\usepackage{hyperref}
\usepackage{cleveref}
\usepackage{cite} 

\DeclareMathOperator{\bigoh}{\mathcal{O}}
\Crefname{observation}{Observation}{Observations}
\Crefname{longobservation}{Observation}{Observations}
\Crefname{algorithm}{Algorithm}{Algorithms}
\Crefname{algorithm}{Listing}{Listings}
\Crefname{section}{Section}{Sections}
\Crefname{property}{Property}{Properties}
\Crefname{longproperty}{Property}{Properties}
\Crefname{lemma}{Lemma}{Lemmas}
\Crefname{longlemma}{Lemma}{Lemmas}
\Crefname{claim}{Claim}{Claims}
\Crefname{longclaim}{Claim}{Claims}
\Crefname{longtheorem}{Theorem}{Theorems}
\Crefname{claimx}{Claim}{Claims}
\Crefname{figure}{Fig.}{Figs.}
\Crefname{figure}{Fig.}{Figs.}
\Crefname{enumi}{Condition}{Conditions}

\Crefname{property}{Property}{Properties}

\newtheorem{longtheorem}{Theorem}
\newtheorem{longlemma}[longtheorem]{Lemma}
\newtheorem{longdefinition}[longtheorem]{Definition}
\newtheorem{longobservation}[longtheorem]{Observation}
\newtheorem{longproposition}[longtheorem]{Proposition}
\newtheorem{longclaim}[longtheorem]{Claim}
\newtheorem{longproperty}[longtheorem]{Property}

\usepackage{xcolor}
\definecolor{defblue}{rgb}{0.274,0.392,0.666}

\let\emph\relax
\DeclareTextFontCommand{\emph}{\color{defblue}\em}
\definecolor{lipicsblue}{rgb}{0.08235294118,0.3098039216,0.537254902}
\hypersetup{
    colorlinks=true,
    linkcolor=lipicsblue,
    filecolor=lipicsblue,      
    urlcolor=lipicsblue,
    citecolor=lipicsblue,
    pdfpagemode=FullScreen,
    }

\newcommand{\CPLS}{\textsc{\textup{CPLS}}}
\newcommand{\CPLSC}{\textsc{\textup{CPLS-Completion}}}

\newcommand{\CPLSF}{\textsc{\textup{CPLSF}}}
\newcommand{\CPLSFC}{\textsc{\textup{CPLSF-Completion}}}

\newcommand{\VVV}{\mathcal{V}}
\newcommand{\NCM}{\textup{MatPenDel}}
\newcommand{\NULL}{\textup{NULL}}
\newcommand{\MS}{\textup{ExtractTriple}}

\nolinenumbers

\title{Exact Algorithms for Clustered Planarity with Linear Saturators}

\author{Giordano {Da Lozzo}}{Roma Tre University, Italy \and \url{http://www.dia.uniroma3.it/~dalozzo}}{giordano.dalozzo@uniroma3.it}{https://orcid.org/0000-0003-2396-5174}{Supported, in part, 
	by MUR of Italy (PRIN Project no.~2022ME9Z78~-- NextGRAAL and PRIN Project
	no.~2022TS4Y3N~-- EXPAND).}
 \author{Robert Ganian}{Algorithms and Complexity Group, TU Wien, Vienna, Austria \and \url{https://www.ac.tuwien.ac.at/people/rganian/}}{rganian@gmail.com}{https://orcid.org/0000-0002-7762-8045}{Supported by the Austrian Science Fund (FWF, Project 10.55776/Y1329) and the Vienna Science and Technology Fund (WWTF, Project 10.47379/ICT22029).}
\author{Siddharth Gupta}{BITS Pilani, K K Birla Goa Campus, India \and \url{https://guptasid.bitbucket.io/}}{siddharthg@goa.bits-pilani.ac.in}{https://orcid.org/0000-0003-4671-9822}{}
\author{Bojan Mohar}{Department of Mathematics, Simon Fraser University, Burnaby, BC, Canada \and \url{https://www.sfu.ca/~mohar/}}{mohar@sfu.ca}{https://orcid.org/0000-0002-7408-6148}{}
\author{Sebastian Ordyniak}{University of Leeds, UK \and \url{https://eps.leeds.ac.uk/computing/staff/8786/dr-sebastian-ordyniak}}{sordyniak@gmail.com}{https://orcid.org/0000-0002-1825-0097}{}
\author{Meirav Zehavi}{Ben-Gurion University of the Negev, Beer-Sheva, Israel \and \url{https://sites.google.com/site/zehavimeirav/}}{meiravze@bgu.ac.il}{https://orcid.org/0000-0002-3636-5322}{}

\Copyright{Giordano {Da Lozzo}, Robert\ Ganian, Siddharth\ Gupta, Bojan\ Mohar, Sebastian\ Ordyniak, Meirav\ Zehavi}

\acknowledgements{This research started at the Dagstuhl Seminar: New Frontiers of Parameterized Complexity in Graph Drawing; seminar number: 23162; April 16-21, 2023.}

\authorrunning{G.\ {Da Lozzo}, R.\ Ganian, S.\ Gupta, B.\ Mohar, S.\ Ordyniak, M.\ Zehavi}


\begin{document}
\maketitle

\begin{abstract}
We study {\sc Clustered Planarity with Linear Saturators}, which is the problem of augmenting an $n$-vertex planar graph whose vertices are partitioned into independent sets (called {\em clusters}) with paths---one for each cluster---that connect all the vertices in each cluster while maintaining~planarity. 
We show that the problem can be solved in time $2^{\mathcal{O}(n)}$
for both the variable and fixed embedding~case. Moreover, we show that it
can be solved in subexponential time $2^{\mathcal{O}(\sqrt{n}\log n)}$ in
the fixed embedding case if additionally the input graph is connected. The
latter time complexity is tight under the Exponential-Time Hypothesis.
We also show that $n$ can be replaced with the vertex cover number of the input graph by providing a linear (resp.\ polynomial) kernel for the variable-embedding (resp.\ fixed-embedding) case; these results contrast the \NP-hardness of the problem on graphs of bounded treewidth (and even on trees).
Finally, we complement known lower bounds for the problem by showing that {\sc Clustered Planarity with Linear Saturators} is \NP-hard even when the number of clusters is at most $3$, thus excluding the algorithmic use of the number of \mbox{clusters as a parameter.}
\end{abstract}

\section{Introduction}

The representation of graphs with a hierarchical structure has become an increasingly crucial tool in the analysis of networked data across various domains. Indeed, by recursively grouping vertices into clusters exhibiting semantic affinity, modern visualization tools allow for the visualization of massive graphs whose entire visualization would otherwise be impossible. Clustered graphs, where a graph's vertex set is partitioned into distinctive subsets, naturally originate in various and diverse fields such as knowledge representation~\cite{DBLP:journals/tog/KamadaK91}, software visualization~\cite{DBLP:conf/sac/PaivaRBL16}, visual statistics~\cite{DBLP:journals/ivs/BrandesL08}, and data mining~\cite{DBLP:phd/dnb/Niggemann01}, only to name a few. 

Formally, a \emph{flat clustered graph} (for short, \emph{clustered graph} or \emph{c-graph}) is a pair $\mathcal{C}=(G,\mathcal{V})$ where $G$ is a graph and $\mathcal{V} = \{V_1,\dots,V_k\}$ is a partition of the vertex set of $G$ into sets~$V_i$ called \emph{clusters}. The graph $G$ is the \emph{underlying graph} of $\cal C$.
A pivotal criterion for a coherent visualization of a clustered graph stems from the classical notion of graph planarity. A \emph{c-planar drawing} of a clustered graph $\mathcal{C}=(G,\mathcal{V})$ is defined as a planar drawing of $G$, accompanied by a representation of each cluster $V_i \in \mathcal V$ as a region $\mathcal{D}_i$ homeomorphic to a closed disc, such that regions associated with different clusters are disjoint, the drawing of the subgraph of $G$ induced by the vertices of each cluster $V_i$ lies in the interior of $\mathcal{D}_i$, and each edge crosses the boundary of a region at most once. The problem of testing for the existence of a c-planar drawing of a clustered graph, called {\sc C-Planarity Testing}, was introduced by Lengauer~\cite{Lengauer89} (in an entirely different context) and then rediscovered by Feng, Cohen, and Eades~\cite{FengCE95}. Determining the complexity of the problem has occupied the agenda of the Graph Drawing community for almost three decades~\cite{AngeliniL19,AngeliniLBFPR15,AngeliniLBFR15,BlasiusR16,ChimaniBF019,LozzoEG018,DBLP:conf/compgeom/CorteseB05,CorteseBFPP08,ChimaniK12,LozzoEGG21,BattistaF09,GoodrichLS05,GutwengerJLMPW02,FulekKMP15,JelinekJKL08,JelinekSTV08,JelinkovaKKPSV09,SneddonB11}. 
The seemingly elusive goal of settling the question regarding the computational complexity of this problem has only been recently addressed by Fulek and T{\'{o}}th~\cite{FulekT22}, who presented a polynomial-time algorithm running in $\bigoh(n^{8})$ time (and in $\bigoh(n^{16})$ time for the version of the problem in which a recursive clustering of the vertices is allowed. It is worth pointing out that, even before a polynomial-time solution for the {\sc C-Planarity} problem was presented, Cortese and Patrignani~\cite{CorteseP18} established that the problem retains the same complexity on both flat and general (i.e., recursively clustered) instances. Subsequently, a more efficient algorithm running in quadratic time has been presented by Bl{\"{a}}sius, Fink, and Rutter~\cite{BlasiusFR23}. 

Patrignani and Cortese studied \emph{independent c-graphs}, i.e., c-graphs where each of the clusters induces an independent set~\cite{CorteseP18}. A characterization given by Di Battista and Frati~\cite{BattistaF09} implies that an independent c-graph is c-planar if and only if the underlying graph can be augmented by adding extra edges, called \emph{saturating edges}, in such a way that the resulting graph has a planar embedding and each cluster induces a tree in the resulting graph. 
Angelini et al.~\cite{AngeliniLBFPR17} considered a constrained version of c-planarity, called {\sc Clustered Planarity with Linear Saturators} (for short, \CPLS), which takes as input an independent c-graph and requires that each cluster induces a path (instead of a tree) in the augmented graph. They proved that the \CPLS\ problem is \NP-complete for c-graphs with an unbounded number of clusters, regardless of whether the input graph is equipped with an embedding or not.
The problem fits the paradigm of augmenting planar graphs with edges in such a way that the resulting graph remains planar, while achieving some other desired property, which is a central question in Algorithmic Graph Theory~\cite{DBLP:conf/soda/FialkoM98,AloupisBCDFM15,ChimaniH16,ChimaniGMW09,KantB91,Poutre94}.

Although \CPLS\ is a topological problem, it stems from a geometric setting within \emph{intersection-link representations}, a form of hybrid representations for locally-dense globally-sparse graphs~\cite{AngeliniLBFPR17}.
Specifically, see also~\cite{AngeliniL20}, given a c-graph $\mathcal{C} = (G, \mathcal{V})$ whose every cluster induces a clique, the {\sc Clique Planarity} problem asks to compute a \emph{clique planar representation} of $\cal C$, i.e., a drawing of $\cal C$ in which each vertex $v \in V(G)$ is represented as a translate~$R(v)$ of a given rectangle $R$, each intra-cluster edge $(u, v)$ is represented by an intersection between $R(u)$ and $R(v)$, and each inter-cluster edge $(u, v)$ is represented by a Jordan arc connecting the boundaries of $R(u)$ and $R(v)$ that intersects neither the interior of any rectangle nor the representation of any other inter-cluster edge. Notably, the authors showed that a c-graph whose every cluster induces a clique admits a clique planar representation if and only if it admits a so-called {\em canonical representation}, where the vertices are squares arranged in a ``linear fashion''; see~\cref{fig:clique-planar}. This allowed them to establish the equivalence \mbox{between \CPLS\ and {\sc Clique Planarity}.} In particular, they proved that a c-graph $\cal C$ whose every cluster induces a clique is a yes-instance of {\sc Clique Planarity} if and only if the c-graph obtained by removing all intra-cluster edges from $\cal C$ is a yes-instance~of~\CPLS.

\begin{figure}
    \centering
    \begin{subfigure}{.32\textwidth}
    \includegraphics[page=1,width=\textwidth]{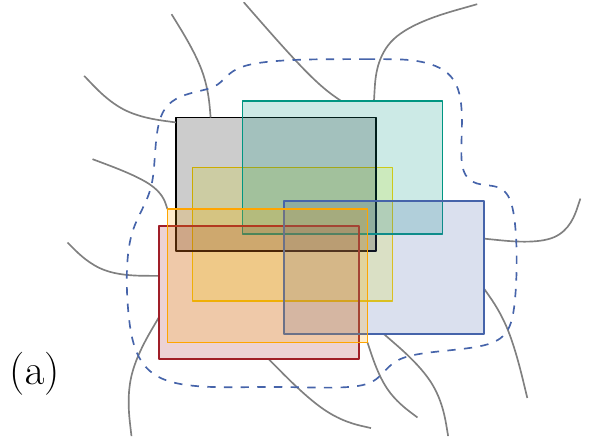}
    \end{subfigure}
    \begin{subfigure}{.32\textwidth}
    \includegraphics[page=2,width=\textwidth]{figures/canonical.pdf}
    \end{subfigure}
    \begin{subfigure}{.32\textwidth}
    \includegraphics[page=3,width=\textwidth]{figures/canonical.pdf}
    \end{subfigure}
    \caption{(a) A partial clique planar representation focused on a cluster. (b) Canonical representation of the cluster in (a). A linear saturation of the vertices of the cluster corresponding to (b).}
    \label{fig:clique-planar}
\end{figure}

\smallskip \noindent
\textbf{Our Contribution.}\quad In this paper, we study the \CPLS\
problem from a computational perspective, in both the fixed embedding
as well as the variable embedding setting. In the fixed embedding
case, the underlying graph of the c-graph comes with a prescribed
combinatorial embedding, which must be preserved by the output drawing. Instead, in the variable embedding setting, we are allowed to select the embedding of the underlying graph. To distinguish these two settings, we refer to the former problem (i.e., where a fixed embedding is provided as part of the input) as \CPLSF. 
Our main results are as follows.

\noindent\textbf{(1)} In \cref{se:exactAlgo} we give exact single-exponential and sub-exponential algorithms for the problems. In particular, both \CPLS\ and \CPLSF\ can be solved in $2^{\mathcal{O}(n)}$ time. Moreover, in \cref{se:exactAlgoFix} we obtain a subexponential $2^{\mathcal{O}(\sqrt{n}\log n)}$ algorithm for \CPLSF\ when the underlying graph is connected; this result is essentially tight under the Exponential Time Hypothesis~\cite{ImpagliazzoPZ01}. In both cases, the main idea behind the algorithms is to use a divide-and-conquer approach that separates the instance according to a hypothetical separator in the solution graph.

\noindent
\textbf{(2)} In \cref{se:fptAlgo} we obtain polynomial kernels (and thus establish fixed-parameter tractability) for both \CPLS\ and \CPLSF\ with respect to the vertex cover number of the underlying graph. Interestingly, while being provided with an embedding allowed us to design a more efficient exact algorithm for \CPLSF, in the parameterized setting the situation is reversed: we obtain a linear kernel for \CPLS, but for \CPLSF\ the size of the kernel is cubic in the vertex cover number.
Combining the former result with our exact algorithm for \CPLS\ allows us to obtain an algorithm that runs in single-exponential time with respect to the vertex cover number.
    
\noindent \textbf{(3)} In \cref{se:npc} we observe that the {\sc CPLS} problem is \NP-complete on trees and even a disjoint union of stars. Since stars have treedepth, pathwidth, and treewidth one, this charts an intractability border between the vertex cover number parameterization used in \cref{se:fptAlgo} and other parameters.
Then we prove that the problem is \NP-complete even for c-graphs having at most $3$ clusters, thus strengthening the previously known hardness result for an unbounded number of clusters. This result combined with the equivalence between {\CPLS} and {\sc Clique Planarity} shows that  {\sc Clique Planarity} is \NP-complete for instances with a bounded number of clusters, which solves an open problem posed in~\cite[OP 4.3]{AngeliniL20}.

\smallskip
\noindent \textbf{Further Related Work.}\quad
Apart from the classical {\sc C-Planarity Testing} problem, several variants of planarity for clustered graphs have attracted considerable attention. Didimo et al.~\cite{DBLP:journals/jgaa/DidimoGL08} studied {\sc Overlapping Clustered Planarity}, where the clusters of the c-graph are not required to define a partition of the vertex set, but are instead allowed to overlap, and thus the disks corresponding to overlapping clusters may also intersect; see also~\cite{AthenstadtHN14,AthenstadtC17,JohnsonP87}.
Angelini et al.~\cite{AngeliniLBFPR15} studied a relaxed version of \mbox{\sc C-Planarity Testing}, where either crossings among clusters disks (\emph{region-region crossings}) or crossings among cluster disks and edges (\emph{edge-regions crossings}) are allowed, while the underlying graph must be drawn planar.
Forster et al.~\cite{DBLP:conf/sofsem/ForsterB04} studied {\sc Clustered-Level Planarity}, where the underlying graph $G$ is a level graph and they seek a c-planar drawing that contains a level planar drawing of $G$ and such that each cluster disk intersects the \mbox{levels of the drawing in a continuous segment; see also~\cite{AngeliniLBFR15}.}

\section{Preliminaries}\label{se:preliminaries}

For a positive integer $k$, we denote by $[k]$ the set $\{1,\dots,k\}$.

\subparagraph{Graphs.} 
Let $G$ be a graph. We denote by $V(G)$ and $E(G)$ the vertex set and the edge set of $G$, respectively. When it is clear from the context, we denote $V(G)=V$ and $E(G)=E$. Let $U \subseteq V \cup E$. We denote by $G[U]$ the subgraph of $G$ induced by $U$.

A \emph{vertex cover} of $G$ is a subset $U \subset V$ such that every edge $e \in E$ has at least an endpoint in $U$. The \emph{vertex cover number} of $G$ is the minimum size of a vertex cover of $G$. An \emph{independent set} of $G$ is a subset $U \subset V$ such that there does not exist any two vertex $u,v \in U$ such that $(u,v) \in E$. 
The graph $G$ is \emph{bipartite} if the vertices of $G$ can be partitioned into two non-empty independent sets $V_1, V_2 \subset V$ such that every edge of $G$ connects a vertex in $V_1$ to a vertex in $V_2$.
The graph $G$ is \emph{connected} if there is a path between any two vertices, and it is \emph{disconnected} otherwise.

\subparagraph{Planar Graphs and Embeddings.}
A \emph{drawing} of a graph is a mapping of each vertex to a distinct point in the plane and of each edge to a Jordan curve between its endpoints. A drawing is \emph{planar} if there are no edge crossings, except possibly at their common endpoints. 
A graph is \emph{planar} if it admits a planar drawing. A planar drawing partitions the plane into topologically connected regions called \emph{faces}. The bounded faces are the \emph{inner faces}, while the unbounded face is the \emph{outer face}. 

A \emph{combinatorial embedding} (for short, \emph{embedding}) is an equivalence class of planar drawings, where two drawings of a graph are \emph{equivalent} if they determine the same \emph{rotation} at each vertex (i.e., the same circular ordering of the edges around each vertex). As the problem considered in this paper is topological in nature and the actual geometry does not matter, we often talk about embeddings as if they are actual planar drawings. When this happens, we are referring to any planar drawing within the equivalence class. An \emph{embedded graph} $G_{\cal E}$ is a planar graph $G$ equipped with an embedding $\cal E$. A drawing $\Gamma$ of $G_{\cal E}$ is a planar drawing of $G$ such that $\Gamma \in \cal E$. 

Let $G_{\cal E}$ be an embedded graph. A \emph{noose} $N$ of $G_{\cal E}$ is a simple closed curve in 
some drawing $\Gamma$ of $G_{\cal E}$ that
\begin{inparaenum}[(i)]
    \item intersects $G$ only at vertices and
    \item traverses each face of $\Gamma$ at most once.
\end{inparaenum}
Given a subgraph $H$ of $G$, we denote by ${\cal E}(H)$ the embedding of $H$ obtained from $\cal E$ by restricting it to $H$.

\subparagraph{Parameterized Complexity.}
A problem $\Pi$ is a \emph{parameterized} problem if each problem instance of $\Pi$ is associated with a \emph{parameter} $k$. The goal of the framework of Parameterized Complexity is to confine the combinatorial explosion in the running time of an algorithm for an \NP-hard parameterized problem $\Pi$ to depend only on $k$.  In particular, a parameterized problem $\Pi$ is \emph{fixed-parameter
tractable (FPT)} if any instance $(I, k)$ of $\Pi$ is solvable in time $f(k) \cdot |I|^{\bigoh(1)}$, where $f$ is an arbitrary computable function of $k$. Moreover, a parameterized problem $\Pi$ is para-NP-hard if it is \NP-hard even for some fixed constant value of the parameter~$k$.  

Another concept related to fixed-parameter tractability is polynomial kernel. A parameterized problem $\Pi$ is said to admit a \emph{polynomial kernel} 
if there exists a polynomial-time algorithm that transforms an arbitrary parameterized instance of $\Pi$ to an equivalent (not necessarily parameterized) instance of $\Pi$ whose size is bounded by a polynomial function $g$ of the parameter of the original instance.

\smallskip \noindent
\textbf{Clustered Planarity with Linear Saturators.}\quad
Let $\cal C$ be a c-graph. We say that $\cal C$ has a \emph{fixed embedding} if the underlying graph of $\cal C$ is an embedded graph, and has a \emph{variable embedding} otherwise. We say that $\cal C$ is an \emph{embedded c-graph} if $\cal C$ has a fixed embedding.

Let $\mathcal{C}=(G,\{V_1,\dots,V_k\})$ be an independent c-graph, i.e., for every $i \in [k]$, $V_i$ is an  independent set of $G$. We say that $G$ can be \emph{linearly saturated} if there exist sets $Z_1,\dots,Z_k$ of non-edges of $G$ such that
\begin{inparaenum}[(i)]\
    \item $H = (V(G), E(G) \cup Z)$ for $Z=\bigcup^k_{i=1}Z_i$ is planar,
    \item for every $i \in [k]$, each edge in $Z_i$ connects two vertices of $V_i$ in $H$, and
    \item for every $i \in [k]$, the graph $H[V_i]$ is a path.
\end{inparaenum}
For $i \in [k]$, the edges in $Z_i$ are the \emph{saturating edges} of cluster $V_i$, and $H$ is the \emph{linear saturation} of $G$ \emph{via} $Z_1,\dots,Z_k$. 
We now define the {\sc Clustered Planarity with Linear Saturators} (for short, \CPLS) and the {\sc Fixed Embedding Clustered Planarity with Linear Saturators} (for short, \CPLSF) problems.

\medskip
\noindent $\blacktriangleright${\kern 0.33em}{\bf\CPLS:} Given an independent c-graph $(G, \VVV)$, does there exist a linear saturation of $G$?
\smallskip

\noindent $\blacktriangleright${\kern 0.33em}{\bf\CPLSF:} Given an independent embedded c-graph $(G_{\cal E}, \VVV)$, does there exist a linear saturation $H$ of $G$ that admits an embedding $\cal E'$ for which ${\cal E}'(G)$ coincides with $\cal{E}$?
\medskip

To devise exact algorithms for the \CPLS\ and \CPLSF\ problems, it will be useful to consider a more general setting. A c-graph is \emph{paths-independent} if each of its clusters induces a collection of paths; the notion of a linear saturator for a paths-independent c-graph is the same as that of an independent c-graph.
We now define the {\sc Clustered Planarity with Linear Saturators Completion} (for short, \CPLSC) and the {\sc Fixed Embedding Clustered Planarity with Linear Saturators Completion} (for short, \CPLSFC) problems as follows.

\begin{longdefinition}[{\bf \CPLSC\ Problem}]
    Given a paths-independent c-graph $(G, \VVV)$, does there exist a linear saturation of $G$?
\end{longdefinition}

\begin{longdefinition}[{\bf \CPLSFC\ Problem}]
    Given a paths-independent embedded c-graph $(G_{\cal E}, \VVV)$, does there exist a linear saturation $H$ of $G$ such that there exists an embedding $\cal E'$ of $H$ for which ${\cal E}'(G)$ is same as $ \cal{E}$?
\end{longdefinition}

Clearly, \CPLSC\ (resp., \CPLSFC) is a generalization of  \CPLS\ (resp., \CPLSF). 
Moreover, observe that in case of \CPLS\ (resp., \CPLSF) $H[V_i] = H[Z_i]$ as every cluster induces an independent set in $G$ which is not true in the case of \CPLSC\ (resp., \CPLSFC).

\section{Exact Algorithms for \CPLS\ and \CPLSF}\label{se:exactAlgo}
This section details our exact single- and sub-exponential algorithms for \CPLS\ and \CPLSF.
We first deal with the variable-embedding case.

\smallskip \noindent
\textbf{Proof Ideas.}
We aim to solve the
problem via a divide-and-conquer approach, where at each iteration, we
``split'' the current instance of the problem into simpler (and, in
particular, substantially smaller) sub-instances of the problem. To
understand how to perform the split, consider an (unknown) solution
$Z$, and the graph $H=(V(G),E(G)\cup Z)$. As this graph is planar,
there exists a noose $N$ that intersects only $\bigoh(\sqrt{|V(G)|})$
vertices of $H$ and does not intersect any edges of $H$, such that
both the interior and the exterior of $N$ contain a constant fraction
(roughly between 1/3 to 2/3) of the vertices of $G$ (\cref{prop:planarSeparator}). Thus, naturally, we would like to split our problem instance into two instances, one corresponding to the interior and boundary of $N$, and the other corresponding to the exterior and boundary of $N$ (so, the boundary is common to both).
However, two issues arise, which we describe next.

The first (simpler) issue is that we do not know $N$ since we do not know $Z$. However, since $N$ intersects only few vertices, we can simply ``guess'' $N$ by guessing the set $U$ of intersected vertices, and the cyclic order $\rho$ in which $N$ intersects them. By guessing, we mean that we iterate over all possible options, and aim to find at least one that yields a solution (if a solution exists). Having $U$ and $\rho$ at hand, we still do not know the interior and exterior of $N$, and therefore, we still do not know how to create the two simpler sub-instances. Thus, we perform additional guesses: We guess the set $I$ of vertices drawn (with respect to the planar drawing of $H$) strictly inside $N$ and thereby also the set $O$ of vertices drawn strictly outside $N$. Additionally, for the set of edges having both endpoints in $U$, we guess a partition $\{E_\mathsf{in}, E_\mathsf{out}\}$ that encodes which of them are drawn inside $N$ and which of them are drawn outside $N$. Specifically for the fixed embedding case, we non-trivially exploit the given embedding to perform the guesses of $I,O$ and $\{E_\mathsf{in},E_\mathsf{out}\}$ in a \mbox{more sophisticated manner that yields only subexponentially many guesses.} 

The second (more complicated) issue is that we cannot just create two instances: One for the subgraph of $G$ induced by $I\cup U$ (and without the edges in $E_\mathsf{out}$) and the other for the subgraph of $G$ induced by $O\cup U$ (and without the edges in $E_\mathsf{in}$) and solve them independently. The two main concerns are the following: First, we need a single planar drawing for the entire (unknown) graph $H$ and so the drawings of the two graphs in the two sub-instances should be ``compatible''. Second, we may not need to (and in some cases, in fact, must not) add edges to a graph in any of the two sub-instances to connect all vertices in the same cluster in that graph into a single path. Instead, we need to create a collection of paths, so that the two collections that we get, one for each of the \mbox{two sub-instances, will together yield a single path.}

To handle the second concern, we perform additional guesses. Specifically, we guess some information on how the (solution) cluster paths in $H$ ``behave'' when they are restricted to the interior of~$N$--- we guess a triple $(M,P,D)$ which encodes, roughly speaking, a pairing~$M$ between some vertices in $U$ that are connected by the cluster paths only using the interior of~$N$, the set of vertices $P$ through which the cluster paths enter the interior of~$N$ and ``do not return'', and the set of vertices $D$ that are incident to two edges in $Z$ drawn in the interior. Now, having such a guess at hand, we handle both concerns by defining a special graph that augments the graph induced by $G[U\cup O]$ (with the edges in $E_\mathsf{in}$ removed) so that the solutions returned for the corresponding sub-instance will have to be, in some sense, ``compatible'' with $(M,P,D)$ as well as draw $O$ and $E_\mathsf{out}$ outside $N$ while preserving the order~$\rho$. The definition of this augmented graph is, perhaps, the most technical definition required for the proof, as it needs to handle both concerns simultaneously. Among other operations performed to obtain this augmented graph, for the first concern, we add extra edges between vertices in~$M$, attach pendants on $P$, and treat vertices in $D$ as if they belong to their own clusters, and for the second concern, we triangulate the result in a~careful~manner.

\subsection{The Variable-Embedding Case}\label{se:exactAlgoVar}
Towards solving \CPLS, we recursively solve the more general problem mentioned earlier, namely, \CPLSC. Additionally, we suppose that some of the vertices of the input graph can be marked, and that we are not allowed to add edges incident to marked vertices. To avoid confusion, we will denote this annotated version by {\sc CPLS-Completion*}.
The rest of this section is devoted to the proof of the following.

\begin{theorem}\label{thm:exactVarialbe}
Let $\mathcal{C}=(G,\mathcal{V})$ be an $n$-vertex paths-independent c-graph. It can be tested whether $\mathcal{C}$ is a positive instance of {\sc CPLS-Completion*} in $8^{n+\bigoh(\sqrt{n}\log n)}=2^{\bigoh(n)}$ time.
\end{theorem}

\smallskip
\noindent
\textbf{Main Definitions and Intermediate Results.}\quad
We start with the following definition.

\begin{definition}[{\bf Non-Crossing Matchings and the Partition \NCM}]\label{def:NCM}
Let $\mathcal{C}=(G,\mathcal{V})$ be a paths-independent c-graph. Let $\rho$ be a cyclic ordering of some subset $U$ of $V(G)$. A matching $M$ is a {\em non-crossing matching} of $\rho$ if it is a matching in the graph $H=(U,\{\{a,b\}: a,b\in U\})$ such that for every pair of edges $\{a,b\},\{c,d\}\in M$, when we traverse $U$ in the cyclic order $\rho$, starting with $a$, we either encounter $b$ before both $c$ and $d$, or we encounter both $c$ and $d$ before $b$.
Denote by $\NCM(\VVV,\rho)$ (which stands for {\em Matching, Pendants and Deleted}) the set of all triples $(M,P,D)$ such that:
\begin{itemize}
\item $M$ is a non-crossing matching of $\rho$ such that each edge of $M$ matches only vertices in the same cluster, and 
\item $P,D \subseteq U\setminus V(M)$ are disjoint sets, where $V(M)$ is the set of vertices matched by $M$.
\end{itemize}
\end{definition}

Intuitively, $\rho$ will represent a cyclic balanced separator of $G$,
i.e., a noose in a drawing of the solution graph that separates the solution graph into two almost
equally sized subgraphs, $M$ will represent path segments between pairs of 
vertices on $\rho$, $P$ (``pendants'') will represent vertices through which the paths leave $\rho$ never to return,\footnote{Thus, for a (non-partial) solution, $P$ contains at most $2$ vertices per cluster, though we do not need to formally demand this already in~\cref{def:extractTriple}.} and $D$ will represent degree-$2$ vertices on the aforementioned path segments (which will be, in a sense, deleted when we ``complement'' the triple). This will be formalized in \cref{def:extractTriple} ahead.

\begin{observation}
\label{obs:computeNCM}
Let $\mathcal{C}=(G,\mathcal{V})$ be a paths-independent c-graph. Let
$\rho$ be a cyclic ordering of some subset $U$ of $V(G)$. Then,
$|\NCM(\VVV,\rho)|=2^{\bigoh(|U|)}$. Moreover, the set
$\NCM(\VVV,\rho)$ can be computed in time
$2^{\bigoh(|U|)}$.
\end{observation}

\begin{proof}
The number of possibilities for choosing a partition $(V(M),P,F,U\setminus (V(M)\cup P\cup D))$ of $U$ is bounded by $4^{|U|}$. Clearly, they can be generated within time $2^{\bigoh(|U|)}$. Now, fix such a partition. Then, it is well-known that the number of non-crossing matchings of $\rho$ is the $t$-th Catalan number, where $t=|M|$, which is bounded by $2^t$, and that all of these matchings can be generated within time $2^{\bigoh(t)}$ (see, e.g., \cite{dorn2012catalan}).  Among these, we sieve those matching only vertices in the same cluster within time $2^{\bigoh(t)}$. Thus, the observation follows.
\end{proof}

We formalize the notion of a {\em partial solution} as follows.

\begin{definition}[{\bf Partial Solution}]
Let $\mathcal{C}=(G,\mathcal{V})$ be a paths-independent c-graph. A {\em cluster path} is a path all of whose vertices belong to the same cluster in $\cal V$. 
A {\em partial solution for $\cal C$} is a set $\cal S$ of vertex-disjoint cluster paths. Note that $\cal S$ can contain several cluster paths whose vertices belong to the same cluster. We say that $\cal S$ is {\em  properly marked} if every edge in $\cal S$ incident to a marked vertex also belongs to $G$.
\end{definition}

Next, we formalize how a triple in \NCM\ captures information on a
partial solution (see \cref{fig:defExtractTriple}). For intuition, think of $\cal S$ as if it consisted
of paths that contain vertices only from the exterior (or only from
the interior) of a cyclic separator. 

\begin{figure}[t!]
\begin{center}
\begin{subfigure}{.49\textwidth}
\centering
\includegraphics[scale=1]{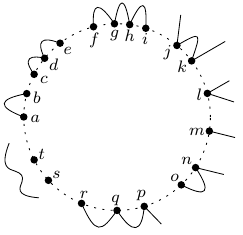}
\vspace{-6mm}
\caption{}
\label{fig:defExtractTriple}
\end{subfigure}
\hfil
\begin{subfigure}{.49\textwidth}
\centering
\includegraphics[scale=1]{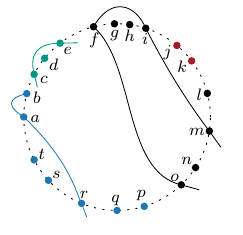}
\vspace{-6mm}
\caption{}
\label{fig:defComplementaryTriples}
\end{subfigure}
\end{center}
\vspace{-0.5cm}
\caption{(a) Example for \cref{def:extractTriple}: The paths in $\cal S$ are drawn as black curves, and the vertices in $U$ are marked by disks. We have $M=\{\{a,b\},\{c,e\},\{f,i\}\}, P=\{m,o,r\}, D=\{d,g,h,j,k,l,n,p,q\}$ and $U\setminus (V(M)\cup P\cup D)=\{s,t\}$.
(b) Example for \cref{def:complement}: The vertices in $U$ are marked by disks, and their association with the clusters is indicated by~colors (explicitly, the clusters are $\{\{a,b,p,q,r,s,t\},\{c,d,e\},\{f,g,h,i,l,m,n,o\},\{j,k\}\}$). Suppose $M_{\mathsf{in}}=\{\{a,b\},\{c,e\},\{f,i\}\}, P_{\mathsf{in}}=\{m,o,r\}, D_{\mathsf{in}}=\{d,g,h,j,k,l,n,p,q\}$ and $U\setminus (V(M_{\mathsf{in}})\cup P_{\mathsf{in}}\cup D_{\mathsf{in}})=\{s,t\}$, and $M_{\mathsf{out}}=\{\{f,o\},\{i,m\},\{r,a\}\}$, $P_{\mathsf{out}}=\{c,e\}$, $D_{\mathsf{out}}=\{s,t\}$ and $U\setminus (V(M_{\mathsf{out}})\cup P_{\mathsf{out}}\cup D_{\mathsf{out}})=\{b,d,g,h,j,k,l,n,p,q\}$. Then, the triples $T_{\mathsf{in}}=(M_{\mathsf{in}},P_{\mathsf{in}},D_{\mathsf{in}})$ and $T_{\mathsf{out}}=(M_{\mathsf{out}},P_{\mathsf{out}},D_{\mathsf{out}})$ are complementary, and $G_{T_{\mathsf{in}},T_{\mathsf{out}}}$ is a subgraph of the illustrated graph induced by $\{a,b,c,e,f,i,m,o,r\}$. The pendants are the endpoints of the edges going inside or outside the cycle (the vertex $b$, for example, is not adjacent to a pendant, while the vertex $c$ is).}
\end{figure}

\begin{definition}[{\bf Extracting a Triple from a Partial Solution}]\label{def:extractTriple}
Let $\mathcal{C}=(G,\mathcal{V})$ be a paths-independent c-graph. Let
$U\subseteq V(G)$ and let ${\cal S}$ be a partial solution. Then, $\MS(U,{\cal S})=(M,P,D)$ is defined as follows: 
\begin{itemize}
\item $M$ has an edge between the endpoints of every path in $\cal S$ that has both endpoints in $U$,
\item $P\subseteq U$ consists of the vertices of degree $1$ in $\cal S$ belonging to $U$ that are
  not in $V(M)$,\footnote{In other words, for every path in $\cal S$ having precisely one endpoint in $U$, $P$ contains the endpoint in $U$.} and
\item $D\subseteq U$ consists of the vertices of degree $2$ in $\cal S$ belonging to $U$. 
\end{itemize}
Further, ${\cal S}$ is {\em compatible with} a cyclic ordering $\rho$ of $U$ if $M$ is a non-crossing matching of $\rho$.
\end{definition}

We have the following immediate observation, connecting \cref{def:NCM,def:extractTriple}.

\begin{observation}\label{obs:inNCM}
Let $\mathcal{C}=(G,\mathcal{V})$ be a paths-independent c-graph. Let $\rho$ be a cyclic ordering of some subset $U$ of $V(G)$.     Let $\cal S$ be a partial solution compatible with $\rho$. Then, $\MS(U,{\cal S}) \in \NCM(\VVV,\rho)$.
\end{observation}

We will be interested, in particular, in triples which are
complementary (see \cref{fig:defComplementaryTriples}), which
intuitively means that partial solutions for the inside and the
outside described by the two triples can be
combined to a solution for the whole graph:

\begin{definition}[{\bf Complementary Triples in \NCM}]\label{def:complement}
Let $\mathcal{C}=(G,\mathcal{V})$ be a paths-independent c-graph. Let $\rho$ be a cyclic ordering of some subset $U$ of $V(G)$.
  Then, $T_\mathsf{in}=(M_\mathsf{in},P_\mathsf{in},D_\mathsf{in}),T_\mathsf{out}=(M_\mathsf{out},P_\mathsf{out},D_\mathsf{out})\in \NCM(\VVV,\rho)$ are {\em complementary} if:
   \begin{enumerate}
    \item\label{def:complement1} $D_\mathsf{in}\subseteq U\setminus (V(M_{\mathsf{out}})\cup P_{\mathsf{out}}\cup D_{\mathsf{out}})$, and $D_\mathsf{out}\subseteq U\setminus (V(M_{\mathsf{in}})\cup P_{\mathsf{in}}\cup D_{\mathsf{in}})$.\footnote{The reason why we write $\subseteq$ rather than $=$ is that some vertices in, e.g., $U\setminus (V(M_{\mathsf{out}})\cup P_{\mathsf{out}}\cup D_{\mathsf{out}})$ can be the endpoints of solution paths. For example, consider the vertex $b$ in \cref{fig:defComplementaryTriples}.}
    \item\label{def:complement2} Let $G_{T_\mathsf{in},T_\mathsf{out}}$ denote the graph on vertex set $V(M_{\mathsf{out}})\cup P_{\mathsf{out}}\cup V(M_{\mathsf{in}})\cup P_{\mathsf{in}}$ and edge set $M_\mathsf{in}\cup M_{\mathsf{out}}$, such that for every vertex in $P_\mathsf{in}$, and similarly, for every vertex in $P_{\mathsf{out}}$, we add a new vertex attached to it and belonging to the same cluster.\footnote{So, if a vertex belongs to both $P_\mathsf{in}$ and $P_\mathsf{out}$, we attach two new vertices to it.} Then, this graph is a collection of paths, such that all vertices in its vertex set that belong to the same cluster in $\VVV$ also belong to the same (single) path, and all vertices that belong to the same path also belong to the same cluster.
   \end{enumerate}
\end{definition}

The utility of complementary triples is in the following definition,
which specifies when a partial solution ${\cal S}$ for the
inner part satisfying $\MS(U,{\cal S})=T_\mathsf{out}$
can be combined with any partial solution ${\cal S}'$ for the outer part \mbox{that
satisfies $\MS(U,{\cal S}')=T_\mathsf{in}$.}

\begin{definition}[{\bf Partial Solution Compatible with $(T_\mathsf{in},I,E_\mathsf{in})$}]\label{def:compatiblePartSol}
Consider $(G,\VVV)$, a cyclic ordering $\rho$ of some $U\subseteq V(G)$, $T_\mathsf{in}=(M_{\mathsf{in}},P_{\mathsf{in}},D_{\mathsf{in}}) \in \NCM(\VVV,\rho)$,
  $I \subseteq V(G)\setminus U$, and $E_\mathsf{in}\subseteq E(G[U])$. Then, a partial solution $\cal S$ is {\em compatible} with $(T_{\mathsf{in}},I,E_{\mathsf{in}})$ if:
  \begin{enumerate}
    \item\label{def:compatiblePartSol1} $\MS(U,{\cal S})=T_\mathsf{out}$ and $T_\mathsf{in}$ are complementary.
    \item\label{def:compatiblePartSol2} Let $G'_\mathsf{in}=G[I\cup U]-(E(G[U])\setminus E_\mathsf{in})$ and $G_\mathsf{in}={G'}_\mathsf{in}-D_{\mathsf{in}}$. We have that $\cal S$ contains all and only the vertices in $G_{\mathsf{in}}$, and all (but not necessarily only)\footnote{As the paths in $\cal S$ can contain edges that are not edges in $G$.} edges in $G_{\mathsf{in}}$ between vertices in the same~cluster.
    \item\label{def:compatiblePartSol3} There exists a planar drawing $\varphi_\mathsf{in}$ of $G'_\mathsf{in}\cup E({\cal S})$ with an inner-face whose boundary contains $U$ (with, possibly, other vertices) ordered as by $\rho$.
    \item\label{def:compatiblePartSol4} Each path in ${\cal S}$ satisfies one of the following conditions: (a) it consists of all vertices of a cluster in $\VVV$ that belong to $G_\mathsf{in}$, and has no endpoint in~$U$, or (b) it has an endpoint in $U$.
  \end{enumerate}
\end{definition}

Towards the statement that will show the utility of compatibility (in
\cref{lem:unionYieldsSolution} ahead), we need one more definition,
which intuitively provides necessary conditions for obtaining a
solution for $(G,\VVV)$ from a partial solution that is compatible with $(T_\mathsf{in},I,E_\mathsf{in})$.

\begin{definition}[{\bf Sensibility of $(T_\mathsf{in},I,E_\mathsf{in})$}]
\label{def:sensible}
Consider $(G,\VVV)$, a cyclic ordering $\rho$ of some $U\subseteq V(G)$, $T_\mathsf{in}=(M_{\mathsf{in}},P_{\mathsf{in}},D_{\mathsf{in}}) \in \NCM(\VVV,\rho)$,
  $I \subseteq V(G)\setminus U$, and $E_\mathsf{in}\subseteq E(G[U])$. Then, $(T_\mathsf{in},I,E_\mathsf{in})$ is {\em sensible} if:
  \begin{enumerate}
    \item\label{def:sensibleItem1} There is no edge $\{u,v\}\in E(G)$ with  $v\in I$ and $u\in O$ for $O=V(G)\setminus(I\cup U)$.
    \item\label{def:sensibleItem2}  No vertex in $D_\mathsf{in}$ is adjacent in $G'_\mathsf{in}=G[I\cup U]-(E(G[U])\setminus E_\mathsf{in})$ to a vertex in the same cluster in $\VVV$. Additionally, no vertex in $U\setminus (V(M_\mathsf{in})\cup P_\mathsf{in}\cup D_\mathsf{in})$ is adjacent in $G'_\mathsf{out}=G[O\cup U]-E_\mathsf{in}$ for $O=V(G)\setminus (I\cup U)$ to a vertex in the same cluster in $\VVV$.
      \item\label{def:sensibleItem3}  No cluster in $\VVV$ has non-empty intersection with both $I\cup (U\setminus D_\mathsf{in})$ and $O\cup D_\mathsf{in}$ but not with $V(M_\mathsf{in})\cup P_\mathsf{in}$.
  \end{enumerate}
\end{definition}

We show that compatible solutions yield solutions to {\sc CPLS-Completion*} in $\NCM$.

\begin{lemma}\label{lem:unionYieldsSolution}
Consider $(G,\VVV)$, a cyclic ordering $\rho$ of some $U\subseteq V(G)$, $T_\mathsf{in}=(M_{\mathsf{in}},P_{\mathsf{in}},D_{\mathsf{in}}) \in \NCM(\VVV,\rho)$,
  $I \subseteq V(G)\setminus U$, and $E_\mathsf{in}\subseteq E(G[U])$. Suppose that $(T_\mathsf{in},I,E_\mathsf{in})$ is sensible. Additionally, consider
  \begin{itemize}
      \item a properly marked partial solution ${\cal S}_\mathsf{in}$ compatible with $(T_{\mathsf{in}},I,E_{\mathsf{in}})$, and
      \item a properly marked partial solution ${\cal S}_\mathsf{out}$ compatible with $\MS(U,{\cal S}_\mathsf{in})=(T_{\mathsf{out}},O,E_\mathsf{out})$ where $O=V(G)\setminus (I\cup U)$ and $E_\mathsf{out}=E(G[U])\setminus E_\mathsf{in}$.
  \end{itemize}
Then, $Z=E({\cal S}_\mathsf{in}\cup{\cal S}_\mathsf{out})\setminus E(G)$ is a solution to $(G,\VVV)$ as an instance of {\sc CPLS-Completion*}.
\end{lemma}

\begin{proof}
We first comment that, since ${\cal S}_\mathsf{in}$ and ${\cal S}_\mathsf{out}$ are properly marked, it is clear that $Z$ does not contain any edge incident to a marked vertex.

\medskip
\noindent{\bf Verifying Planarity.} We first argue that $G\cup Z$ is a planar graph. Since ${\cal S}_\mathsf{in}$ compatible with $(T_{\mathsf{in}},I,E_{\mathsf{in}})$, \cref{def:compatiblePartSol3} in \cref{def:compatiblePartSol} implies that there exists a planar drawing $\varphi_\mathsf{in}$ of $G'_\mathsf{in}\cup E({\cal S}_\mathsf{in})$ with an inner-face whose boundary contains $U$ (with, possibly, other vertices) ordered as by $\rho$, where $G'_\mathsf{in}=G[I\cup U]-E_\mathsf{out}$. 
Similarly, since ${\cal S}_\mathsf{out}$ compatible with $(T_{\mathsf{out}},I,E_{\mathsf{out}})$, \cref{def:compatiblePartSol3} implies there exists a planar drawing $\varphi_\mathsf{out}$ of $G'_\mathsf{out}\cup E({\cal S}_\mathsf{out})$ with an inner-face whose boundary contains $U$ (with, possibly, other vertices) ordered as by $\rho$, where $G'_\mathsf{out}=G[O\cup U]-E_\mathsf{in}$. Clearly, we can flip $\varphi_\mathsf{out}$ to obtain another planar drawing $\varphi'_\mathsf{out}$ of $G'_\mathsf{out}\cup E({\cal S}_\mathsf{out})$ whose outer-face's boundary contains $U$ (with, possibly, other vertices) ordered as by $\rho$. Due to \cref{def:sensibleItem1} in \cref{def:sensible}, we know that the (non-disjoint) union of $G'_\mathsf{in}$ and $G'_\mathsf{out}$ yields $G$. 
Thus, by taking the (non-disjoint) union of the drawings $\varphi_\mathsf{in}$ and $\varphi'_\mathsf{out}$, we obtain a planar drawing of $G\cup Z$. So, $G\cup Z$ is a planar graph.

\medskip
\noindent{\bf Verifying That All Edges Are Covered.} Consider some cluster $V_i\in \VVV$, and let $Z_i$ denote the subset of edges in $Z$ with both endpoints in $V_i$. We need to prove that $G[V_i]\cup Z_i$ is a path. Let $Z_{i,\mathsf{in}}=Z_i\cap E(S_\mathsf{in})$ and $Z_{i,\mathsf{out}}=Z_i\cap E(S_\mathsf{out})$. Let ${\cal S}_{i,\mathsf{in}}$ (resp., ${\cal S}_{i,\mathsf{out}}$) denote the collection of paths in ${\cal S}_\mathsf{in}$ (resp., ${\cal S}_\mathsf{out}$) between vertices in $V_i$.
Since ${\cal S}_\mathsf{in}$ is compatible with
$(T_{\mathsf{in}},I,E_{\mathsf{in}})$,
\cref{def:compatiblePartSol2} in
\cref{def:compatiblePartSol} implies that
$G_\mathsf{in}[V_i]\cup Z_{i,\mathsf{in}}$, where
$G_\mathsf{in}=G'_\mathsf{in}-D_\mathsf{in}$, is exactly the
collection of paths ${\cal S}_{i,\mathsf{in}}$. Similarly, since
${\cal S}_\mathsf{out}$ is compatible with $(T_{\mathsf{out}},I,E_{\mathsf{out}})$, \cref{def:compatiblePartSol2} implies that $G_\mathsf{out}[V_i]\cup Z_{i,\mathsf{out}}$, where $G_\mathsf{out}=G'_\mathsf{out}-D_\mathsf{out}$, is exactly the collection of paths~${\cal S}_{i,\mathsf{out}}$. 

Recall that we have already argued that the (non-disjoint) union of $G'_\mathsf{in}$ and $G'_\mathsf{out}$ equals $G$.
Next, we further argue that the (non-disjoint) union of $G_\mathsf{in}[V_i]$ and $G_\mathsf{out}[V_i]$ equals $G[V_i]$. To this end, it suffices to show that:  (I) $D_\mathsf{in}\cap D_{\mathsf{out}}=\emptyset$ (and hence $D_\mathsf{in}\cap V_i\subseteq V(G_\mathsf{out})$ and $D_\mathsf{out}\cap V_i\subseteq V(G_\mathsf{in})$), and (II) there does not exist an edge between a vertex in $D_\mathsf{in}\cap V_i$ (resp., $D_\mathsf{out}\cap V_i$) and a vertex in $V(G'_\mathsf{in})\cap V_i$ (resp., $V(G'_\mathsf{out})\cap V_i$). Because ${\cal S}_\mathsf{in}$ is compatible with $(T_{\mathsf{in}},I,E_{\mathsf{in}})$, we have that $T_\mathsf{out}$ and $T_\mathsf{in}$ are complementary, and hence, by \cref{def:complement1} in \cref{def:complement}, we get that $D_\mathsf{in}\cap D_\mathsf{out}=\emptyset$, therefore Item (I) holds. Because $(I_\mathsf{in},I,E_\mathsf{in})$ is sensible, Item (II) directly follows from \cref{def:sensibleItem2} in \cref{def:sensible}.

Overall, we conclude, so far, that $G[V_i]\cup Z_i$ is exactly the (non-disjoint) union of the  two collections of paths ${\cal S}_{i,\mathsf{out}}$ and ${\cal S}_{i,\mathsf{in}}$. Next, we verify that this union yields a single path.

\medskip
\noindent{\bf Verifying That the Union Yields a Single Path.} Recall that $\MS(U,{\cal S}_\mathsf{in})=T_\mathsf{out}=(M_\mathsf{out},P_\mathsf{out},D_\mathsf{out})$, and let $\MS(U,{\cal S}_\mathsf{out})=T'_\mathsf{in}=(M'_\mathsf{in},P'_\mathsf{in},D'_\mathsf{in})$. Note that, possibly, $T_\mathsf{in}\neq T'_\mathsf{in}$. By the definition of $\MS$:
\begin{enumerate}
\item\label{condition:MS1} $M_\mathsf{out}$ (resp., $M'_\mathsf{in}$) has an edge between the endpoints of every path in ${\cal S}_{\mathsf{in}}$ (resp., ${\cal S}_{\mathsf{out}}$) that has both endpoints in $U$, 
\item\label{condition:MS2} $P_\mathsf{out}\subseteq U$ (resp., $P'_\mathsf{in}\subseteq U$) consists of the vertices of degree $1$ in ${\cal S}_\mathsf{out}$ (resp., ${\cal S}_\mathsf{in}$) belonging to $U$ that are not in $V(M_\mathsf{out})$ (resp., $V(M'_\mathsf{in})$), and
\item\label{condition:MS3} $D_\mathsf{out}\subseteq U$ (resp., $D'_\mathsf{in}\subseteq U$) consists of the vertices of degree $2$ in ${\cal S}_{\mathsf{in}}$ (resp., ${\cal S}_{\mathsf{out}}$) belonging to $U$.
\end{enumerate}

Further, notice that $S_{\mathsf{in}}$ does not contain vertices in $D_\mathsf{in}$, and ${\cal S}_\mathsf{out}$ does not contain vertices in $D_\mathsf{out}$. In particular, the latter implies that $D_\mathsf{out}\subseteq U\setminus (V(M_\mathsf{in}')\cup P_\mathsf{in}'\cup D_\mathsf{in}')$. 
Since ${\cal S}_\mathsf{out}$ is compatible with $(T_{\mathsf{out}},I,E_{\mathsf{out}})$, by \cref{def:compatiblePartSol1} in \cref{def:compatiblePartSol},  $T_\mathsf{in}'$ and $T_\mathsf{out}$ are complementary. 
 So, by \cref{def:complement1} in \cref{def:complement}, $D_\mathsf{in}'\subseteq U\setminus(V(M_\mathsf{out})\cup P_\mathsf{out}\cup D_\mathsf{out})$. Thus, all vertices in $D_\mathsf{out}\cup D_\mathsf{in}'$ are of degree $2$ in ${\cal S}_{i,\mathsf{out}}\cup {\cal S}_{i,\mathsf{in}}$.

Due to \cref{def:sensibleItem3} in \cref{def:sensible} and \cref{def:compatiblePartSol4} in \cref{def:compatiblePartSol}, we know that if $V_i$ contains vertices only from $G_\mathsf{in}$ or only from $G_\mathsf{out}$, then ${\cal S}_{i,\mathsf{out}}\cup {\cal S}_{i,\mathsf{in}}$ is indeed a path (more precisely, one of these two sets is a path and the other is empty), and that, otherwise, every path in ${\cal S}_{i,\mathsf{in}}$ or in ${\cal S}_{i,\mathsf{out}}$ has an endpoint in $U\setminus (D_\mathsf{in}\cup D_\mathsf{out})$. In this latter case, by the conditions implied by the definition $\MS$, which are specified above, and by \cref{def:complement2} in \cref{def:complement}, we again derive that ${\cal S}_{i,\mathsf{in}}\cup {\cal S}_{i,\mathsf{out}}$ is a path. This completes the proof.
\end{proof}

The key player in partitioning an instance of the {\sc CPLS-Completion*} problem into two smaller instances is that of the augmented graph, defined as follows. 

\begin{SCfigure}[3.7][t]
\includegraphics[scale=1.3]{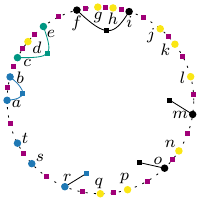}
\centering
\caption{Example for \cref{def:augmented}. Here, the vertices in $U$ are marked by disks, and the clusters in $\VVV$ are $\{\{a,b,p,q,r,s,t\},\{c,d,e\},\{f,g,h,i,l,m,n,o\},\{j,k\}\}$. Suppose $M=\{\{a,b\},\{c,e\},\{f,i\}\}, P=\{m,o,r\}, D=\{d,g,h,j,k,l,n,p,q\}$ and $U\setminus (V(M)\cup P\cup D)=\{s,t\}$. Regarding $\widetilde{\cal C}$, the new vertices are marked by squares, vertices belonging to their own singleton clusters are drawn in yellow, and the other clusters are: (i) $a,b,r,s,t$ and the two neighboring squares drawn inside the cycle (blue); (ii) $\{c,e\}$ and the neighboring square drawn inside the cycle (green); (iii) $\{f,i,m,o\}$ and the three neighboring squares drawn inside the cycle (black).}
\label{fig:augment}
\end{SCfigure}

\begin{longdefinition}[{\bf Augmented Graph}]\label{def:augmented}
Let $\mathcal{C}=(G,\mathcal{V})$ be a paths-independent c-graph. Let $\rho$ be a cyclic ordering of some subset $U$ of $V(G)$.
Let $(M,P,D) \in \NCM(\VVV,\rho)$. 
Towards the definition of the {\em augmented graph} \Call{AugmentGraph}{$(G,\VVV),\rho,(M,P,D)$}, we define an ``intermediate'' paths-independent c-graph $\widetilde{\mathcal{C}}=(\widetilde{G},\widetilde{\mathcal{V}})$ as follows:\footnote{Here, vertices in $V(M)\cup P$ will have degree 3, new vertices will have degree $2$ or $1$, and the remaining vertices in $U$ (some of which belong to $D$) will have degree $2$.}
\begin{enumerate}
  \item\label{augment1} Start with $\widetilde{G}$ being the edge-less graph on $U$ and with $\widetilde{\cal V}$ being the restriction of $\cal V$ to $U$. Add edges to obtain the cycle given by $\rho$. Subdivide each one of these edges, and make every new vertex be in its own new cluster. As $\widetilde{G}$ is a cycle, it has a unique embedding in the plane.
  \item\label{augment2} For every $e \in M$, add a new vertex adjacent to the endpoints of $e$ and belonging to the same cluster. Considering $\rho$ as the outer face, the added vertex and its associated edges have a unique face they can be inserted in. Thus, the resulting graph has a unique embedding with $\rho$ as its outer face.
  \item\label{augment3} For every $p \in P$, add a new vertex adjacent to $p$ and belonging to the same cluster. By an analogous argument as the one in the previous step, the resulting graph has a unique embedding with $\rho$ as its outer face.
  \item\label{augment4} Remove every vertex in $D$ from its cluster, and create a new one-vertex cluster for it.
 \end{enumerate}
Let $\Gamma$ be the unique planar embedding of $\widetilde{G}$ having $\rho$ as its outer face. Augment $\widetilde{G}$ to an ``inner triangulation'' as follows:
  \begin{enumerate}
  \item For every inner face $f$, add a vertex $v_f$ (belonging to its own new cluster) and make it adjacent to all vertices on the boundary of $f$. This yields an embedding $\Gamma'$ of a $2$-connected graph $G'$.
  \item For every inner face $f'$ in $\Gamma'$, we add a vertex $v_f'$ (belonging to its own new cluster) and make it adjacent to all vertices on the boundary of $f'$. This yields a unique embedding of a $3$-connected graph $G^*$.
  \end{enumerate}  
  The augmented graph is the (non-disjoint) union of $G^*$ and $G$. The vertices that are marked are those marked in $G$ as well as the new vertices made adjacent to vertices in $P$ while belonging to the same cluster.
\end{longdefinition}

We have the following immediate observation.
\begin{longobservation}\label{obs:restrictFaces}
Let $\mathcal{C}=(G,\mathcal{V})$ be a paths-independent c-graph. Let $\rho$ be a cyclic ordering of some subset $U$ of $V(G)$.
Let $(M,P,D) \in \NCM(\VVV,\rho)$. Then, the boundary of every
inner-face in $G^\star$ (see \cref{def:augmented}) consists of exactly three vertices
, such that to at most one of them new edges can be added,\footnote{Since the other two are either marked, or belong to singleton clusters, or already have two neighbors in the same cluster.} and the one to which edges can be added (if it exists) belongs to $U$, which is part of the cycle constructed in Step~\ref{augment1} in \cref{def:augmented}.
\end{longobservation}
The central statement about augmented graphs is the following.

\begin{lemma}\label{lem:propertiesOfAug}
Consider $(G,\VVV)$, a cyclic ordering $\rho$ of some $U\subseteq V(G)$, $T_\mathsf{in}=(M_{\mathsf{in}},P_{\mathsf{in}},D_{\mathsf{in}}) \in \NCM(\VVV,\rho)$,
  $I \subseteq V(G)\setminus U$, and $E_\mathsf{in}\subseteq E(G[U])$.  Let $Z_\mathsf{in}\neq$NULL\footnote{We use NULL to algorithmically represent the non-existence of a solution (particularly in pseudocode).} be a solution to the instance $(G_A,\VVV_A)=$\Call{AugmentGraph}{$G'_\mathsf{in},\rho, T_\mathsf{in}$} of {\sc CPLS-Completion*} (where $G'_\mathsf{in}=G[I\cup U]-(E(G[U])\setminus E_\mathsf{in})$). Then, in polynomial time, we can compute a partial solution compatible with $(T_{\mathsf{in}},I,E_{\mathsf{in}})$ that is marked properly.     
\end{lemma}

\begin{proof}
Let ${\cal S}'_\mathsf{in}$ denote the subgraph of $G_A\cup Z_\mathsf{in}$ on all vertices, and on the edges in $G_A$ between vertices in the same cluster (w.r.t.~$\VVV_A$) and the edges in $Z_\mathsf{in}$. Further, let $\varphi'$ be a planar drawing of $G_A\cup Z_\mathsf{in}$.  Clearly, ${\cal S}'_\mathsf{in}$ is a partial solution that is marked properly, and $\varphi'$ coincides with the unique drawing of the $3$-connected graph $G^*$ from \cref{def:augmented}.

In what follows, we first show how to modify ${\cal S}'_\mathsf{in}$ and $\varphi'$ so that the only vertices drawn inside the cycle $C$ constructed in Step~\ref{augment1} in \cref{def:augmented} will be the new vertices (existing in $G_A$ but not in $G$). 

\medskip
\noindent{\bf Ensuring That the Inside of $C$ is ``Clean''.} Let $\cal Q$ be the connected components of $(G_A\cup Z_\mathsf{in})-V(G^*)$ drawn inside $C$. Due to \cref{obs:restrictFaces}, and since  $Z_\mathsf{in}$ is a solution to $(G_A,\VVV_A)$, we know that each component $Q\in{\cal Q}$ satisfies the following: $Q$ has a vertex adjacent in $G_A$ to a vertex $v\in V(M_\mathsf{in})\cup P_\mathsf{in}$, where $v\in V(C)$; further, $Q$ has no other vertex adjacent in $G_A\cup Z_\mathsf{in}$ to $v$ or to any other vertex in $(G_A\cup Z_\mathsf{in})-V(Q)$. Then, we simply redraw each such component $Q$ (along with the edge connecting it to $v$) in some face outside $C$ whose boundary contains $v$ (at least one such face exists). We let $\widetilde{\cal S}_\mathsf{in}$ and $\widetilde{\varphi}$ denote the results of these modifications, which can trivially be done in polynomial time. It is clear that $\widetilde{\cal S}_\mathsf{in}$ is a partial solution that is marked properly. Additionally, we let ${\cal S}_\mathsf{in}$ denote the properly marked partial solution obtained from ${\cal S}_\mathsf{in}$ by: (i) removing, from each of its paths having one or two endpoints not in $G$, these one or two endpoints (which must be vertices adjacent to vertices in $P_\mathsf{in}$ belonging to the same cluster as them); (ii) removing every path consisting of a single vertex on its own new singleton cluster;\footnote{Thus, we altogether remove all the vertices not in $G$ as well as all the vertices in $D_\mathsf{in}$.} (iii) removing, from all paths, all edges .belonging to $M_\mathsf{in}$, which might split some paths into several paths. No additional modifications will be made.  So, in what follows, we verify that ${\cal S}_\mathsf{in}$ is compatible with $(T_\mathsf{in},I,E_\mathsf{in})$, which will complete the proof.

\medskip\noindent{\bf Verifying Compatibility.} Here, we verify each of the conditions in \cref{def:compatiblePartSol}.

\smallskip
\noindent{\bf\em \cref{def:compatiblePartSol1}.} We need to show that $\MS(U,{\cal S}_\mathsf{in})=T_\mathsf{out}=(M_\mathsf{out},P_\mathsf{out},D_\mathsf{out})$ and $T_\mathsf{in}$ are complementary. To this end, we need to verify the satisfaction of the two conditions in \cref{def:complement}. By the construction of the augmented graph, all vertices in $D_\mathsf{in}$ belong to their own clusters, and so ${\cal S}_\mathsf{in}$ cannot include newly added edges (i.e., which do not already belong to $G'_\mathsf{in}$) incident to them, implying that $D_\mathsf{in}\subseteq U\setminus (V(M_\mathsf{out})\cup P_\mathsf{out}\cup D_\mathsf{out})$. Furthermore, by its construction, every vertex in $V(M_\mathsf{in})\cup P_\mathsf{in}$ is incident to one edge that is removed from $\widetilde{\cal S}_\mathsf{in}$ when constructing ${\cal S}_\mathsf{in}$, which implies that $D_\mathsf{out}\subseteq U\setminus (V(M_\mathsf{in})\cup P_\mathsf{in}\cup D_\mathsf{in})$.
Thus, \cref{def:complement1} in \cref{def:complement} is satisfied. For the second condition, notice that in $\widetilde{\cal S}_\mathsf{in}$, every cluster in $\VVV_A$ corresponds to a single path. Now, contract the edges in $\widetilde{\cal S}_\mathsf{in}$ belonging to $G'_\mathsf{in}$ in the augmented graph, and keep only edges connecting vertices belonging to the same cluster---then, we obtain the graph $G_{T_\mathsf{in},T_\mathsf{out}}$, excluding, possibly, for every vertex in $P_\mathsf{out}$, the newly added edge incident to it when constructing $G_{T_\mathsf{in},T_\mathsf{out}}$.
So, the graph we created is a collection of paths as required in \cref{def:complement2} in \cref{def:complement}, and it stays so also after adding the ``missing'' edges (if any) since by the definition of $P_\mathsf{out}$, the vertices within it are endpoints of paths in $\widetilde{\cal S}_\mathsf{in}$. Thus, by the definition of ${\cal S}_\mathsf{in}$ based on $\widetilde{\cal S}_\mathsf{in}$, we conclude that \cref{def:complement2} in \cref{def:complement} is satisfied.

\smallskip
\noindent{\bf\em \cref{def:compatiblePartSol2}.} Because the augmented graph contains $G_\mathsf{in}$, and when creating ${\cal S}_\mathsf{in}$, we removed all vertices and edges that do not belong to this graph, we derive the satisfaction of this condition.

\smallskip
\noindent{\bf\em \cref{def:compatiblePartSol3}.} The drawing $\widetilde{\varphi}$ certifies the satisfaction of this condition.

\smallskip
\noindent{\bf\em \cref{def:compatiblePartSol4}.} Recall that every path in $\widetilde{\cal S}_\mathsf{in}$ consists of all and only the vertices of a single cluster in $\VVV_A$. Since all of the edges that are removed when we create ${\cal S}_\mathsf{in}$ from $\widetilde{\cal S}_\mathsf{in}$ are incident to vertices in $U$, we derive the satisfaction of this condition.
\end{proof}

\smallskip
\noindent
\textbf{The Algorithm and Its Correctness.}\quad
We present the pseudocode of the algorithm as
\cref{alg:gammaone}. Here, \Call{BruteForce}{$G,\VVV$} is a
procedure that iterates over all edge subsets, and for each one of
them, checks whether it is a solution (which, in particular, means
that it does not include edges incident to marked vertices). It either returns such a subset (if one exists) or NULL. Notice that the call made in Lines~\ref{line:callAug2} and~\ref{line:fix2} are valid (i.e., we make the calls with triples in $\NCM$) due to \cref{obs:inNCM}.
\algrenewcommand\algorithmicindent{1.0em}%
\begin{algorithm}[b!]
    \caption{}
    \label{alg:gammaone}
    \begin{algorithmic}[1]
        \Procedure{Solve}{$G,\VVV$}
          \If{$|V(G)| \leq 100$}\Comment{The number $100$ represents
            an arbitrary constant.}
          \State \Return \Call{BruteForce}{$G,\VVV$}
          \EndIf
          \For{$U \subseteq V(G)$ with $|U|\leq 2 \sqrt{2|V(G)|}$}
            \For{each cyclic ordering $\rho$ of $U$}
              \For{$I \subseteq V(G)\setminus U$ with $|I|\leq
                2|V(G)\setminus U|/3$ and $|I|\geq |V(G)\setminus U|/3$}
                \For{each partition $\{E_\mathsf{in},E_\mathsf{out}\}$ of $E(G[U])$}                  
                  \For{$T_\mathsf{in}=(M_\mathsf{in},P_\mathsf{in},D_\mathsf{in}) \in \NCM(\VVV,\rho)$ (enumerate by Obs.~\ref{obs:computeNCM})}
                    \If{$(T_\mathsf{in},I,E_\mathsf{in})$ is not sensible (\cref{def:sensible})}\label{line:isSensible}               
                    \State break
                    \EndIf  
                
                    \State $(G_A,\VVV_A) \leftarrow$ \Call{AugmentGraph}{$G[I\cup U]-E_\mathsf{out},\rho, T_\mathsf{in}$} 
                    \State $Z_\mathsf{in} \leftarrow$ \Call{Solve}{$G_A,\VVV_A$}

                    \If{$Z_\mathsf{in} \neq \NULL$}
                    \State Use the algorithm in \cref{lem:propertiesOfAug} w.r.t.~$T_\mathsf{in},I,E_\mathsf{in},Z_\mathsf{in}$ to get ${\cal S}_\mathsf{in}$.
                
                      \State $T_\mathsf{out}=(M_\mathsf{out},P_\mathsf{out},D_\mathsf{out}) \leftarrow \MS(U,{\cal S}_\mathsf{in})$
                      \State $O=V(G)\setminus (I\cup U)$
                      \State\label{line:callAug2} $(G_B,\VVV_B) \leftarrow$ \Call{AugmentGraph}{$G[O\cup U]-E_\mathsf{in},\rho,T_\mathsf{out}$}
                      \State$ Z_\mathsf{out} \leftarrow$ \Call{Solve}{$G_B,\VVV_B$}

\If{$Z_\mathsf{out} \neq \NULL$}
                      \State\label{line:fix2} Use the algorithm in \cref{lem:propertiesOfAug} w.r.t.~$T_\mathsf{out},O,E_\mathsf{out},Z_\mathsf{out}$ to get ${\cal S}_\mathsf{out}$.
                      
                      \State \Return $E({\cal S}_\mathsf{in}\cup {\cal S}_\mathsf{out})\setminus E(G)$
                      \EndIf
                      \EndIf
                  \EndFor
                \EndFor
              \EndFor
            \EndFor
          \EndFor  
          \Return $\NULL$
        \EndProcedure        
        \algrenewcommand\algorithmicindent{1.0em}%
    \end{algorithmic}
\end{algorithm}
\algrenewcommand\algorithmicindent{2.0em}%

For the proof of correctness, we will use \cref{lem:unionYieldsSolution,lem:propertiesOfAug}, and the following proposition.

\begin{longproposition}[\cite{lipton1979separator}]\label{prop:planarSeparator}
Let $G$ be a graph with a planar drawing $\varphi$. Then, there exists
a noose $N$ (intersection $\varphi$ only at the vertices of $G$) such that:
\begin{itemize}
    \item The set of vertices  of $G$ that $N$ intersects, say, $U$, satisfies $|U|\leq 2\sqrt{2|V(G)|}$.
    \item The set of vertices contained in the strict interior of $N$,
      say, $I$, satisfies $|I|\leq 2|V(G)\setminus U|/3$ and $|I|\geq
      |V(G)\setminus U|/3$.
\end{itemize}
\end{longproposition}

We are now ready to conclude the correctness of the algorithm.

\begin{longlemma}\label{lem:algoCorrect}
\cref{alg:gammaone} solves the {\sc CPLS-Completion*} problem.
 \end{longlemma}

\begin{proof}
The proof is by induction on $|V(G)|$. When $|V(G)|\leq 100$, correctness is trivial since the algorithm performs brute-force. So, next suppose that $|V(G)|>100$ and that the algorithm is correct whenever it is called with a graph having fewer vertices than $|V(G)|$.

\medskip
\noindent{\bf Reverse Direction.} For the proof of this direction, suppose that the algorithm returns $Z=E({\cal S}_\mathsf{in}\cup{\cal S}_\mathsf{out})\setminus E(G)$. Then, we need to prove that $Z$ is a solution. Let $U,\rho,I,E_\mathsf{in},\mathsf{out},T_\mathsf{in},Z_\mathsf{in},T_\mathsf{out},O,Z_\mathsf{out}$ be the entities corresponding to the iteration when $Z$ is returned.
Due to Line~\ref{line:isSensible} in \cref{alg:gammaone}, $(T_\mathsf{in},I,E_\mathsf{in})$ is sensible. Moreover, due to \cref{lem:propertiesOfAug}, we have that ${\cal S}_\mathsf{in}$ is a partial solution compatible with $(T_{\mathsf{in}},I,E_{\mathsf{in}})$,  and
 ${\cal S}_\mathsf{out}$ is a partial solution compatible with $(T_{\mathsf{out}},O,E_\mathsf{out})$, with both being marked properly.
So, due to \cref{lem:unionYieldsSolution}, the proof of this direction is complete.

\medskip
\noindent{\bf Forward Direction.} For the proof of this direction, suppose a solution $Z^\star$ exists. We need to prove that the algorithm does not return NULL (and hence, by the correctness of the reverse direction, this would mean that it returns a solution). Let $G^\star$ denote the graph obtained from $G$ by adding the edges in $Z^\star$. Let ${\cal S}^\star$ denote the paths corresponding to the different clusters in $G^\star$. Let $\varphi$ be some planar drawing of $G^\star$. By \cref{prop:planarSeparator}, there exists a noose $N$ such that 
\begin{itemize}
    \item The set of vertices  of $G^\star$ that $N$ intersects, say, $U$, satisfies $|U|\leq 2\sqrt{2|V(G)|}$, and $N$ does not intersect any edge of $G^\star$.
    \item The set of vertices of $G^\star$ contained in the strict interior of $N$, say, $I$, satisfies $|I|\leq 2|V(G)|/3$ and $|I|\geq |V(G)\setminus S|/3$.
\end{itemize}
Let $\rho$ denote the cyclic ordering in which $N$ traverses $U$, and let $\{E_\mathsf{in},E_\mathsf{out}\}$ denote the partition of $E(G[U])$ according to which edges are drawn (by $\varphi$) on the inside and on the outside of $N$. 
Let $\{Z^\star_\mathsf{in},Z^\star_\mathsf{out}\}$ denote the partition of $Z^\star$ according to which edges are drawn on the inside and on the outside of $N$. Let ${\cal S}^Z_\mathsf{out}$ denote the partial solution obtained from the subgraph of $G^\star$ drawn inside and on the boundary of $N$ (i.e., having $N$ on the ``outside'' in its outer-face) when we only keep edges connecting vertices in the same cluster. Since $G^\star$ is planar, ${\cal S}^Z_\mathsf{out}$ is compatible with $\rho$. Let $T_\mathsf{in}=\MS(U,{\cal S}^Z_\mathsf{out})$. By \cref{obs:inNCM}, $T_\mathsf{in}\in\NCM(\VVV,\rho)$. In what follows, consider the iteration of \cref{alg:gammaone} corresponding to these $U,I,\{E_\mathsf{in},E_\mathsf{out}\}$ and $T_\mathsf{in}$.

First, we claim that $(T_\mathsf{in},I,E_\mathsf{in})$ is sensible. Here, \cref{def:sensibleItem1,def:sensibleItem3} follow from the choice of $N$, ensuring that $U$ separates $I$ from $O=V(G)\setminus (U\cup I)$ in $G^\star$ (and, so, also in particular in $G$). Also, from this, and since in ${\cal S}^\star$, every vertex in $D_\mathsf{in}$ (resp., $U\setminus V(M_\mathsf{in}\cup P_\mathsf{in}\cup D_\mathsf{in})$) only has neighbors in $G'_\mathsf{in}$ (resp., $G'_\mathsf{out}$) while ${\cal S}^\star$ includes all edges between vertices in the same cluster, we get that \cref{def:sensibleItem2} is satisfied as well. 

Second, we claim that $Z_\mathsf{in}$ computed in the iteration under consideration must not be NULL. To this end, it suffices to note that  $Z^\star_\mathsf{in}$ is a solution to the instance $(G_A,\VVV_A)$ of {\sc CPLS-Completion*} constructed by the algorithm. Similarly, $Z_\mathsf{out}$ computed in the iteration under consideration must not be NULL, since  $Z^\star_\mathsf{out}=Z^\star\setminus Z^\star_\mathsf{in}$ is a solution to the instance $(G_B,\VVV_B)$ of {\sc CPLS-Completion*} constructed by the algorithm. Hence, in the iteration under consideration, the algorithm will return a set of edges, which, in particular, means that it does not return NULL.
\end{proof} 

\subparagraph{Running Time Analysis.}

\begin{longlemma}\label{lem:runtime}
    \cref{alg:gammaone} runs in time $8^{n+\bigoh(\sqrt{n}\log n)}=2^{\bigoh(n)}$.
\end{longlemma}

\begin{proof}
Let $T(n)$ denote the running time of the algorithm when called with a graph on $n$ vertices. When $n\leq 100$, this is bounded by a constant.

When $n>100$, the algorithm iterates over $\binom{n}{\leq 2 \sqrt{2n}}=2^{\bigoh(\sqrt{n}\log n)}$ sets $U$. For each set $U$, the algorithm iterates over $|U|!=(2\sqrt{2n})!=2^{\bigoh(\sqrt{n}\log n)}$ orderings $\rho$.
For each ordering $\rho$, since $G[U]$ is a planar graph on $|U|\leq 2\sqrt{2n}$ vertices, which, hence, has at most $\bigoh(\sqrt{n})$ edges, we have that the algorithm iterates over $2^{\bigoh(\sqrt{n})}$ partitions $\{E_\mathsf{in},E_\mathsf{out}\}$ of $E(G[U])$. For each partition $\{E_\mathsf{in},E_\mathsf{out}\}$, the algorithm iterates over at most $\sum_{i=n/3}^{2n/3}\binom{n}{i}\leq 2^n\cdot n$ sets $I$. For each set $I$, \cref{obs:computeNCM} implies that the algorithm iterates over $2^{\bigoh(\sqrt{n})}$ triples $T_\mathsf{in}$. So, in total, the algorithm performs $2^{n+\bigoh(\sqrt{n}\log n)}$ iterations. For each triple, the algorithm performs two recursive calls on graphs with at most $2n/3 + \bigoh(\sqrt{n})$ vertices each, and other internal calculations in polynomial time. So,

    $$T(n) = 2^{n+\bigoh(\sqrt{n}\log n)}\cdot (2\cdot T(\frac{2n}{3}+\bigoh(\sqrt{n})) + n^{\bigoh(1)}) = 2^{\frac{2^0}{3^0}n+\frac{2^1}{3^1}n+\frac{2^2}{3^2}n+\frac{2^3}{3^3}n+\cdots}\cdot 2^{\bigoh(\sqrt{n}\log n)},$$

which evaluates to $8^{n+\bigoh(\sqrt{n}\log n)}$. This completes the proof. 
\end{proof}

From \cref{lem:algoCorrect,lem:runtime}, we conclude the correctness of \cref{thm:exactVarialbe}.

\subsection{Fixed Embedding}\label{se:exactAlgoFix}
In this section, we provide a single-exponential algorithm for \CPLSF{}, which becomes
subexponential if additionally the input graph is connected.

\begin{theorem}\label{thm:solvefixed}
  \CPLSFC{} (and thus also \CPLSF{}) can be solved in time~$2^{\bigoh(n)}$. Moreover, it can be solved in time~$2^{\bigoh(\sqrt{n}\log n)}$ if the input graph is connected.
\end{theorem}

\newcommand{\curv}{\textsf{curv}}
\newcommand{\FB}{\textsf{FB}}
\newcommand{\inside}[2]{G_{\mathsf{in}}(#1,#2)}
\newcommand{\outside}[2]{G_{\mathsf{out}}(#1,#2)}
\newcommand{\insideR}[2]{I_#1(#2)}
\newcommand{\outsideR}[2]{O_#1(#2)}

\newcommand{\bl}[2]{B_{#1}(#2)}
\newcommand{\blI}[2]{B_{#1}(#2)}

\algrenewcommand\algorithmicindent{1.0em}%
\begin{algorithm}[t!]
    \caption{The algorithm for \CPLSFC{}.}
    \label{alg:solvefixed}
    \begin{algorithmic}[1]
        \Procedure{SolveCPLSFC}{$G_{\mathcal{E}}$,$\VVV$}
          \If{$|V(G)| \leq 100$}
          \State \Return \Call{BruteForceCPLSFC}{$G_{\mathcal{E}},\VVV$}
          \EndIf
          \For{$U \subseteq V(G)$ with $|U|\leq 2 \sqrt{2|V(G)|}$}
          \For{each cyclic ordering $\rho$ of $U$}
          \For{$N \in $\Call{getNooses}{$G_{\mathcal{E}}$, $\VVV$, $U$, $\rho$,
              $\emptyset$}}\label{line:getnooses}
              \State $G_{\mathsf{in}}\leftarrow \inside{\mathcal{E}}{N}$
              \State $G_{\mathsf{out}}\leftarrow \outside{\mathcal{E}}{N}$
              \If{$|V(G_{\mathsf{in}})\setminus U|\leq 2|V(G)\setminus U|/3$ and
                $|V(G_{\mathsf{in}})\setminus U|\geq|V(G)\setminus U|/3$}
                  \For{$T_\mathsf{in}=(M_\mathsf{in},P_\mathsf{in},D_\mathsf{in}) \in \NCM(\VVV,\rho)$ (enumerate by Obs.~\ref{obs:computeNCM})}
                    \If{$(T_\mathsf{in},N)$ is not sensible (\cref{def:sensibleF})}\label{line:isSensibleF}
                    \State break
                    \EndIf  
                
                    \State $(G^A_{\mathcal{E}^A},\VVV^A) \leftarrow$ \Call{AugmentGraph}{$G_{\mathsf{in}},N, T_\mathsf{in}$} 
                    \State $Z_\mathsf{in} \leftarrow$ \Call{SolveCPLSFC}{$(G^A_{\mathcal{E}^A},\VVV^A)$}

                    \If{$Z_\mathsf{in} \neq \NULL$}
                    \State Use \Cref{lem:propertiesOfAugF} w.r.t.~$T_\mathsf{in},G_{\textsf{in}},N,Z_\mathsf{in}$ to obtain ${\cal S}_\mathsf{in}$.
                    
                      \State $T_\mathsf{out}=(M_\mathsf{out},P_\mathsf{out},D_\mathsf{out}) \leftarrow \MS(U,{\cal S}_\mathsf{in})$
                      \State\label{line:callAug2F} $(G^B_{\mathcal{E}^B},\VVV^B) \leftarrow$ \Call{AugmentGraph}{$G_{\mathsf{out}},N,T_\mathsf{out}$}
                      \State$ Z_\mathsf{out} \leftarrow$
                      \Call{Solve}{$G^B_{\mathcal{E}^B},\VVV^B$}

\If{$Z_\mathsf{out} \neq \NULL$}
                      \State\label{line:fix2F} Use \Cref{lem:propertiesOfAugF} w.r.t.~$T_\mathsf{out},G_{\mathsf{out}},N,Z_\mathsf{out}$ to obtain ${\cal S}_\mathsf{out}$.
                      
                      \State \Return $E({\cal S}_\mathsf{in}\cup {\cal S}_\mathsf{out})\setminus E(G)$
                      \EndIf
                      \EndIf
                  \EndFor
                \EndIf
              \EndFor
            \EndFor
          \EndFor  
          \Return $\NULL$
        \EndProcedure        
        \algrenewcommand\algorithmicindent{1.0em}%
    \end{algorithmic}
\end{algorithm}
\algrenewcommand\algorithmicindent{2.0em}%

The algorithm behind \Cref{thm:solvefixed}, which is illustrated in
\cref{alg:solvefixed},  works similarly as \cref{alg:gammaone} for
\CPLSC{} given in
\cref{se:exactAlgoVar}. Probably the main difference is that we
can now use the given embedding to obtain a noose $N$ as opposed to merely
the pair $(U,\rho)$ with $U \subseteq V(G)$ and $\rho$ being a cyclic
ordering of $U$. This has two main advantages: (1) the noose already
defines the separation between the inside and outside regions, and (2) if
the input graph is connected the number of possible nooses (and therefore the number of
possible partitions of the instance by the solution separator) for a
given pair $(U,\rho)$ can be shown to be at most
$2^{\bigoh(\sqrt{n}\log n)}$~instead~of~$2^{\bigoh(n)}$. While the
former merely allows us to simplify some parts of the algorithm and
its analysis, the latter is the crucial ingredient to obtain the
subexponential run-time in the connected case.

In what follows we will provide a complete exposition of
the algorithm while focusing specifically on the main differences to
\cref{alg:gammaone} for \CPLSC{}.
In the following let $(G_{\cal E},\mathcal{V})$ be an instance of
\CPLSFC{} and let $U \subseteq V(G)$ and $\rho$ be a cyclic order of $U$.
We say that $N$ is \emph{a $(U,\rho)$-noose}
if $N$ is a noose in $\mathcal{E}$ that intersects $\mathcal{E}$ only at the vertices in
$U$ according to the order given by $\rho$. We say that $N$ is \emph{a
  $(U,\rho)$-subnoose} if $N$ can be obtained from a $(U,\rho)$-noose
by deleting sub-curves between consecutive vertices~in~$\rho$.

We start by providing the procedure \Call{getNooses}{$G_{\cal E}$, $\VVV$, $U$, $\rho$,
  $\emptyset$}, which is used in Line~\ref{line:getnooses} of the
  \cref{alg:solvefixed} to enumerate all possible nooses that can
serve as solution separators and represents the main
addition to the algorithm compared to \cref{alg:gammaone}. 
Note that there are two variants of \Call{getNooses}{$G_{\cal E}$, $\VVV$, $U$, $\rho$,
  $\emptyset$} depending on whether or not $G$ is connected.
If $G$ is not connected, then \Call{getNooses}{$G_{\cal E}$, $\VVV$, $U$, $\rho$,
  $\emptyset$} simply branches over all pairs
$(G_{\mathsf{in}},G_{\mathsf{out}})$ of subgraphs of $G$ with
$V(G_{\mathsf{in}})\cap V(G_{\mathsf{out}})=U$ and
$E(G_{\mathsf{in}})\cap E(G_{\mathsf{out}})=\emptyset$ and then checks
whether there is a noose $N$ whose inside is $G_{\mathsf{in}}$ and
whose outside is $G_{\mathsf{out}}$.
Therefore, if $G$ is not
connected, we end up with $2^{\bigoh(n)}$ possible solution separators
and thus do not obtain any run-time improvement over the case of
\CPLS{}.
More interesting is therefore the case that $G$ is
connected, since in this case the number of possible solution
separators reduces to $2^{\bigoh(\sqrt{n}\log n)}$.
Let $(u,v) \in \rho$, i.e.,
$u,v \in U$ such that $v$ is the successor of $u$
in $\rho$. Let $F$ be the set of all
faces in $\cal E$ that contain both $u$ and $v$ on their boundary. Moreover, for a
face $f \in F$ and $x \in \{u,v\}$, let $\tau(\mathcal{E},f,x)$ be the clockwise
sequence of all but the (clockwise) last neighbor of $x$ in $G$ that lie on
$f$'s boundary.
Informally, because $G$ is connected, the possible choices for drawing
a noose between $u$
and $v$ in $\mathcal{E}$ are given by specifying two indices, specifically the index of
$v$ relative to the neighbors of $u$ and the index of $u$ relative to
the neighbors of $v$. Altogether this leaves us with $\sum_{f \in
  F}(|\rho(\mathcal{E},f,v)|)$ choices for $v$ and similarly with $\sum_{f \in
  F}(|\rho(\mathcal{E},f,u)|)$ choices for $u$. Therefore, there are altogether
$\sum_{f \in F}(|\rho(\mathcal{E},f,v)|)(|\rho(\mathcal{E},f,u)|)\leq |V(G)|^2$ choices
for drawing a curve between $u$ and $v$. For $n_u \in \rho(\mathcal{E},f,u)$
and $n_v \in \rho(\mathcal{E},f,v)$, we denote by $\curv(\mathcal{E},u,n_u,v,n_v)$ a curve
drawn in $\mathcal{E}$ between $u$ and $v$ that starts below $n_u$ and ends
below $n_v$ and intersects $\mathcal{E}$ only at $u$ and $v$.
These observations lead to the algorithm illustrated
in \cref{alg:computesep}, which defines the function
\Call{getNooses}{$G_{\mathcal{E}}$, $\VVV$, $U$, $\rho$, $N$} that
given an instance $(G_{\mathcal{E}},\VVV)$ of \CPLSFC{}, where $G$ is connected,
together with a subset $U
\subseteq V(G)$, a cyclic ordering $\rho$ of $U$, and a
$(U,\rho)$-subnoose $N$, outputs all $(U,\rho)$-noose extending $N$. The function is
initially called with $N=\emptyset$ and in this case returns all
$(U,\rho)$-nooses.

\begin{algorithm}[tb!]
    \caption{The algorithm for enumerating all $(U,\rho)$-nooses for
      when the input is connected.}
    \label{alg:computesep}
    \begin{algorithmic}[1]
      \Procedure{getNooses}{$G_{\mathcal{E}}$,$\VVV$, $U$, $\rho$,
        $N$}
      \If{$N$ is a noose}\label{line:checknoose}
        \State \Return $N$
      \EndIf
      \State $\mathcal{P} \leftarrow \emptyset$
      \State $(u,v) \leftarrow $ $(u,v)\in \rho$ and $N$ has no subcurve between $u$ and $v$\label{line:beforeenum}
      \For{faces $f$ of $\mathcal{E}\cup N$ having $u$ and $v$ on its boundary}
        \For{$n_v \in \rho_v(f)$ and $n_u \in \rho_u(f)$}
          \State $N' \leftarrow \curv(\mathcal{E}, u, n_u, v, n_v)$
          \State $\mathcal{P} \leftarrow \mathcal{P} \cup$
          \Call{getNooses}{$G_{\mathcal{E}}$, $\VVV$, $U$,
          $\rho$, $N \cup N'$}\label{line:reccall}
        \EndFor   
      \EndFor   
      \Return $\mathcal{P}$
      \EndProcedure        
    \end{algorithmic}
\end{algorithm}

  \begin{longlemma}
  \label{lem:algnooses}
  \cref{alg:computesep} is correct, runs in time
  $\bigoh(|V(G)|^{2|U|}|V(G)|^{\bigoh(1)})$ and outputs at most
  $|V(G)|^{2|U|}$ $(U,\rho)$-nooses.
\end{longlemma}

\begin{proof}
  The correctness of the algorithm follows because $G$ is connected
  and therefore every component of $G-\{v\}$ for any vertex $v$ of $G$
  is incident to $v$. Concerning the run-time of the algorithm,
  first note that all operations inside one recursive call of the
  algorithm can be achieved in polynomial-time
  ($\bigoh(|V(G)|^{\bigoh(1)}$).
  Moreover, as stated above there are at most $|V(G)|^2$ choices to draw a curve
  between $u$ and $v$ for every $(u,v) \in \rho$ and therefore for
  every $(u,v) \in \rho$, the algorithm calls itself recursively at
  most $|V(G)|^2$ times (in Line~\ref{line:reccall}). Since there are
  at most $|U|$ many pairs $(u,v) \in \rho$, the total number of
  recursive calls of the algorithm is therefore at most
  $(|V(G)|^2)^{|U|}=|V(G)|^{2|U|}$, which shows the stated run-time.
\end{proof}

We are now ready to show the correctness of
\cref{alg:solvefixed}. Since the proof is rather similar to
the proof for \cref{alg:gammaone} from
\cref{se:exactAlgoVar}, we will focus on outlining the main
differences in the proof. We start with some additional notation for nooses.
For a noose $N$ in $\mathcal{E}$ we denote by $\inside{\mathcal{E}}{N}$ and
$\outside{\mathcal{E}}{N}$ the subgraph of $G$ that is inside and outside of $N$ in $\mathcal{E}$,
respectively; here we define inside as the region that does not
contain the outer-face of $\mathcal{E}$. Similarly, we denote by $\insideR{\mathcal{E}}{N}$ and
$\outsideR{\mathcal{E}}{N}$, the embedding $\mathcal{E}$ restricted to the inside and outside of $N$,
respectively.

The first definition that we need to adapt to the fixed embedding case
is \cref{def:compatiblePartSol} for compatibility between a
partial solution and a pair $(T_{\mathsf{in}},N)$ for $T_\mathsf{in} \in
\NCM(\VVV,\rho)$ and a $(U,\rho)$-noose $N$.

\begin{longdefinition}[{\bf Partial Solution IN/OUT-Compatible with $(T_\mathsf{in},N)$}]\label{def:compatiblePartSolF}
  Consider $(G_{\mathcal{E}},\VVV)$, a cyclic ordering $\rho$ of some $U\subseteq
  V(G)$,
  $T_\mathsf{in}=(M_{\mathsf{in}},P_{\mathsf{in}},D_{\mathsf{in}}) \in
  \NCM(\VVV,\rho)$, and a $(U,\rho)$-noose $N$.
  Then, a partial solution $\cal S$ is {\em IN-compatible} ({\em
    OUT-compatible}) with $(T_{\mathsf{in}},N)$ if:
  \begin{enumerate}
    \item\label{def:compatiblePartSol1F} $\MS(U,{\cal S})=T_\mathsf{out}$ and $T_\mathsf{in}$ are complementary.
    \item\label{def:compatiblePartSol2F} Let
      $G'=\inside{\mathcal{E}}{N}-D_{\mathsf{in}}$
      ($G'=\outside{\mathcal{E}}{N}-D_{\mathsf{out}}$). Then, $\cal S$ contains
      all and only the vertices in $G'$, and all (but not
      only) edges in $G'$ between vertices in the
      same cluster.
    \item\label{def:compatiblePartSol3F} There exists a planar drawing
      of $G'\cup E({\cal S})$ that extends
      $\insideR{\mathcal{E}}{N}$ ($\outsideR{\mathcal{E}}{N}$) with an inner-face whose boundary contains $U$
      (with, possibly, other vertices) ordered as by $\rho$.
    \item\label{def:compatiblePartSol4F} Each path in ${\cal S}$ satisfies one of the following conditions:
        \begin{enumerate}
        \item It consists of all vertices of a cluster in $\VVV$ that belong to $G'$, and has no endpoint in~$U$.
        \item It has an endpoint in $U$.
        \end{enumerate}
  \end{enumerate}
\end{longdefinition}

We also need the following slightly adapted version of
\cref{def:sensible} for sensibility. The main difference to
the its counterpart for the variable embedding case is that we no
longer require \cref{def:sensibleItem1} since this is already
guaranteed by the existence of the $(U,\rho)$-noose $N$.

\begin{longdefinition}[{\bf Sensibility of $(T_\mathsf{in},N)$}]
  \label{def:sensibleF}
  Consider $(G_{\mathcal{E}},\VVV)$, a cyclic ordering $\rho$ of some $U\subseteq
  V(G)$,
  $T_\mathsf{in}=(M_{\mathsf{in}},P_{\mathsf{in}},D_{\mathsf{in}}) \in
  \NCM(\VVV,\rho)$,
  a $(U,\rho)$-noose $N$, and $G_{\mathsf{in}}=\inside{\mathcal{E}}{N}$. Then, $(T_\mathsf{in},N)$ is {\em sensible} if:
  \begin{enumerate}
    \item\label{def:sensibleItem2F}  No vertex in $D_\mathsf{in}$ is adjacent in $G_\mathsf{in}$ to a vertex that belongs to the same cluster in $\VVV$. Additionally, no vertex in $U\setminus (V(M_\mathsf{in})\cup P_\mathsf{in}\cup D_\mathsf{in})$ is adjacent in $G_\mathsf{out}=G[O\cup U]-E(G_\mathsf{in})$ for $O=V(G)\setminus (V(G_{\mathsf{in}})$ to a vertex that belongs to the same cluster in $\VVV$.
      \item\label{def:sensibleItem3F}  No cluster in $\VVV$ has non-empty intersection with both $V(G_{\mathsf{in}})\setminus D_\mathsf{in})$ and $O\cup D_\mathsf{in}$ but not with $V(M_\mathsf{in})\cup P_\mathsf{in}$.
  \end{enumerate}
\end{longdefinition}

We now show how using compatible solutions in $\NCM$ yields solutions to \CPLSFC.
We now need the following slightly adapted and mostly simplified
version of \Cref{lem:unionYieldsSolution}, whose proof is basically
identically to the proof of \Cref{lem:unionYieldsSolution} and is
provided here merely for completeness.
\begin{longlemma}\label{lem:unionYieldsSolutionF}
  Consider $(G_{\mathcal{E}},\VVV)$, a cyclic ordering $\rho$ of some $U\subseteq
  V(G)$,
  $T_\mathsf{in}=(M_{\mathsf{in}},P_{\mathsf{in}},D_{\mathsf{in}}) \in
  \NCM(\VVV,\rho)$, and a $(U,\rho)$-noose $N$. Suppose that $(T_\mathsf{in},N)$ is sensible. Additionally, consider
  \begin{itemize}
  \item a partial solution ${\cal S}_\mathsf{in}$ IN-compatible with $(T_{\mathsf{in}},N)$, and
  \item a partial solution ${\cal S}_\mathsf{out}$ OUT-compatible with $\MS(U,{\cal S}_\mathsf{in})=(T_{\mathsf{out}},N)$.
  \end{itemize}
  Then, $Z=E({\cal S}_\mathsf{in}\cup{\cal S}_\mathsf{out})\setminus
  E(G)$ is a solution to $(G_{\mathcal{E}},\VVV)$ as an instance of \CPLSFC{}.
\end{longlemma}
 \begin{proof}

\medskip
\noindent{\bf Verifying Planarity.} We first argue that $G\cup Z$ is a
planar graph. Since ${\cal S}_\mathsf{in}$ is IN-compatible with
$(T_{\mathsf{in}},N)$,
\cref{def:compatiblePartSol3F} in
\cref{def:compatiblePartSolF} implies that there exists a
planar drawing $\varphi_{\mathsf{in}}$ of $G_{\mathsf{in}}\cup E({\cal S}_{\mathsf{in}})$ that extends
$\insideR{\mathcal{E}}{N}$ ($\outside{\mathcal{E}}{N}$) with an inner-face whose boundary contains $U$
(with, possibly, other vertices) ordered as by $\rho$, where $G_{\mathsf{in}}=\inside{\mathcal{E}}{N}-D_{\mathsf{in}}$.

Similarly, since ${\cal S}_\mathsf{out}$ is OUT-compatible with
$(T_{\mathsf{out}},N)$, \cref{def:compatiblePartSol3F}
implies there exists a planar drawing $\varphi_\mathsf{out}$
of $G_{\mathsf{out}}\cup E({\cal S}_{\mathsf{out}})$ that extends
$\insideR{\mathcal{E}}{N}$ ($\outside{\mathcal{E}}{N}$) with an inner-face whose boundary contains $U$
(with, possibly, other vertices) ordered as by $\rho$, where $G_{\mathsf{out}}=\outside{\mathcal{E}}{N}-D_{\mathsf{out}}$.
Clearly, we can flip $\varphi_\mathsf{out}$ to obtain another planar
drawing $\varphi'_\mathsf{out}$ of $G_\mathsf{out}\cup E({\cal
  S}_\mathsf{out})$ whose outer-face's boundary contains $U$ (with,
possibly, other vertices) ordered as by $\rho$. Since $N$ is a noose
of $\mathcal{E}$, we know that the (non-disjoint) union of $\inside{\mathcal{E}}{N}$ and
$\outside{\mathcal{E}}{N}$ yields $G$.
Thus, by taking the (non-disjoint) union of the drawings
$\varphi_\mathsf{in}$ and $\varphi'_\mathsf{out}$, we obtain a planar
drawing of $G\cup Z$. So, $G\cup Z$ is a planar graph.

\medskip
\noindent{\bf Verifying That All Edges Are Covered.} Consider some
cluster $V_i\in \VVV$, and let $Z_i$ denote the subset of edges in $Z$
with both endpoints in $V_i$. We need to prove that $G[V_i]\cup Z_i$
is a path. Let $Z_{i,\mathsf{in}}=Z_i\cap E(S_\mathsf{in})$ and
$Z_{i,\mathsf{out}}=Z_i\cap E(S_\mathsf{out})$. Let ${\cal
  S}_{i,\mathsf{in}}$ (resp., ${\cal S}_{i,\mathsf{out}}$) denote the
collection of paths in ${\cal S}_\mathsf{in}$ (resp., ${\cal
  S}_\mathsf{out}$) between vertices in $V_i$.
Since ${\cal S}_\mathsf{in}$ is IN-compatible with
$(T_{\mathsf{in}},N)$, \cref{def:compatiblePartSol2F} in
\cref{def:compatiblePartSolF} implies that
$G_\mathsf{in}[V_i]\cup Z_{i,\mathsf{in}}$, where
$G_\mathsf{in}=G'_\mathsf{in}-D_\mathsf{in}$, is exactly the
collection of paths ${\cal S}_{i,\mathsf{in}}$. Similarly, since
${\cal S}_\mathsf{out}$ is OUT-compatible with $(T_{\mathsf{out}},N)$,
\cref{def:compatiblePartSol2F} implies that
$G_\mathsf{out}[V_i]\cup Z_{i,\mathsf{out}}$, where
$G_\mathsf{out}=G'_\mathsf{out}-D_\mathsf{out}$, is exactly the
collection of paths~${\cal S}_{i,\mathsf{out}}$.

Recall that we have already argued that the (non-disjoint) union of
$G_\mathsf{in}$ and $G_\mathsf{out}$ equals~$G$.

Next, we further argue that the (non-disjoint) union of
$G_\mathsf{in}[V_i]$ and $G_\mathsf{out}[V_i]$ equals $G[V_i]$. To
this end, it suffices to show that:  (I) $D_\mathsf{in}\cap
D_{\mathsf{out}}=\emptyset$ (and hence $D_\mathsf{in}\cap V_i\subseteq
V(G_\mathsf{out})$ and $D_\mathsf{out}\cap V_i\subseteq
V(G_\mathsf{in})$), and (II) there does not exist an edge between a
vertex in $D_\mathsf{in}\cap V_i$ (resp., $D_\mathsf{out}\cap V_i$)
and a vertex in $V(G'_\mathsf{in})\cap V_i$ (resp.,
$V(G'_\mathsf{out})\cap V_i$). Because ${\cal S}_\mathsf{in}$ is
IN-compatible with $(T_{\mathsf{in}},N)$, we have that
$T_\mathsf{out}$ and $T_\mathsf{in}$ are complementary, and hence, by
\cref{def:complement1} in \cref{def:complement}, we
get that $D_\mathsf{in}\cap D_\mathsf{out}=\emptyset$, therefore Item
(I) holds. Because $(T_\mathsf{in},N)$ is sensible, Item
(II) directly follows from \cref{def:sensibleItem2F} in
\cref{def:sensibleF}.

Overall, we conclude, so far, that $G[V_i]\cup Z_i$ is exactly the
(non-disjoint) union of the  two collections of paths ${\cal
  S}_{i,\mathsf{out}}$ and ${\cal S}_{i,\mathsf{in}}$. Next, we verify
that this union yields a single path.

\medskip
\noindent{\bf Verifying That the Union Yields a Single Path.} Recall
that $\MS(U,{\cal
  S}_\mathsf{in})=T_\mathsf{out}=(M_\mathsf{out},P_\mathsf{out},D_\mathsf{out})$,
and let $\MS(U,{\cal
  S}_\mathsf{out})=T'_\mathsf{in}=(M'_\mathsf{in},P'_\mathsf{in},D'_\mathsf{in})$. Note
that, possibly, $T_\mathsf{in}\neq T'_\mathsf{in}$. By the definition
of $\MS$:
\begin{enumerate}
\item\label{condition:MS1F} $M_\mathsf{out}$ (resp., $M'_\mathsf{in}$)
  has an edge between the endpoints of every path in ${\cal
    S}_{\mathsf{in}}$ (resp., ${\cal S}_{\mathsf{out}}$) that has both
  endpoints in $U$,
\item\label{condition:MS2F} $P_\mathsf{out}\subseteq U$ (resp.,
  $P'_\mathsf{in}\subseteq U$) consists of the vertices of degree $1$
  in ${\cal S}_\mathsf{out}$ (resp., ${\cal S}_\mathsf{in}$) belonging
  to $U$ that are not in $V(M_\mathsf{out})$ (resp.,
  $V(M'_\mathsf{in})$), and
\item\label{condition:MS3F} $D_\mathsf{out}\subseteq U$ (resp.,
  $D'_\mathsf{in}\subseteq U$) consists of the vertices of degree $2$
  in ${\cal S}_{\mathsf{in}}$ (resp., ${\cal S}_{\mathsf{out}}$) belonging to $U$.
\end{enumerate}

Further, notice that $S_{\mathsf{in}}$ does not contain vertices in $D_\mathsf{in}$, and ${\cal S}_\mathsf{out}$ does not contain vertices in $D_\mathsf{out}$. In particular, the latter implies that $D_\mathsf{out}\subseteq U\setminus (V(M_\mathsf{in}')\cup P_\mathsf{in}'\cup D_\mathsf{in}')$. 
Since ${\cal S}_\mathsf{out}$ is compatible with $(T_{\mathsf{out}},I,E_{\mathsf{out}})$, by \cref{def:compatiblePartSol1F} in \cref{def:compatiblePartSolF},  $T_\mathsf{in}'$ and $T_\mathsf{out}$ are complementary. 
 So, by \cref{def:complement1} in \cref{def:complement}, $D_\mathsf{in}'\subseteq U\setminus(V(M_\mathsf{out})\cup P_\mathsf{out}\cup D_\mathsf{out})$. Thus, all vertices in $D_\mathsf{out}\cup D_\mathsf{in}'$ are of degree $2$ in ${\cal S}_{i,\mathsf{out}}\cup {\cal S}_{i,\mathsf{in}}$.

Due to \cref{def:sensibleItem3F} in
\cref{def:sensibleF} and
\cref{def:compatiblePartSol4F} in
\cref{def:compatiblePartSolF}, we know that if $V_i$
contains vertices only from $G_\mathsf{in}$ or only from
$G_\mathsf{out}$, then ${\cal S}_{i,\mathsf{out}}\cup {\cal
  S}_{i,\mathsf{in}}$ is indeed a path (more precisely, one of these
two sets is a path and the other is empty), and that, otherwise, every
path in ${\cal S}_{i,\mathsf{in}}$ or in ${\cal S}_{i,\mathsf{out}}$
has an endpoint in $U\setminus (D_\mathsf{in}\cup D_\mathsf{out})$. In
this latter case, by the conditions implied by the definition $\MS$,
which are specified above, and by \cref{def:complement2} in
\cref{def:complement}, we again derive that ${\cal
  S}_{i,\mathsf{in}}\cup {\cal S}_{i,\mathsf{out}}$ is a path. \mbox{This
completes the proof.}
\end{proof} 

The following definition of the augmented graph is a simpler version
of the original definition for the case of the variable embedding. The
main difference is that we do no longer need to take into account
marked vertices and that because of the fixed drawing we use a
slightly weaker notion of ``inner triangulation''.
\begin{longdefinition}[{\bf Augmented Graph}]\label{def:augmentedF}
  Let $\mathcal{C}=(G_{\mathcal{E}},\mathcal{V})$ be an instance of \CPLSFC{}. Let
  $\rho$ be a cyclic ordering of some subset $U$ of $V(G)$ and let $N$
  be a $(U,\rho)$-noose such that $\inside{\mathcal{E}}{N}=G$ (or $\outside{\mathcal{E}}{N}=G$).
  Let $(M,P,D) \in \NCM(\VVV,\rho)$. 
  Then, the {\em augmented graph}
  \Call{AugmentGraph}{$G_{\mathcal{E}},\VVV,N,(M,P,D)$} and a embedding
  of it, are obtained from
  $\mathcal{C}=(G_{\mathcal{E}},\mathcal{V})$ after doing the following:
  \begin{enumerate}
  \item for every $(u,v) \in \rho$ subdivide the subcurve of $N$ between
    $u$ and $v$ with a new vertex $n_{u,v}$ taking its own new cluster and add
    the edges $\{u,n\}$ and $\{n,v\}$ to $G$; the drawing of the new
    edges is given by the subcurve of $N$ between $u$ and $v$.
  \item For every $e \in M$, add a new vertex adjacent to the endpoints
    of $e$ and belonging to the same cluster and draw the two edges outside (inside) of $N$.
  \item For every $p \in P$ add a new vertex $n$, drawn outside (inside) close to
    $p$, and make it adjacent to $p$. The new edge between $n$ and $p$
    is drawn outside (inside) of $N$.
  \item  Remove every vertex in $D$ from its cluster, and create a new one-vertex cluster for it.
  \item for every face $f$ outside (inside) of $N$ add a new vertex $n$
    into $f$ and make it adjacent to all vertices on the boundary of
    $f$. This, together with the subdivision vertices introduced earlier, ensures that no edges between
    vertices in the same cluster can be drawn~outside~of~$N$.
  \end{enumerate}
\end{longdefinition}

Similarly, to \Cref{obs:restrictFaces}, the following observation allows us to assume that all edges that are
part of a solution for the augmented graph can be drawn either only outside
or only inside of the noose $N$.
\begin{longobservation}\label{obs:augrestF}\sloppypar
  Let $\mathcal{C}=(G_{\mathcal{E}},\mathcal{V})$ be an instance of \CPLSFC{}. Let
  $\rho$ be a cyclic ordering of some subset $U$ of $V(G)$ and let $N$
  be a $(U,\rho)$-noose such that $\inside{\mathcal{E}}{N}=G$ (or $\outside{\mathcal{E}}{N}=G$).
  Let $T \in \NCM(\VVV,\rho)$ and
  $(G^A_{\mathcal{E}^A},\VVV^A)=$\Call{AugmentGraph}{$G_\mathcal{E},\VVV,\rho,N,T$}. Then,
  every two vertices outside (inside) $N$ that are in the same cluster
  are in different faces in $\mathcal{E}^A$. 
\end{longobservation}

We are ready to obtain the following analog of \Cref{lem:propertiesOfAugF}
\begin{longlemma}\label{lem:propertiesOfAugF}
  Consider $(G_{\mathcal{E}},\VVV)$, a cyclic ordering $\rho$ of some $U\subseteq
  V(G)$, $T_\mathsf{in}=(M_{\mathsf{in}},P_{\mathsf{in}},D_{\mathsf{in}}) \in
  \NCM(\VVV,\rho)$, and a $(U,\rho)$-noose $N$.
  Let $Z_\mathsf{in}\neq$NULL be a solution to the instance
  $(G^A_{\mathcal{E}^A},\VVV^A)=$\Call{AugmentGraph}{$\inside{\mathcal{E}}{N},N,T_\mathsf{in}$}
  of \CPLSFC{}. Then, in polynomial time we can obtain a partial solution
  ${\cal S}_{\mathsf{in}}$ that is IN-compatible with
  $(T_{\mathsf{in}},N)$.
\end{longlemma}

 \begin{proof}
  Let ${\cal S'}_\mathsf{in}$ denote the subgraph of $G^A\cup
  Z_\mathsf{in}$ on all vertices, and on the edges in $G^A$ between
  vertices in the same cluster (w.r.t.~$\VVV^A$) and the edges in
  $Z_\mathsf{in}$. Clearly, ${\cal S'}_\mathsf{in}$ is a partial
  solution. Let $\varphi$ be the planar
  drawing of $G^A\cup Z_{\mathsf{in}}$ that extends $\mathcal{E}^A$ and exists
  because $Z_{\mathsf{in}}$ is a solution for $(G^A_{\mathcal{E}^A},\VVV^A)$. Note
  that because of \cref{obs:augrestF} it holds that all edges of
  $Z_{\mathsf{in}}$ are drawn inside of $N$ in $\varphi$.
  Let ${\cal S}_\mathsf{in}$ denote
  the
  partial solution obtained from ${\cal S'}_\mathsf{in}$ by: (i)
  removing, from each of its paths having one or two endpoints not in
  $G$, these one or two endpoints (which must be vertices adjacent to
  vertices in $P_\mathsf{in}$ belonging to the same cluster as them);
  (ii) removing every path consisting of a single vertex on its own new
  singleton cluster;
  (iii) removing, from all paths, all edges belonging to
  $M_\mathsf{in}$, which might split some paths into several paths. No
  additional modifications will be made.  So, in what follows, we verify
  that ${\cal S}_\mathsf{in}$ is IN-compatible with
  $(T_\mathsf{in},N)$, which will complete the proof.
  We do so by verifying each of the conditions in
  \cref{def:compatiblePartSolF}.

\smallskip
\noindent{\bf\em \cref{def:compatiblePartSol1F}.} We need to
show that $\MS(U,{\cal
  S}_\mathsf{in})=T_\mathsf{out}=(M_\mathsf{out},P_\mathsf{out},D_\mathsf{out})$
and $T_\mathsf{in}$ are complementary. To this end, we need to verify
the satisfaction of the two conditions in
\cref{def:complement}. By the construction of the augmented
graph, all vertices in $D_\mathsf{in}$ belong to their own clusters,
and so ${\cal S}_\mathsf{in}$ cannot include newly added edges incident to them,
implying that $D_\mathsf{in}\subseteq U\setminus
(V(M_\mathsf{out})\cup P_\mathsf{out}\cup
D_\mathsf{out})$. Furthermore, by its construction, every vertex in
$V(M_\mathsf{in})\cup P_\mathsf{in}$ is incident to one edge that is
removed from ${\cal S'}_\mathsf{in}$ when constructing ${\cal
  S}_\mathsf{in}$, which implies that $D_\mathsf{out}\subseteq
U\setminus (V(M_\mathsf{in})\cup P_\mathsf{in}\cup D_\mathsf{in})$.

Thus, \cref{def:complement1} in
\cref{def:complement} is satisfied. For the second
condition, notice that in ${\cal S}_\mathsf{in}$, every
cluster in $\VVV_A$ corresponds to a single path. Now, contract the
edges in ${\cal S}_\mathsf{in}$ belonging to
$\inside{\mathcal{E}}{N}$ in the augmented graph, and keep only edges
connecting vertices belonging to the same cluster---then, we obtain
the graph $G_{T_\mathsf{in},T_\mathsf{out}}$, excluding, possibly, for
every vertex in $P_\mathsf{out}$, the newly added edge incident to it
when constructing $G_{T_\mathsf{in},T_\mathsf{out}}$.

So, the graph we created is a collection of paths as required in
\cref{def:complement2} in \cref{def:complement},
and it stays so also after adding the ``missing'' edges (if any) since
by the definition of $P_\mathsf{out}$, the vertices within it are
endpoints of paths in ${\cal S'}_\mathsf{in}$. Thus, by the
definition of ${\cal S}_\mathsf{in}$ based on ${\cal
  S'}_\mathsf{in}$, we conclude that \cref{def:complement2} in
\cref{def:complement} is satisfied.

\smallskip
\noindent{\bf\em \cref{def:compatiblePartSol2F}.} Because the
augmented graph contains $\inside{\mathcal{E}}{N}$, and when creating ${\cal
  S}_\mathsf{in}$, we removed all vertices and edges that do not
belong to this graph, we derive the satisfaction of this condition.

\smallskip
\noindent{\bf\em \cref{def:compatiblePartSol3F}.} The drawing
$\varphi$ certifies the satisfaction of this condition.

\smallskip
\noindent{\bf\em \cref{def:compatiblePartSol4F}.} Recall that
every path in ${\cal S'}_\mathsf{in}$ consists of all and
only the vertices of a single cluster in $\VVV_A$. Since all of the
edges that are removed when we create ${\cal S}_\mathsf{in}$ from
${\cal S'}_\mathsf{in}$ are incident to vertices in $U$, we
derive the satisfaction of this condition.
\end{proof} 

\subparagraph{The Algorithm and Its Correctness.}

We present the pseudocode of the algorithm as
\cref{alg:solvefixed}. Here, \Call{BruteForceFixed}{$G_{\mathcal{E}},\VVV$} is
a procedure that iterates over all edge subsets and extensions of $\mathcal{E}$
by those edges, and for each one of
them, checks whether it is a solution; if at
least one of them is, then it returns it, and otherwise it returns
NULL.

For the proof of correctness, we will use
\Cref{lem:unionYieldsSolutionF,lem:propertiesOfAugF}
and~\Cref{prop:planarSeparator}.
We are now ready to conclude the correctness of the algorithm.

\begin{longlemma}\label{lem:algoCorrectF}
  \cref{alg:solvefixed} solves the \CPLSFC{} problem.
\end{longlemma}
 \begin{proof}
    The proof is by induction on $|V(G)|$. When $|V(G)|\leq 100$,
    correctness is trivial since the algorithm performs brute-force. So,
    next suppose that $|V(G)|>100$ and that the algorithm is correct
    whenever it is called with a graph having fewer vertices than
    $|V(G)|$.

    \medskip
    \noindent{\bf Reverse Direction.} For the proof of this direction,
    suppose that the algorithm returns $Z=E({\cal S}_\mathsf{in}\cup{\cal
      S}_\mathsf{out})\setminus E(G)$. Then, we need to prove that $Z$ is
    a solution. Let
    $U,\rho,N,T_\mathsf{in},Z_\mathsf{in},T_\mathsf{out},Z_\mathsf{out}$
    be the entities corresponding to the iteration when $Z$ is returned.
    Due to Line~\ref{line:isSensibleF} in \cref{alg:solvefixed},
    $(T_\mathsf{in},N)$ is sensible. Moreover, due to
    \cref{lem:propertiesOfAugF}, we have that ${\cal S}_\mathsf{in}$
    is a partial solution IN-compatible with
    $(T_{\mathsf{in}},N)$,  and
    ${\cal S}_\mathsf{out}$ is a partial solution OUT-compatible with
    $(T_{\mathsf{out}},N)$. So, due to \cref{lem:unionYieldsSolutionF}, the proof of this
    direction is complete.

    \medskip
    \noindent{\bf Forward Direction.} For the proof of this direction,
    suppose a solution $Z^\star$ exists. We need to prove that the
    algorithm does not return NULL (and hence, by the correctness of the
    reverse direction, this would mean that it returns a solution). Let
    $G^\star$ denote the graph obtained from $G$ by adding the edges in
    $Z^\star$. Let ${\cal S}^\star$ denote the paths corresponding to the
    different clusters in $G^\star$. Let $\varphi$ be some planar drawing
    of $G^\star$ that extends $\mathcal{E}$. By \cref{prop:planarSeparator}, there exists
    a noose $N$ such that:
    \begin{itemize}
    \item $N$ intersects $\varphi$ only at the vertices in $U$ and $|U|\leq 2\sqrt{2|V(G)|}$.
    \item The set $I=V(\inside{\varphi}{G})\setminus U$ of vertices
      satisfies $|I|\leq 2|V(G)\setminus U|/3$ and $|I|\geq |V(G)\setminus U|/3$.
    \end{itemize}
    Let $\rho$ denote the cyclic ordering in which $N$ traverses $U$.
    
    Let $\{Z^\star_\mathsf{in},Z^\star_\mathsf{out}\}$ denote the
    partition of $Z^\star$ according to which edges are drawn on the
    inside and on the outside of $N$. Let ${\cal S}^Z_\mathsf{out}$
    denote the partial solution restricted to $\inside{\varphi}{G^\star}$. Since $G^\star$ is
    planar, ${\cal S}^Z_\mathsf{out}$ is compatible with $\rho$. Let
    $T_\mathsf{in}=\MS(U,{\cal S}^Z_\mathsf{out})$. By
    \cref{obs:inNCM}, $T_\mathsf{in}\in\NCM(\VVV,\rho)$. In
    what follows, consider the iteration of
    \cref{alg:solvefixed} corresponding to these
    $U,\rho,N$ and $T_\mathsf{in}$.
        
    First, we claim that $(T_\mathsf{in},N)$ is sensible. Here,
    \cref{def:sensibleItem3}
    follows from the choice of $N$, ensuring that $U$ separates
    $I=V(\inside{\varphi}{G})\setminus U$
    from $O=V(G)\setminus (U\cup I)$ in $G^\star$ (and, so, also in
    particular in $G$). Also, from this, and since in ${\cal
      S}^\star$, every vertex in $D_\mathsf{in}$ (resp., $U\setminus
    V(M_\mathsf{in}\cup P_\mathsf{in}\cup D_\mathsf{in})$) only has
    neighbors in $\inside{\mathcal{E}}{N}$ (resp., $\outside{\mathcal{E}}{N}$) while
    ${\cal S}^\star$ includes all edges between vertices in the same
    cluster, we get that \cref{def:sensibleItem2} is
    satisfied as well.
        
    Second, we claim that $Z_\mathsf{in}$ computed in the iteration
    under consideration must not be NULL. To this end, it suffices to
    note that  $Z^\star_\mathsf{in}$ is a solution to the instance
    $(G^A_{\mathcal{E}^A},\VVV^A)$ of \CPLSFC{} constructed by the
    algorithm. Similarly, $Z_\mathsf{out}$ computed in the iteration
    under consideration must not be NULL, since
    $Z^\star_\mathsf{out}=Z^\star\setminus Z^\star_\mathsf{in}$ is a
    solution to the instance $(G^B_{\mathcal{E}^B},\VVV^B)$ of \CPLSFC{}
    constructed by the algorithm. Hence, in the iteration under
    consideration, the algorithm will return a set of edges, which, in
    particular, means that it does not return NULL.
\end{proof} 

\subparagraph{Running Time Analysis.}

\begin{longlemma}\label{lem:runtimeF}
    \cref{alg:solvefixed} runs in time $2^{\bigoh(n)}$ for
    arbitrary instances. Moreover, \cref{alg:solvefixed} runs
    in time $2^{\bigoh(\sqrt{n}\log n)}$ if $G$ is connected.
\end{longlemma}

 \begin{proof}
    Since the analysis for the general case is literally the same as in
    variable embedding case, i.e., in the proof of \Cref{lem:runtime}, we provide
    here only the proof for the case that $G$ is connected.
    
    Let $T(n)$ denote the running time of the algorithm when called with a
    graph on $n$ vertices. When $n\leq 100$, this is bounded by a
    constant.

    When $n>100$, the algorithm iterates over $\binom{n}{\leq 2
      \sqrt{2n}}=2^{\bigoh(\sqrt{n}\log n)}$ sets $U$. For each set $U$,
    the algorithm iterates over $|U|!=(2\sqrt{2n})!=2^{\bigoh(\sqrt{n}\log
      n)}$ orderings $\rho$. Because of \Cref{lem:algnooses}, for each $\rho$ the algorithm iterates
    over at most $n^{\bigoh(\sqrt{n})}$ $(U,\rho)$-nooses $N$.
    For
    each noose $N$, \cref{obs:computeNCM} implies that the
    algorithm iterates over $2^{\bigoh(\sqrt{n})}$ triples
    $T_\mathsf{in}$. So, in total, the algorithm performs
    $2^{\bigoh(\sqrt{n}\log n)}$ iterations. For each triple, the
    algorithm performs two recursive calls on graphs with at most
    $2n/3 + \bigoh(\sqrt{n})$ vertices each, and other internal
    calculations in polynomial time. So,
    $$T(n) = 2^{\bigoh(\sqrt{n}\log n)}\cdot (2\cdot
    T(\frac{2n}{3}+\bigoh(\sqrt{n})) + n^{\bigoh(1)}) =
    2^{\sqrt{n}\log n(\frac{2^0}{3^0}+\frac{2^1}{3^1}+\frac{2^2}{3^2}+\frac{2^3}{3^3}+\cdots)}\cdot
    2^{\bigoh(\sqrt{n}\log n)},$$

    which evaluates to $2^{\bigoh(\sqrt{n}\log n)}$.
\end{proof} 

From \cref{lem:algoCorrectF,lem:runtimeF}, we conclude the correctness of \cref{thm:solvefixed}.

\section{The Kernels}\label{se:fptAlgo}
In this section we provide kernelization algorithms for \CPLS\ and \CPLSF\ parameterized by the vertex cover number
$k$ of the input graph $G$. Assume that~$X$ is a vertex cover of $G$ of size $k$; we will deal with computing a suitable vertex cover in the proofs of the main theorems of this section.
As our first step, we construct the set $Z$ consisting of the union of $X$ with all vertices of degree at least $3$ in $G$. Since $G$ can be assumed to be planar, we have:

\begin{lemma}[{\cite[Lemma 13.3]{fomin2019kernelization}}]
\label{lem:linearplanarN}
$|Z|\leq 3k$.
\end{lemma}

Note that each vertex in~$V(G)\setminus Z$ now has $0$, $1$ or $2$ neighbors in $Z$. For each subset~$Q\subseteq Z$ of size at most $2$, let the \emph{neighborhood type} $T_Q$ consist of all vertices in $V(G)\setminus Z$ whose neighborhood in $Z$ is precisely $Q$. Moreover, for $i\in \{0,1,2\}$ we let $T_i=\bigcup_{Q\subseteq Z, |Q|=i}T_Q$ contain all vertices outside of $Z$ with degree $i$.
At this point, our approach for dealing with \CPLS\ and \CPLSF\ will diverge.

\subsection{The Fixed-Embedding Case}
\label{sub:fixembvc}

We will begin by obtaining a handle on vertices with precisely two neighbors in $Z$. However, to do so we first need some specialized terminology. To make our arguments easier to present, we assume w.l.o.g.\ that the input instance $\mathcal{I}$ is equipped with a drawing $D$ of $G$ that respects the given embedding.

Let $G_{2}$ be the subgraph of $G$ induced on $Z\cup T_2$, where $T_2$ is the set of all vertices in $V(G)\setminus Z$ with precisely two neighbors in $Z$. Let~$D_2$ be the restriction of $D$ to $G_2$, and observe that~$D_2$ only differs from $D$ by omitting some pendant and isolated vertices. Let a face in~$D_2$ be \emph{special} if it is incident to more than $2$ vertices of $Z$, and \emph{clean} otherwise; notice that the boundary of a clean face must be a $C_4$ which has precisely two vertices of $Z$ that lie on opposite sides of the $C_4$ and which contains only vertices in $T_0\cup T_1$ in its interior.

\begin{lemma}
\label{lem:specialf}
The number of special faces in $D_2$ is upper-bounded by $9k$, and the number of vertices in $T_2$ that are incident to at least one special face is upper-bounded by $36k$.
\end{lemma}

 \begin{proof}
The hypothetical graph $G'$ obtained by adding a vertex into each special face and making this vertex adjacent to each vertex in $Z$ on the boundary of its face is planar. Since $|Z|\leq 3k$ as per \cref{lem:linearplanarN}, we can repeat the same argument on $G'$ as the one used for \cref{lem:linearplanarN} to obtain that the number of special faces must be upper-bounded by $9k$. 

For the second part of the claim, we observe that the graph $G'$ has at most $12k$ vertices and each vertex in $T_2$ that is incident to a special face can be represented as a unique edge added into $G'$ while preserving planarity. The claim then followed by the well-known fact that a $12k$-vertex planar graph cannot have more than $36k$ edges.
\end{proof} 

For a pair $\{a,b\}\subseteq Z$, we say that a set $P$ of clean faces is an \emph{($ab$)-brick} if (1) the only vertices of $Z$ they are incident to are $a$ and $b$, and (2) $P$ forms a connected region in $D_2$, and (3) $P$ is maximal with the above properties. Since the boundaries of every clean face in~$P$ consists of $a$, $b$, and two degree-$2$ vertices, this in particular implies that the clean faces in $P$ form a sequence where each pair of consecutive faces shares a single degree-$2$ vertex; this sequence may either be cyclical (in the case of $a$ and $b$ not being incident to any special faces; we call such bricks \emph{degenerate}), or the first and last clean face in $P$ are adjacent to special faces. Observe that a degenerate brick may only occur if $|Z|=2$.

\begin{observation}
\label{obs:blocks}
The total number of bricks in $D_2$ is upper-bounded by $24k$.
\end{observation}

 \begin{proof}
The claim clearly holds if $D_2$ contains a degenerate brick. Otherwise, 
consider the graph $G''$ obtained by taking the subgraph of $G$ induced on $Z$, and adding to it a vertex for each special face and making it adjacent to each of that face's incident vertices in $Z$. Observe that at this point, $|V(G'')|\leq 12k$. Now for each $ab$-brick, we add a vertex and make it adjacent to $a$, $b$, and the two vertices representing the special faces that the brick is adjacent to in $D_2$. Since each of the newly created vertices has degree at least $3$ (in fact, they have $4$), we can once again invoke the argument of \cref{lem:linearplanarN} to upper-bound the total number of vertices in $G''$ by $36k$, and in particular the total number of bricks by $24k$.
\end{proof} 

In the next lemma, we use \cref{obs:blocks} to guarantee the existence of a brick with sufficiently many clean faces to support a safe reduction rule.
\begin{lemma}
\label{lem:fixedreducetwo}
Assume $|T_2|\geq 420k+1$, where $k$ is the size of a provided vertex cover of $G$. Then we can, in polynomial time, either correctly determine that $\mathcal{I}$ is a no-instance or find a vertex $v\in  T_2$ with the following property: $\mathcal{I}$ is a yes-instance if and only if so is the instance~$\mathcal{I'}$ obtained from $\mathcal{I}$ by removing $v$. 
\end{lemma}

 \begin{proof}
By our assumption on the size of $T_2$, there must exist $16$ consecutive clean faces in some brick $P$. Indeed, by \cref{lem:specialf} there are at most $36k$ vertices in $T_2$ incident to any special face of $D_2$, meaning that at least $384k+1$ vertices in $T_2$ which are incident exclusively to faces occurring in bricks. The existence of $P$ then follows by the number of bricks being upper-bounded by $24k$ as per \cref{obs:blocks}. Let us denote the faces occurring in $P$ as  $f_1,\dots,f_{16}$ and recall that each of these faces of $D_{a,b}$ only contain vertices of degree $1$ and $0$ in $D$. For $2\leq i\leq 16$, let $v_i$ be the unique vertex in $T_2$ shared by $f_i$ and $f_{i+1}$; let $v_1$ ($v_{17}$) be the unique vertex in $T_2$ on the boundary of $f_1$ ($f_{16}$, respectively) but no other face in $P$.

Let us call a clean face \emph{superclean} if it does not contain any degree-$0$ and degree-$1$ vertices from the same clusters as those of $a$ and $b$.
Assume that there are no three consecutive superclean faces in $f_1,\dots,f_{16}$. Then we claim that $\mathcal{I}$ must be a no-instance: indeed, a solution for \CPLS\ consists of a set of curves which only traverse each vertex at most once, but at most two curves for the clusters of $a$ and $b$ can only visit at most four out of the sixteen faces in $f_1,\dots,f_{16}$.

We proceed by assuming that a consecutive set of three of the above faces, say $f_{i-1}$-$f_{i+1}$, are superclean. For $j\in \{i,i+1\}$, let $v_j$ be \emph{bad} if it belongs to a different cluster than $v_{j-1}$ and $v_{j+1}$ but the two latter vertices belong to the same cluster. Observe that if both $v_i$ and $v_{i+1}$ are bad then $\mathcal{I}$ must be a no-instance: indeed, in this case no curve can connect $v_{i-1}$ with $v_{i+1}$ without crossing through a vertex belonging to another cluster. By symmetry, we proceed under the assumption that $v_i$ is not bad.

Next, we check whether there is a vertex $c\in f_{i-1}$ such that $c$ belongs to a cluster different from the cluster(s) of $v_{i-1}$ and $v_{i}$, and yet the cluster of $c$ also contains some vertex $d$ outside of $f_{i-1}$. If this check succeeds then $\mathcal{I}$ must be a no-instance since every curve connecting $c$ to $d$ must cross the boundary of $f_i$. We perform the analogous check for $f_i$, and then---assuming no such $c$ exists---distinguish two cases based on the cluster of $v_i$:

For the first case, let $v_i$ belong to the same cluster as either $v_{i+1}$ or $v_{i-1}$. We provide the argument for the former case, whereas the latter case is entirely symmetric. Set $v:=v_i$ and assume that the instance $\mathcal{I'}$ obtained from $\mathcal{I}$ by removing $v$ is a yes-instance. A solution for $\mathcal{I'}$  must contain a (possibly degenerate) curve for the cluster of $v$ which touches $v_{i+1}$, and since $v_{i+1}$ has degree $2$ this curve can be assumed to either end at $v_{i+1}$ or cross into $f_i$. In either of these cases, we can extend the curve to touch $v_i$ in order to also obtain a solution for $\mathcal{I}$, whereas---since $f_i$ is clean---the only other curve such an extension may need to intersect would be a curve in the solution of $\mathcal{I'}$ connecting $a$ to $b$. If this happens, we redraw the curve connecting $a$ to $b$ to follow the edges $av$ and then $vb$ in $f_{i-1}$ whenever it would intersect into $f_i$; this is safe since we ensured $f_i$ is superclean. Hence, in this case we obtain a valid solution for $\mathcal{I}$, as desired.

For the second case, let $v_{i}$ belong to a different cluster than both $v_{i+1}$ and $v_{i-1}$, and recall that since $v_i$ is not bad the latter two vertices must themselves belong to different clusters. Say that these vertices belong to the clusters $V_s$, $V_p$ and $V_q$, respectively. Here, we need to perform a set of additional checks before proceeding: we check whether the set of degree-one vertices adjacent to $a$ (or, symmetrically, $b$) inside $f_i$ contains a pair of vertices $x,y$ such that $x\in V_s$, $y\in V_q$, and the order of neighbors of $a$ in a circular traversal around $a$ (or, symmetrically, $b$) contains the pattern $v_{i+1},y,x,v_i$. If this check fails, then $\mathcal{I}$ must be a no-instance since every curve connecting $y$ to $v_i$ must cross every curve connecting $x$ to $v_{i+1}$. We also perform an analogous check for the face $f_{i-1}$. 

Now assume that the above checks did not fail, and that the instance $\mathcal{I'}$ obtained from $\mathcal{I}$ by removing $v:=v_{i+1}$ is a yes-instance. We construct a solution for $\mathcal{I}$ from one for $\mathcal{I'}$ as follows. For each cluster other than $V_s$, we keep their curves as in the solution for $\mathcal{I'}$, but with the difference that whenever such a curve would intersect either of the two edges incident to $v_i$ we instead redraw it to follow that edge without crossing it; this is possible since by the non-badness of $v_i$, the check performed for choices of $c$ and the supercleanliness of $f_i$ and $f_{i-1}$, the only cluster that can contain vertices in both of these faces is $V_s$. Next, we construct a curve for the cluster $V_s$ by first having it collect all of its vertices in $f_{i-1}$ (in a circular order along the boundary of $f_{i-1}$), then traversing through $v_i$ to $f_i$, and collecting all of its vertices in $f_{i}$. This concludes the construction of the sought-after solution for $\mathcal{I}$, and the lemma follows.
\end{proof} 

Our next goal will be to reduce the size of $T_1$. To do so, we will first reduce the total number of clusters occurring in the instance---in particular, while by now we have the tools to reduce the size of $G_2$ (and hence also the number of clusters intersecting $V(G_2)=T_2\cup Z$), there may be many other clusters that contain only vertices in $T_1$ and $T_0$. 
Let $V_i$ be a cluster which does not intersect $V(G_2)=T_2\cup Z$. Observe that if $V_i$ contains vertices in more than a single face of $D_2$, then $\mathcal{I}$ must be a no-instance; for the following, we shall hence assume that this is not the case. In particular, for a cluster $V_i$ such that all of its vertices are contained in a face $f$ of $D$, we define its \emph{type} $t(i)$ as follows. $t_i$ is the set which contains $f$ as well as all vertices $v$ on the boundary of $f$ fulfilling the following condition: each $v\in t(i)$ is adjacent to a pendant vertex $a\in V_i$ where $a$ is drawn in $f$. An illustration of types is provided in~\cref{fig:types}.

\begin{SCfigure}[3][t]
\includegraphics[scale=0.6]{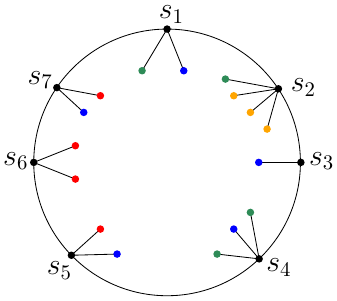}
\caption{An illustration of cluster types in a depicted face $f$. In this example, individual clusters are marked by colors and none of the vertices $s_1,\dots,s_7$ belong to any of the colored clusters. The types of the red, blue, yellow and green clusters are $\{f, s_5, s_6, s_7\}$, $\{f, s_1, s_3, s_4, s_5, s_7\}$, $\{f,s_2\}$ and $\{f, s_1, s_2, s_4\}$, respectively. Note that the depicted example cannot occur in a yes-instance since the curves for, e.g., the blue and red clusters would need to cross each other.}
\label{fig:types}
\end{SCfigure}

For the next lemma, let $\tau$ be the set of all types occurring in $\mathcal{I}$.

\begin{lemma}
\label{lem:reducing types}
If $|V(G_2)|=\alpha$ and $D_2$ has $\beta$ faces, then $|\tau|\leq 7\alpha+8\beta$ or $\mathcal{I}$ is a no-instance. 
\end{lemma}

 \begin{proof}
To prove the lemma, it will be useful to consider the auxiliary planar graph $G^*$ obtained from $G_2$ and its drawing $D_2$ by adding a vertex $v_f$ into each face of $D_2$ and making it adjacent to each vertex of $G_2$ on the boundary of $f$. We call the newly created edges incident to $v_f$ \emph{marker edges}, and observe that---by the planarity of $G^*$ and the fact that the marker edges induce a bipartite graph---the number of marker edges is upper-bounded by $2(\alpha+\beta)$. Crucially, each type $\{f,a_1,\dots,a_\ell\}\in \tau$ can be uniquely associated with the set of marker edges that is induced by $v_f$ and the vertices in that type (in particular, these will be the edges connecting the face in the type to the vertices occurring in that type).

By the above, we can immediately conclude that the number of types of size $2$ is upper-bounded by the total number of marker edges, i.e., $2(\alpha+\beta)$. The number of types of size $1$ is trivially upper-bounded by $\beta$.
As for a types of size $3$, we observe the following: for each type in $\tau$ containing two vertices $a_1$, $a_2$ on the boundary of its face $f$, either there is a simple closed curve in $f$ connecting $a_1$ to $a_2$ which does not cross any part of $D_2$ or any other such simple closed curve for types of size $3$, or $\mathcal{I}$ is a no-instance. In other words, we can enhance $D_2$ by adding, for each $\{f,a_1,a_2\}\in \tau$, an edge $a_1a_2$ into $f$ while maintaining planarity (or, if this fails, we can correctly reject the instance). By planarity, it then follows that the number of types of size $3$ cannot exceed $3(\alpha+\beta)$.

Finally, for types of size at least $4$, a hypothetical solution would have to add curves into $D_2$ connecting all the pendant vertices contributing to that type while maintaining planarity. Thus, we can once again invoke Lemma 13.3 in~\cite{fomin2019kernelization} (i.e., \cref{lem:linearplanarN}) to bound the total combined number of clusters in all such types---and hence also the number of types of size at least $4$---by $2\alpha+2\beta$, whereas if this bound is exceeded then we correctly identify that $\mathcal{I}$ is a no-instance.
\end{proof} 

\cref{lem:reducing types} allows us to bound the total number of cluster types in the instance via Lemma~\ref{lem:fewclusters} below. 
The proof relies on a careful case analysis that depends on the size of the cluster types of the considered clusters; for cluster types of size at least $4$ we directly obtain a contradiction with planarity, but for cluster types of size $2$ or $3$ we identify ``rainbow patterns'' that must be present and allow us to simplify the instance.

\begin{lemma}
\label{lem:fewclusters}
Assume there are three distinct clusters, say $V_1$, $V_2$ and $V_3$, which do not intersect $V(G_2)$ and all have the same type of size at least $2$. Then we can, in polynomial time, either correctly identify that $\mathcal{I}$ is a no-instance or find a non-empty set $A\subseteq T_1$ with the following property: $\mathcal{I}$ is a yes-instance iff so is the instance $\mathcal{I'}$ obtained from $\mathcal{I}$ by~removing~$A$.
\end{lemma}

 \begin{proof}
Let $f$ be the unique face in $t(1)$. 
Assume $|t(1)|=2$ and in particular $t(1)=\{f,a\}$ for some $a$ on the boundary of $f$. If $V_1$ contains only a single vertex, then the claim holds by setting $A:=V_1$. Otherwise, let $b,c\in V_1$, and let $L$ be the set of all pendant vertices adjacent to $a$ that occur between $b$ and $c$ when traversing the boundary of $f$ in $D$ (in particular, every vertex in $L$ must be a pendant vertex attached to $a$, just like $b$ and $c$). Notice that all vertices in $L$ as well as $b$ and $c$ must lie in different clusters than the one containing $a$. We distinguish two cases:

\begin{enumerate}
\item If $L$ contains at least one vertex $d\in V_i$ such that $V_i$ also contains at least one vertex inside $f$ which lies outside of $L$, then we identify $\mathcal{I}$ as a no-instance. This is correct, since a solution would need to draw a curve for $V_i$ and this would necessarily cross the curve connecting $b$ to $c$;
\item Otherwise, we set $A:=\{c\}$. This is correct, since every hypothetical solution for $\mathcal{I}'$ obtained after removing $c$ can be extended to a solution for $\mathcal{I}$ as follows: at the point where the curve for $V_1$ reaches $b$, we let it (a) follow the edge $ba$, then follow the boundary of all curves and edges incident to $L$, then follow the edge $ac$ to reach $c$, and finally backtrack its steps to return to $b$. The fact that a  hypothetical solution for $\mathcal{I}$ implies the existence of a solution for $\mathcal{I}'$ follows directly by yes-instances being closed under vertex deletion.
\end{enumerate}

Next, assume $|t(1)|=3$ and in particular $t(1)=\{f,s,t\}$ for some $s,t$ on the boundary of $f$, and let $V_s$ and $V_t$ be the clusters containing $s$ and $t$, respectively; notice that $V_s$ could be identical to $V_t$. For $1\leq j\leq 3$, let $s^j$ ($t^j$) denote an arbitrary vertex in $V_j$ adjacent to $s$ ($t$, respectively). By renaming these three clusters, we may assume w.l.o.g.\ that in a traversal of the boundary of $f$, we encounter these vertices in the order $(s^1,s^2,s^3,t^3,t^2,t^1)$, meaning that they form a ``rainbow pattern''. Indeed, if this were not the case, then a hypothetical solution would need to have the curves for at least two of these clusters intersect each other, meaning that we can immediately reject instances without this rainbow pattern. Moreover, a circular traversal of $f$ allows us to partition the vertices in $T_1$ based on their position w.r.t.\ the rainbow pattern identified above. In particular, let 
\begin{itemize}
\item $L_0$ be the vertices encountered from $t^1$ to $s^1$, 
\item $L_1$ be the vertices encountered from $s^1$ to $s^2$ and from $t^2$ to $t^1$,
\item $L_2$ be the vertices encountered from $s^2$ to $s^3$ and from $t^3$ to $t^2$,
\item $L_3$ be the vertices encountered from $s^3$ to $t^3$.
\end{itemize}

For $L_1$ and $L_2$, we can now perform an analogous case distinction as we did for $L$ in the case of $|t(1)|=2$ earlier. In particular:

\begin{enumerate}
\item If $L_1$ contains at least one vertex $d\in V_i$ such that $V_i$ also contains at least one vertex in $L_2$, then we identify $\mathcal{I}$ as a no-instance. This is correct, since a solution would need to draw a curve for $V_i$ and this would necessarily cross the curve connecting $s^2$ to $t^2$;
\item Otherwise, we set $A:=\{s^2\}$. This is correct, since every hypothetical solution for $\mathcal{I}'$ obtained after removing $s^2$ can be extended to a solution for $\mathcal{I}$ as follows. First, we observe that a solution for $\mathcal{I}'$ cannot separate $s^2$ from $t^2$ due to the condition stated in the previous point not being satisfied. Hence, to extend a solution for $\mathcal{I}'$ to $\mathcal{I}$ we can extend the curve for $V_2$ as follows: when it touches $t^2$ it will make a detour to reach $s^2$, and then route back to the vicinity of $t^2$, and then proceed as per the original solution for $\mathcal{I}'$.
The fact that a  hypothetical solution for $\mathcal{I}$ implies the existence of a solution for $\mathcal{I}'$ follows directly by yes-instances being closed under vertex deletion.
\end{enumerate}

Finally, if $|t(1)|\geq 4$ then a hypothetical solution would need to contain three pairwise non-crossing curves connecting the three vertices on the boundary of $f$ belonging to $t(1)$ (implying the existence of a $K_{3,3}$ with one side being obtained by the contraction of the three clusters), which is not possible by the planarity of the solution.
\end{proof}

Having bounded the number of cluster types (\cref{lem:reducing types}) and the number of clusters of each type (\cref{lem:fewclusters}), it remains to bound the number of vertices in each of the clusters.

\begin{lemma}
\label{lem:reducecluster}
Let $a\in V(G_2)$ and $f$ be a face of $D_2$ incident to $a$, and let $k$ be the size of a provided vertex cover of $G$. Assume a cluster $V_i$ contains at least $2k+3$ vertices in $T_1$ that are adjacent to $a$ and lie in $f$. Then we can, in polynomial time, either correctly identify that $\mathcal{I}$ is a no-instance, or find a vertex $q\in V_i\cap T_1$ such that $\mathcal{I}$ is a yes-instance if and only if so is the instance $\mathcal{I}'$ obtained by deleting $q$.
\end{lemma}

\begin{proof}
Let $v_1,\dots,v_{2k+3}$ be the first $2k+3$ vertices in $V_i\cap T_1$ in the neighborhood of $a$ encountered when walking in a counterclockwise manner along the boundary of $f$. For each $j\in [2k+2]$, let $L_j$ denote the set of vertices in $T_1$ that are encountered in the same traversal between $v_j$ and $v_{j+1}$; note that $L_j$ may be empty. Let $L^*$ be the set of all vertices encountered by following the counterclockwise traversal beyond $v_{2k+3}$. To avoid confusion, we remark that $L^*$ will also contain vertices outside of $T_1$.

We say that $L_j$ is \emph{bad} if it contains at least one vertex $d\in L_j\cap V_p$ (for some $p$) such that $V_p$ also contains at least one vertex in or on the boundary of $f$ but outside of $L_j$; more precisely, it contains a vertex in either $L^*$ or in $L_{j'}$ for some $j'\neq j$. If $L_j$ is not bad, then we call it \emph{good}. We now distinguish two cases that are somewhat analogous to the cases considered in the proof of \cref{lem:fewclusters}.

Assume there are at least $2k+2$ bad sets $L_{j^1},\dots, L_{j^{2k+2}}$. We claim that in this case, $\mathcal{I}$ can be rejected. Indeed, let us consider a hypothetical solution for $\mathcal{I}$, and notice that such a solution must include a curve that connects a vertex from each $L_{j^u}$, $u\in [2k+2]$, to some vertex on the boundary of $f$ in $D$ that lies outside of $L_{j^u}$. 
Crucially, such a curve must separate $v_{j^u}$ from $v_{j^{u+1}}$ in $f$, and in particular the set of these curves partition $f$ into $2k+3$ many connected regions such that each of the vertices $v_{j^u}$ lies in a separate region. 

Next, we consider the curve for $V_i$ in the hypothetical solution and divide it into $2k+1$ subcurves which start and end at the individual vertices $v_{j^u}$, $u\in [2k+2]$; note that the curve for $V_i$ need not visit the vertices in $V_i$ in any particular order (and, e.g., not in the order $v_1,\dots,v_{2k+3}$). However, each of the at least $2k+1$ subcurves would have to cross through the boundary of $f$ in $D_2$---but this is a contradiction with the fact that the boundary of $f$ in $D_2$ cannot have more than $2k$ vertices due to $k$ being the size of a provided vertex cover of $G$.

Hence, we are left with the case where there are at most $2k+1$ bad sets $L_j$. 
Since the total number of such sets is $2k+2$, there must be at least one good set, say $L_p$. We now set $q:=v_p$. To conclude the proof, it remains to argue that a hypothetical solution $S$ for $\mathcal{I}'$ can be extended to a solution for $\mathcal{I}$ (the converse holds by the fact that yes-instances are closed under vertex deletion). Towards this, we begin by first taking $S$ and topologically shifting each of the curves which happen to intersect or touch the edge $av_p$ to avoid any collisions; this can be done without causing any new intersections between the shifted curves and other edges of $D$ or other curves since $v_p$ is a pendant vertex. 

Before proceeding, it will be useful to observe that since $L_p$ is good, the drawing of all the curves in $S$ intersecting the vertices in $L_p$ does not touch the boundary of $f$ outside of $L_p$. Let us denote by $\delta$ the curve in $S$ for the cluster $V_0$ containing $a$.
For our second step of adapting $S$, we alter the curve $\iota$ for $V_i$ as follows: when the curve touches $v_{p+1}$, we route it along the edge $v_{p+1}a$ and then have it follow the boundary of $f\cup (S\setminus \{\iota\})$ until it reaches $v_p$. We then have $\iota$ touch $v_p$ and subsequently backtrack in the same way until we get back to $v_{p+1}$, and at that point we follow the original route of $\iota$. As observed at the beginning of this paragraph, this can be done without crossing any edge in $D$ and any curve in $S$ possibly except for $\delta$.

For our third and final step of adapting $S$, we reroute $\delta$ to obtain a solution for $\mathcal{I}$. Specifically, at this point $\delta$ may cross $\iota$ if it enters or exits from $a$ between $v_p$ and $v_{p+1}$ (in a counterclockwise circular traversal of the edges incident to $a$), and this may happen up to two times. We alter the drawing of $\delta$ as follows: whenever it enters (or exits) $a$ between $v_p$ and $v_{p+1}$, we instead enter (or exit) $a$ immediately to the left of $v_p$, then have $\delta$ follow along the edge $av_p$ and then along the boundary of $\iota$ until we reach the angular position at which $\delta$ originally entered (or exited) $a$, at which point we have $\delta$ follow its original route. Now none of the curves cross each other or any of the edges in $D$, and hence we have obtained a solution for $\mathcal{I}$ as required. 
\end{proof}

We now have all the ingredients required to obtain our polynomial kernel:

\begin{theorem}
\CPLSF\ has a cubic kernel when parameterized \mbox{by the vertex cover number.}
\end{theorem}

 \begin{proof}
We begin by computing a vertex cover $X$ of size at most $2k$ and then the set $Z$ whereas $|Z|\leq 6k$. We then exhaustively apply \cref{lem:fixedreducetwo} to either solve the input instance or reduce the number of vertices in $T_2$ to at most $840k$. At this point $G_2$ has at most $846k$ vertices and hence at most $1692k$ faces. 

We now invoke \cref{lem:reducing types} to either reject the instance or obtain a guarantee that the number $|\tau|$ of types of clusters which do not intersect $V(G_2)$ is at most $5922k+13536k=19458k$. Next, we exhaustively apply \cref{lem:fewclusters} to either solve the instance, or obtain a bound of $38916k$ on the total number of clusters which (1) do not intersect $V(G_2)$ and at the same time (2) contain at least one vertex in $T_1$. 
For each cluster which contains vertices exclusively in $T_0$, we check whether all of its vertices occur in the same face (in which case we can safely delete it from the instance) or not (in which case we can reject the instance).
Since the number of clusters which do intersect $V(G_2)$ is trivially upper-bounded by $840k$, we are guaranteed that the instance at this point contains at most $39756k$ clusters. 

Before applying the final reduction rule, we need to bound the number of pairs $(v,f)$ such that $v$ is adjacent to at least one pendant vertex that is drawn in the face $f$. Towards this, let us consider the auxiliary graph $H$ whose one side consists of $X$ and whose other side consists of the set of faces in $D_2$. We add an edge $(v,f)$ into $H$ whenever $v\in X$ is incident to the face $f$ of $D_2$. Since $|V(H)|\leq 2k+1692k$ and $H$ is a bipartite planar graph, the total number of edges in $H$ is upper-bounded by $3388k$---and this then also bounds the number of pairs $(v,f)$ such that $v$ is adjacent to at least one pendant vertex that is drawn in the face $f$. Now, we apply exhaustively \cref{lem:reducecluster} for each of the at most $3388k$ choices of $(a,f)\in E(H)$ and for each choice $V_i$ of the at most $39756k$ clusters remaining in the instance; at the end of this procedure, we have either solved the instance or are left with a guarantee that each $V_i$ contains at most $4k+3$ vertices in $T_1$ that are adjacent to $a$ and lie in $f$. Altogether, this yields a bound of at most $538 773 312 k^3+404 079 984 k^2$ vertices in $T_1$. Finally, we keep at most one isolated vertex per cluster in each of the faces of $D$, which allows us to bound the number of vertices in $T_0$ by $67 267 152\cdot k^2$, completing our kernel.
\end{proof}

\subsection{The Variable-Embedding Case}

For \CPLS, we can immediately observe that the vertices in $T_0$ are entirely irrelevant.

\begin{observation}
\label{obs:variable-degzero}
Let $(G',\mathcal{V'})$ be the instance of \CPLS\ obtained by removing all vertices in $T_0$ from $(G,\mathcal{V})$. Then $(G',\mathcal{V'})$ is a yes-instance of \CPLS\ if and only if so is $(G,\mathcal{V})$.
\end{observation}

 \begin{proof}
If $(G,\mathcal{V})$ is a yes-instance, then so is $(G',\mathcal{V'})$ since yes-instances are closed under vertex deletion. On the other hand, if $(G',\mathcal{V'})$ admits a solution consisting of a drawing $D$ and a set of curves, then we can construct a solution for $(G,\mathcal{V})$ by placing each vertex in $T_0$ directly on the curve for that cluster (whereas if the cluster was empty, we place all of the vertices next to each other in an arbitrary face of $D$).
\end{proof} 

For the remainder of this subsection, we hence assume that $T_0=\emptyset$. While ultimately our aim will be to reduce the size of each of the remaining neighborhood types, for now we will restrict ourselves to instead reduce the number of vertices belonging to the same cluster occurring in each of these neighborhood types.

\begin{observation}
\label{obs:variable-onetype}
Let $a,b\in T_1$ be two vertices that lie in the same neighborhood type and in the same cluster. Then $(G,\mathcal{V})$ is a yes-instance if and only if so is $(G-a,\mathcal{V}-a)$.
\end{observation}

 \begin{proof}
One direction once again follows by the fact that yes-instances are closed under vertex deletion. For the non-trivial direction, assume we have a solution for $(G-a,\mathcal{V}-a)$. We can extend such a solution to a solution for $(G,\mathcal{V})$ by placing $a$ at an $\epsilon$-distance next to $b$ in the drawing (for a sufficiently small $\epsilon$). The edge connecting $a$ to its neighbor follows the corresponding edge connecting $b$ to that same neighbor, and when the curve for the cluster containing $a$ and $b$ visits $b$ we route it through $a$ and then let it return to $b$ before proceeding.
\end{proof} 

\begin{lemma}
\label{lem:variable-twotype}
Let $a,b,c,d\in T_2$ be four vertices that lie in the same neighborhood type and in the same cluster. Then $(G,\mathcal{V})$ is a yes-instance if and only if so is $(G-a,\mathcal{V}-a)$.
\end{lemma}

 \begin{proof}
One direction once again follows by the fact that yes-instances are closed under vertex deletion. For the non-trivial direction, assume we have a solution for $(G-a,\mathcal{V}-a)$, let $v,w\in Z$ be the two neighbors of $a,b,c,d$ and let $V_i$ be the cluster containing $a,b,c,d$. Observe that the edges between $b,c,d$ and $v,w$ partition the plane into three regions, and that the curve for $V_i$ in the hypothetical solution must cross into at least two of these regions. However, the only way such a curve may move from one region to the other is by passing through one of these vertices; without loss of generality and by symmetry, let us assume that is the vertex $d$. We can now extend the solution for $(G-a,\mathcal{V}-a)$ to a solution for $(G,\mathcal{V})$ by placing $a$ right next to $d$ and having the curve for $V_i$ visit $a$ either just before or just after it visits $d$. 
\end{proof} 

\cref{obs:variable-onetype} and \cref{lem:variable-twotype} allow us to deal with cases when there are large clusters; however, it may still happen that there are many clusters occurring in the instance. We first handle this case for clusters intersecting $T_2$.

For a $2$-vertex subset $Q=\{a,b\}$ of $Z$, let a cluster $V_i$ occurring in $T_Q$ be $Q$-\emph{isolated} if $V_i\subseteq T_{\{a,b\}}\cup T_a \cup T_B$, and \emph{isolated} if it is $Q$-\emph{isolated} for some $Q$. We now show that if there are sufficiently many clusters in $T_2$, then there must be at least two isolated clusters.

\begin{lemma}
\label{lem:findisolated}
Let $|Z|\leq p$ for some integer $p$. Assume that $|T_2|\geq 60p+1$ and that we have exhaustively applied \cref{lem:variable-twotype}. Then the instance either contains two $Q$-isolated clusters for some $T_Q\subseteq T_2$, or must be a no-instance.
\end{lemma}

 \begin{proof}
Consider a hypothetical solution for $(G,\mathcal{V})$ consisting of a drawing $D$ along with a set of curves. Recalling the definition of $D_2$ and of special and clean faces from \cref{sub:fixembvc}, by \cref{lem:specialf} we obtain that the number of vertices in $T_2$ that are incident to at least one special face in $D$ is upper-bounded by $12p$.
By the exhaustive application of \cref{lem:variable-twotype}, we are guaranteed that the number of vertices in $T_2$ which lie in the same cluster as those that are incident to at least one special face in $D$ is upper-bounded by $36p$. 

This means that there remain at least $24p+1$ vertices in $T_2$ which are not incident to any special face, and these must lie in bricks as defined in \cref{sub:fixembvc} whereas the number of bricks is upper-bounded by $8p$ by \cref{obs:blocks}. Hence, there must be a brick containing at least $4$ vertices; these vertices all belong to the same $T_Q\subseteq T_2$ (for some $Q$) and must lie in at least two distinct clusters. Crucially, we observe that these two clusters must be $Q$-isolated, since their corresponding curves cannot leave their brick in $D$. This means that if $(G,\mathcal{V})$ is a yes-instance, then it must contain two $Q$-isolated clusters.
\end{proof} 

\begin{lemma}
\label{lem:removeisolated}
Let $V_i$. $V_j$ be two $\{a,b\}$-isolated clusters. Then $(G,\mathcal{V})$ is a yes-instance if and only if so is $(G-V_j,\mathcal{V}\setminus \{V_j\})$. 
\end{lemma}

 \begin{proof}
As in the previous reduction rules, one direction follows by yes-instances being closed under vertex deletion. For the non-trivial direction, let us consider the drawing $D$ used in a hypothetical solution for $(G-V_j,\mathcal{V}\setminus V_j)$. The drawing of the curve for $V_j$ in the solution together with the drawing of $V_j$ in $D$ only touches the drawing of $G-V_j$ at $a$ and $b$ and does not intersect any other curve in the solution. Hence, in the solution there must exist a curve $\alpha$ from $a$ to $b$ which does not intersect or touch any other part of $D$ or any other curve in the solution (indeed, such an $\alpha$ may be obtained by following the boundary of the object described in the previous sentence). We duplicate $\alpha$ to construct $|V_j|$-many such pairwise non-crossing $a$-$b$ curves that are very close to each other, and place each vertex in $V_j$ on its own $a$-$b$ curve at an $\epsilon$ distance from $a$ (for a sufficiently small $\epsilon$). Notice that can extend $D$ by adding each such vertex to $a$, $b$, or both by following the curve assigned to each vertex in $V_j$. Moreover, we can connect all of the vertices in $V_j$ with a single curve that only intersects the rest of the drawing at the vertices of $V_j$ by simply drawing a circular segment centered at $a$ with radius $\epsilon$.
\end{proof} 

The above reduction rules are sufficient to deal with all vertices except for those in $T_1$, i.e., the pendant vertices. We handle these with the following observation and two lemmas, which distinguish how large are the clusters containing these vertices.

\begin{observation}
\label{obs:Toneone}
Let $v\in T_1\cap V_i$ be such that $|V_i|=1$. Then $(G,\mathcal{V})$ is a yes-instance if and only if so is $(G-v,\mathcal{V}-v)$.
\end{observation}

\begin{lemma}
\label{lem:Tonetwo}
Let $Y=\{v\in T_1~|~v\in V_i, |V_i|=2\}$ be the set of all pendant vertices occurring in clusters of size $2$. If $|Y|\geq 6|V(G_2)|+2$, then we can in polynomial time either determine that $(G,\mathcal{V})$ is a no-instance, or identify a non-empty vertex subset $A\subseteq Y$ such that $(G,\mathcal{V})$ is a yes-instance if and only if so is $(G-A,\mathcal{V}\setminus A)$.
\end{lemma}

 \begin{proof}
Let $V_Y$ be the set of all clusters intersecting $Y$; clearly, $|Y|\leq 2|V_Y|$ and $|V_Y|\leq |Y|$. Towards the proof, let us consider the auxiliary bipartite $G_Y$ obtained as follows. The vertex set of $G_Y$ is $V(G_2)$. For each pair of vertices $w_1w_2\in V(G_Y)$, we add $w_1w_2$ to $E(G_Y)$ if and only if there is a cluster $V_x\in V_Y$ with the following property: for each $i\in \{1,2\}$, either $w_i\in V_x$ or $w_i$ is adjacent to some vertex in $Y\cap V_x$. Crucially, if $(G,\mathcal{V})$ is a yes-instance then $G_Y$ must be planar. Indeed, a planar drawing of $G_Y$ can be obtained from a solution for $(G,\mathcal{V})$ by copying the placement of all vertices in the solution, and for each edge $e\in E(G_Y)$ tracing the curve corresponding to that cluster, whereas whenever the curve ends at a pendant vertex we then follow the edge connecting it to its neighbor in $V(G_Y)$.

Hence, if $G_Y$ is not planar then we can immediately reject the instance, and otherwise it contains at most $3|V(G_2)|$ many edges. Since $6|V(G_2)|+2\leq |Y|$, we have that $3|V(G_2)|+1\leq |V_Y|$. By the pigeon-hole principle, there must some edge, say $w_1w_2\in E(G_2)$, which was ``created twice''---more precisely, both $w_1$ and $w_2$ are adjacent to pendants from (or belong to) two different clusters $V_a, V_b\in V_Y$. Let us set $A:=V_a\cap Y$ (in particular, $A$ could either be the whole $V_a$ or one pendant in $V_a$, depending on whether the other vertex in $V_a$ is a vertex of $G_2$ or not).

To conclude the proof, it remains to argue that if $(G-V_q,\mathcal{V}\setminus \{V_q\})$ is a yes-instance, then so is $(G,\mathcal{V})$. Towards this, we observe that a solution for the former can be extended for the latter by drawing at most two pendants in $V_q$ right next to the placement of the two pendants in $V_b$ (but on the same side of the curve for $V_b$). A curve connecting the two vertices in $V_q$ can then be drawn by simply following the curve for $V_b$.
\end{proof} 

\begin{lemma}
\label{lem:Tonethree}
Let $Y=\{v\in T_1~|~v\in V_i, |V_i|\geq 3\}$ be the set of all pendant vertices occurring in clusters of size at least $3$. If $|Y|\geq 15|V(G_2)|+1$, then $(G,\mathcal{V})$ is a no-instance.
\end{lemma}

 \begin{proof}
We begin by following the idea of the previous proof, and set $V_Y$ to be the set of all clusters intersecting $Y$; note that unlike in the previous proof, here $V_Y$ could be much smaller than $Y$. Let us now consider a hypothetical solution for $(G,\mathcal{V})$ consisting of a drawing $D$ along with a curve associated with each cluster. We say that such a curve \emph{visits} a vertex $v$ of $G_2$ if it either touches or passes through $v$, or if it touches or passes through a pendant in $Y$ that is adjacent to $v$.

We once again construct an auxiliary graph $G_Y$, but this time the construction will be more involved. In particular, $G_Y$ will be a bipartite graph with the first side of the bipartition being $V(G_2)$. For each cluster $V_i\in V_Y$ such that $\lfloor \frac{|V_i|}{3}\rfloor=p$, we now add $p$ vertices $v_i^1,\dots,v_i^p$ to the second side of the bipartition in $G_Y$, whereas $v_i^1$ will be adjacent to the first $3+(|V_i| \mod 3)$ vertices visited by the curve for $V_i$, and each subsequent $v_i^2,\dots,v_i^p$ will be adjacent to the next $3$ vertices visited by the curve for $V_i$. Observe that while $V_Y$ could be much smaller than $Y$, the number of vertices in the second side of the bipartition of $G_2$---which we hereinafter denote $V_Y^\text{split}$---is lower-bounded by $\frac{|Y|}{5}$, and every vertex in $V_Y^\text{split}$ has degree at least $3$.

Crucially, we can again argue that this new graph $G_Y$ must be planar. Indeed, the curves for the clusters in $V_Y$ are pairwise non-crossing, and this also holds if we split the curves into subcurves (one for each set of vertices adjacent to the second side of the bipartition of $G_2$). A planar drawing of $G_Y$ can then be obtained from the solution by placing each vertex $v_i^j$ at an arbitrary position on its corresponding subcurve and then tracing the subcurve to draw all of the edges to its neighbors in $G_2$ without interfering with any other subcurve.

Since $G_Y$ is planar, by the argument of \cref{lem:linearplanarN} we obtain a guarantee that if $(G,\mathcal{V})$ admits a solution then $|V_Y^\text{split}|\leq 3|V(G_2)|$. We already observed that $\frac{|Y|}{5}\leq |V_Y^\text{split}|$, and so we obtain that either $|Y|\leq 15|V(G_2)|$ or $(G,\mathcal{V})$ must be a no-instance.
\end{proof}

\begin{theorem}
\label{thm:varembkernel}
\CPLS\ has a linear kernel when parameterized \mbox{by the vertex cover number.}
\end{theorem}

\begin{proof}
We begin by computing a vertex cover $X$ of size at most $2k$ and then the set $Z$ whereas $|Z|\leq 6k$. \cref{obs:variable-degzero} allows us to remove all isolated vertices from the instance, and we exhaustively apply \cref{obs:variable-onetype} and \cref{lem:variable-twotype}. At this point, our instance contains at most $6k$ vertices in $Z$, and we are guaranteed that each cluster contains at most three vertices in each neighborhood type w.r.t.\ $Z$. 

Next, we exhaustively check for isolated clusters, and if we detect two $Q$-isolated clusters then we delete one as per \cref{lem:removeisolated}. After this step, we apply \cref{lem:findisolated} to either reject the instance or obtain a guarantee that $|T_2|\leq 360k$. We now exhaustively apply \cref{obs:Toneone}, \cref{lem:Tonetwo} and \cref{lem:Tonethree} to either reject the instance, or obtain an upper-bound on $|T_1|$ of $(6|V(G_2)|+2) + (15|V(G_2)|+1)$, where $V(G_2)=Z\cup T_2$. Altogether, this yields the bound $|T_1|\leq 21\cdot (366k) +3=7686k+3$.

The theorem now follows simply by the fact that the vertices in $G$ are partitioned into $Z$, $T_2$ and $T_1$, with the kernel having a size guarantee of at most $8052k+3$.
\end{proof}

As an immediate consequence of \cref{thm:varembkernel} and \cref{thm:exactVarialbe}, we can obtain an improved asymptotic running time for solving \CPLS:

\begin{corollary}
\label{cor:singleexpvc}
\CPLS\ can be solved in time $2^{\bigoh(k)}\cdot n^{\bigoh(1)}$, where $k$ and $n$ are the vertex cover number and the number of vertices of the input graph, respectively. 
\end{corollary}

\section{NP-completeness}\label{se:npc}

The \CPLS\ problem was shown \NP-complete even for independent c-graphs with an unbounded number of clusters, whose underlying graph is a subdivision of $3$-connected planar graph~\cite{AngeliniLBFPR17}. It is a simple exercise to show that \CPLS\ is \NP-complete for trees and for forests of stars. To this aim, consider the following transformation applied to an edge $(u,v)$ of a c-graph: (i) subdivide the edge $(u,v)$ with two dummy vertices $u'$ and $v'$, (ii) remove the edge $(u',v')$, and (iii) assign $u'$ and $v'$ to a new cluster. Clearly, this yields an equivalent instance of \CPLS. Applying the above procedure to a connected graph as long as there exist cycles yields a tree, whereas applying this procedure to all edges yields a forest of stars. 

We establish the \NP-completeness of \CPLS\ and \CPLSF\ when restricted to instances with up to three clusters. The proof is based on a polynomial-time Turing reduction from the \NP-complete  {\sc Bipartite 2-Page Book Embedding} problem~\cite{AngeliniLBFP21}.

In the following, we show the \NP-completeness of the \CPLS\ problem for instances with up to three clusters, even if the underlying graph has a prescribed embedding. To this aim, we will employ a polynomial-time Turing reduction from the \NP-complete  {\sc Bipartite 2-Page Book Embedding}  (for short, {\sc B2PBE}) problem~\cite{AngeliniLBFP21}.
A \emph{polynomial-time Turing reduction} (or \emph{Cook reduction}) from a problem $\cal A$ to a problem $\cal B$, denoted $\mathcal{A} \leq^P_T \mathcal{B}$, is an algorithm that solves $\cal A$ using a polynomial number of calls to a subroutine for $\cal B$, and polynomial-time additional extra time. 

Let $G=(U_b,U_r,E)$ be a bipartite planar graph, where the vertices in $U_b$ and $U_r$ are called \emph{black} and \emph{red}, respectively. 
A \emph{bipartite $2$-page book embedding} of $G$ is a planar embedding $\cal E$ of $G$ in which the vertices are placed along a Jordan curve $\ell_{\cal E}$, called \emph{spine} of ${\cal E}$, the black vertices appear consecutively along $\ell_{\cal E}$ (and, hence, the red vertices do as well), and each edge lies entirely in one of the two regions of the plane, called \emph{pages}, bounded by $\ell_{\cal E}$. 
The {\sc B2PBE} problem asks whether a given bipartite graph admits a bipartite $2$-page book embedding. 

Let $G^+$ be a planar supergraph of $G$ whose vertex set is $U_b \cup U_r$ and whose edge set is $E \cup E(\mathcal{\sigma})$, where $\sigma$ is a cycle that traverses all the vertices of $U_b$ (and, thus, also of $U_r$) consecutively. 
A cycle $\sigma$ exhibiting the above properties is a \emph{connector} of $G$, and its edges are called \emph{connecting edges}\footnote{In~\cite{AngeliniLBFP21}, connectors are called {\em saturators} and connector edges are called {\em saturating edges}. To avoid confusion with the the standard notions of {\em saturation} and {\em saturating edges} used in the context of {\sc C-Planarity} that we adopt in this paper (see \cref{se:preliminaries}), we have decided to rename such concepts within the extend of this section.}. 
A connecting edge is \emph{black} if it connects two black vertices, and it is \emph{red} if it connects two red vertices.  
By definition, $\sigma$ is a Hamiltonian cycle of $G^+$ consists of four paths: A path $\pi_b=(b_s,\dots,b_t)$ passing through all the vertices of $U_b$ (consisting of black connecting edges), a path $\pi_r=(r_s,\dots,r_t)$ passing through all the vertices of $U_r$ (consisting of red connecting edges), and two edges, namely, the edge $(b_t,r_s)$ and the edge $(r_t,b_s)$. 
The end-vertices of $\pi_b$ and the end-vertices of $\pi_r$ are the \emph{black} and the \emph{red end-vertices} of $\sigma$, respectively. 
We will exploit the following characterization for the {\sc B2PBE} problem.

\begin{longlemma}[\cite{AngeliniLBFP21}]\label{le:characterization-btpbe}
	A bipartite graph $G=(U_b,U_r,E)$ admits a bipartite $2$-page book embedding if and only if it admits a connector.
\end{longlemma}

We will exploit a specialized version of the {\sc B2PBE} problem, which we call {\sc Bipartite 2-Page Book Embedding with Prescribed End-vertices} (for short, {\sc B2PBE-PE}), defined as follows. Given a bipartite planar graph  $G=(U_b,U_r,E)$ and an ordered quadruple $(b_s,b_t,r_s,r_t) \in U_b \times U_b \times U_r \times U_r$, the {\sc B2PBE-PE} problem asks whether $G$ admits a bipartite $2$-page book embedding $\cal E$ in which the vertices $b_s$, $b_t$, $r_s$, and $r_t$ appear in this counter-clockwise order along the spine of $\cal E$. 
By \cref{le:characterization-btpbe} and the definition of connector, we immediately get the following.

\begin{longlemma}\label{le:characterization-btpbe-pe}
	A bipartite graph $G=(U_b,U_r,E)$ with distinct vertices $b_s, b_t \in U_b$ and $r_s,r_t \in U_r$ admits a bipartite $2$-page book embedding ${\cal E}$ in which the vertices $b_s$, $b_t$, $r_s$, and $r_t$ appear in this counter-clockwise order along the spine of ${\cal E}$ if and only if it admits a connector $\sigma$ in which 
 $b_t$ immediately precedes $r_s$ counter-clockwise along $\sigma$ and
 $r_t$ immediately precedes $b_s$ counter-clockwise along $\sigma$.
\end{longlemma}

For convenience, the above statements can be essentially summarized as follows:

\begin{lemma}\label{le:characterization-btpbe-peFROMSHORT}
	A bipartite graph $G=(U_b,U_r,E)$ with distinct vertices $b_s, b_t \in U_b$ and $r_s,r_t \in U_r$ is a positive instance of {\sc B2PBE-PE} iff it admits a connector $\sigma$ in which 
	$b_t$  and
	$r_t$ immediately precede $r_s$ and
	 $b_s$, respectively, counter-clockwise along $\sigma$.
\end{lemma}

Next, we prove the main result of the section, which rules out the existence of fixed-parameter tractable algorithms for the {\sc Clustered Planarity with Linear Saturators} problem parameterized by the number of clusters, unless $\P = \NP$. 

\begin{theorem}\label{th:hard}
Both \CPLS\ and \CPLSF\ are \NP-complete, even when restricted to instances with at most three clusters.
\end{theorem}

\begin{proof}
The membership in \NP{} is obvious. In fact, given an independent c-graph $\mathcal{C}=(G,\{V_1,\dots,V_k\})$, it is possible to verify in (deterministic) polynomial time
whether a given collection $Z_1,\dots,Z_k$ of sets of non-edges of $G$, one for each of the $k$ clusters of $\cal C$, is such that $G \cup \bigcup^k_{i=1}Z_i$ is planar (possibly, for {\CPLSF}, respecting a fixed embedding of $G$) and, for each $i \in [k]$, the graph $G_i = (V_i,Z_i)$ is a path.

In the remainder, we show the \NP-hardness of the {\sc CPLS} problem, when the input independent c-graph contains only three clusters (hereafter denoted, for short, as the {\sc CPLS-3} problem). As discussed at the end of the proof, the same reduction shows that {\CPLSF} is \NP-hard for instances with at most three clusters.
Our reduction is from the \textsc{B2PBE} problem, which is known to be \NP-complete~\cite{AngeliniLBFP21}. In fact, it is shown in~\cite{AngeliniLBFP21} that the \NP-hardness holds even for {\em $2$-connected} graphs with a {\em fixed embedding} (for details, see also~\cite{AngeliniLBFP21}). The \NP-hardness of \textsc{B2PBE} immediately implies the \NP-hardness of the {\sc B2PBE-PE} problem for the same restricted instances. In fact, given a bipartite graph $G$, the {\sc B2PBE} problem for~$G$ can be solved using a polynomial number of calls to a sub-procedure solving {\sc B2PBE-PE} for all possible pairs of end-vertices, and in addition to that spends only polynomial time. That is, {\sc B2PBE} $\leq^P_T$ {\sc B2PBE-PE}. In \cref{le:reduction} below, we show a polynomial-time many-one reduction, denoted by $\leq^P_m$, from {\sc B2PBE-PE} to {\sc CPLS-3}. This provides, in combination with the polynomial-time Turing reduction from {\sc B2PBE} to {\sc B2PBE-PE}, a polynomial-time Turing reduction from {\sc B2PBE} to {\sc CPLS-3}, which proves the statement.

In our reduction from {\sc B2PBE-PE} to {\sc CPLS-3}, we will use the following gadgets. The \emph{originating gadget} (see \cref{fig:gadget-Q}) is an independent c-graph $\mathcal{Q} = (G_{\mathcal{Q}},\mathcal{V}_{\mathcal{Q}})$, where $\mathcal{V}_{\mathcal{Q}} = \{V^\mathcal{Q}_r, V^\mathcal{Q}_b, V^\mathcal{Q}_c\}$ consists of three clusters, defined as follows. The underlying graph $G_{\mathcal{Q}}$ of $\mathcal{Q}$ is obtained from the union of the 4-cycle $(b'_t,\beta,r'_s,\psi)$ (called \emph{outer cycle} of $\mathcal Q$), of the 4-cycle $(\alpha,r^*,\omega,b^*)$ (called \emph{inner cycle} of $\mathcal Q$), and of the edges $(b'_t,\alpha)$, $(\beta,r^*)$, $(r'_s,\omega)$, $(\psi, b^*)$, and $(r^*,b^*)$. The clusters of $\mathcal{Q}$ are $V^\mathcal{Q}_b = \{b'_t,b^*\}$, $V^\mathcal{Q}_r = \{r'_s,r^*\}$, and $V^\mathcal{Q}_c = \{\alpha,\beta,\psi,\omega\}$.
The \emph{traversing gadget}  (see \cref{fig:gadget-P})  is an independent c-graph $\mathcal{P} = (G_{\mathcal{P}}, \mathcal{V}_{\mathcal{P}})$, where $\mathcal{V}_{\mathcal{P}} = \{V^{\mathcal{P}}_b, V^{\mathcal{P}}_r, V^{\mathcal{P}}_c\})$ has three clusters, obtained from the originating gadget $\cal Q$ by removing the edge $(r^*,b^*)$ and by renaming the vertices 
$\beta$, $r'_s$, $\psi$, $b'_t$, $r^*$, $\omega$, $b^*$, and $\alpha$ as 
$\mu$, $r'_t$, $\nu$, $b'_s$, $r^\diamond$, $\chi$, $b^\diamond$, and $\lambda$, respectively.
The clusters of $\mathcal{Q}$ are $V^\mathcal{P}_b = \{b'_s,b^\diamond\}$, $V^\mathcal{P}_r = \{r'_t,r^\diamond\}$, and $V^\mathcal{P}_c = \{\lambda,\mu,\nu,\chi\}$.
In particular, $(b'_s,\mu,r'_t,\nu)$ is the \emph{outer cycle} of $\mathcal P$ and $(\lambda,r^\diamond,\chi,b^\diamond)$ is the \emph{inner cycle} of $\mathcal P$. Observe that, the underlying graphs of $\mathcal P$ and $\mathcal Q$ are $3$-connected. An illustration of the gadgets defined above is provided in \cref{fig:Q-plus}.

Let $W$ be a graph with $4$ special vertices $b_s$, $b_t$, $r_b$, and $r_t$. Let $W^*$ be the graph obtained by taking the union of $W$, $G_{\mathcal{Q}}$, and $G_{\mathcal{P}}$ and by identifying the vertices $b_t$ and $r_s$ of $W$ with the vertices $b'_t$ and $r'_s$ of $G_{\mathcal Q}$, respectively, and by identifying the vertices $b_s$ and $r_t$ of $W$ with the vertices $b'_s$ and $r'_t$ of $G_{\mathcal{P}}$, respectively. We say that $W^*$ is a \emph{$\mathcal{PQ}$-merge} of $W$.

\begin{figure}[tb!]
\centering
    \begin{minipage}{0.24\textwidth}
    \centering
        \includegraphics[width=\textwidth,page=1]{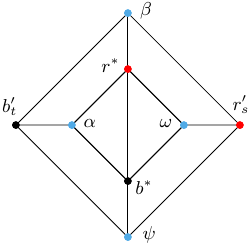}
        \subcaption{Gadget $\cal Q$}
        \label{fig:gadget-Q}
    \end{minipage}
    \begin{minipage}{0.24\textwidth}
    \centering
        \includegraphics[width=\textwidth,page=2]{figures/gadgets.pdf}
        \subcaption{Saturation of $\cal Q$}
        \label{fig:gadget-Q-saturation}
    \end{minipage}
        \begin{minipage}{0.24\textwidth}
    \centering
        \includegraphics[width=\textwidth,page=3]{figures/gadgets.pdf}
        \subcaption{Gadget $\cal P$}
        \label{fig:gadget-P}
    \end{minipage}
    \begin{minipage}{0.24\textwidth}
    \centering
        \includegraphics[width=\textwidth,page=4]{figures/gadgets.pdf}
        \subcaption{Saturation of $\cal P$}
        \label{fig:gadget-P-saturation}
    \end{minipage}
    \caption{Illustrations for the gadget $\mathcal{Q}$ and $\mathcal{P}$.}
    \label{fig:Q-plus}
\end{figure}

We are now ready to prove the following.

\begin{longclaim}\label{le:reduction}
{\sc B2PBE-PE} $\leq^P_m$ {\sc CPLS-3}.
\end{longclaim}

\begin{proof}
Given a connected bipartite planar graph $G=(U_b,U_r,E)$ and an ordered quadruple $(b_s,b_t,r_s,r_t) \in U_b \times U_b \times U_r \times U_r$, we construct an independent c-graph $\mathcal{C}=(H,\{V_b,V_r,V_c\})$  with three clusters as follows (see also \cref{fig:Q-plus-claim}). 
\begin{enumerate}
\item First, we initialize the underlying graph $H$ of $\mathcal{C}$ to be
the $\mathcal{PQ}$-merge of $G$ (which is well-defined, since $G$ contains the $4$ distinct vertices $b_s$, $b_t$, $r_s$, and $r_t$.). Also, we initialize $V_b = U_b \cup \{b^*,b^\diamond\}$, $V_r = U_r \cup \{r^*,r^\diamond\}$, and $V_c = \{\alpha,\beta,\lambda,\mu,\nu,\chi,\psi,\omega\}$. 

\item Second, we subdivide, in $H$, each edge $e$ of $E(G)$ with a dummy vertex $c_e$ and assign the vertex $c_e$ to $V_c$. We say that $c_e$ \emph{stems from} $e$.
 
\end{enumerate}

Clearly, the above reduction can be carried out in polynomial time, in fact, in linear time, in the size of $G$.
In what follows, we show that $\mathcal{C}$ is a positive instance of {\sc CPLS-3}
if and only if $(G,b_s,b_t,r_s,r_t)$ is a positive instance of {\sc B2PBE-PE}.

$(\Rightarrow)$ Suppose first that $G$ admits a bipartite $2$-page book embedding $\cal E$ in which the end-vertices $b_s$, $b_t$, $r_s$, and $r_t$ appear in this counter-clockwise order along the spine of $\cal E$ (see \cref{fig:bip2pageBook}). Then, by \cref{le:characterization-btpbe-pe}, $G$ admits a connector $\sigma$ in which 
 $b_t$ immediately precedes $r_s$ counter-clockwise along $\sigma$ and
 $r_t$ immediately precedes $b_s$ counter-clockwise along $\sigma$.
Recall that $\sigma$ consists of four paths: 
A path $\pi_b=(b_s,\dots,b_t)$ passing through all the vertices in $U_b$ (consisting of black connecting edges), a path $\pi_r=(r_s,\dots,r_t)$ passing through all the vertices in $U_r$ (consisting of red connecting edges), the edge $(b_t,r_s)$, and the edge $(r_t,b_s)$.  
Let $\mathcal{E}_\sigma$ be any planar embedding of $G_\sigma = (V(G),E(G) \cup E(\sigma))$, which exists by the definition of connector (see \cref{fig:connector}). All the faces of $\mathcal{E}_\sigma$ are incident to a single subpath of $\pi_b$ and to a single subpath of $\pi_r$.
Moreover,  $\mathcal{E}_\sigma$ has the following properties:
\begin{enumerate}[Pi]
\item \label{p:init}There exists a unique face $f'_{in}$ (resp. $f'_{out}$) of $\mathcal E_\sigma$ that lies in the interior of $\sigma$ (resp. in the exterior of $\sigma$) and is incident to the edge $(b_t,r_s)$. Similarly, there exists a unique face $f''_{in}$ (resp. $f''_{out}$) of $\mathcal E_\sigma$ that lies in the interior of $\sigma$ (resp. in the exterior of $\sigma$) and is incident to the edge $(r_t,b_s)$. Moreover, each of these faces is incident to a unique edge connecting a black and a red vertex; we let $e'_{in}$, $e'_{out}$, $e''_{in}$, and $e''_{out}$ be these edges incident to $f'_{in}$, $f'_{out}$, $f''_{in}$, and $f''_{out}$, respectively.\footnote{It may happen that $f'_{in} = f''_{in}$, in which case the edges $e'_{in}$ and $e''_{in}$ do not exist. Similarly if $f'_{out} = f''_{out}$.} We call these faces \emph{extremal}. 
\item \label{p:traverse} All other faces of $\mathcal E_\sigma$ are incident to exactly two edges connecting a black and a red vertex. We call such faces \emph{transversal}. 
\end{enumerate}

We show that $\cal C$ is a positive instance of \textsc{CPLS-3}. In particular, we 
shall compute sets $Z_b$, $Z_r$, and $Z_c$ of non-edges of $H$ such that
$H^* = (V(H),E(H) \cup Z_b \cup Z_r \cup Z_c)$ is planar and such 
that $H[Z_b]$, $H[Z_r]$, and $H[Z_c]$ are paths connecting the vertices of $V_b$, of $V_r$, and of $V_c$, respectively. To this aim, we show how to construct $H^*$ and simultaneously a planar embedding $\mathcal E_{H^*}$ of it, starting from $\mathcal E_\sigma$; refer to \cref{fig:saturator}.

\begin{figure}[hbt!]
\centering
    \begin{minipage}{\textwidth}
    \centering
        \includegraphics[width=.543\textwidth,page=8]{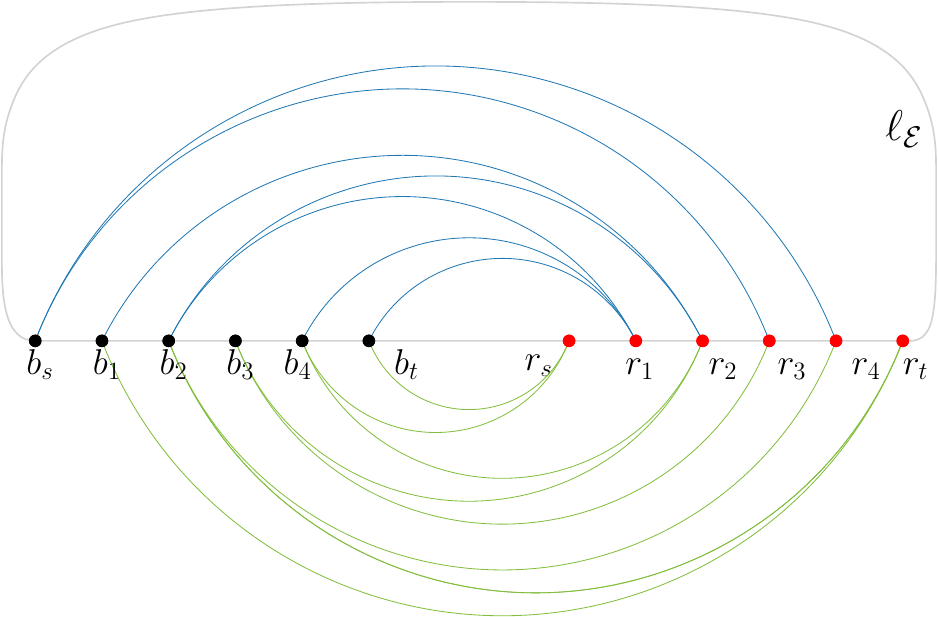}
        \subcaption{A bipartite $2$-page book embedding $\mathcal E$ of a graph $G$
        in which the vertices $b_s$, $b_t$, $r_s$, and $r_t$ appear in this counter-clockwise order along the spine $\ell_{\mathcal {E}}$ of $\cal E$.}
        \label{fig:bip2pageBook}
    \end{minipage}
    \\
    \begin{minipage}{\textwidth}
    \centering
        \includegraphics[width=.543\textwidth,page=5]{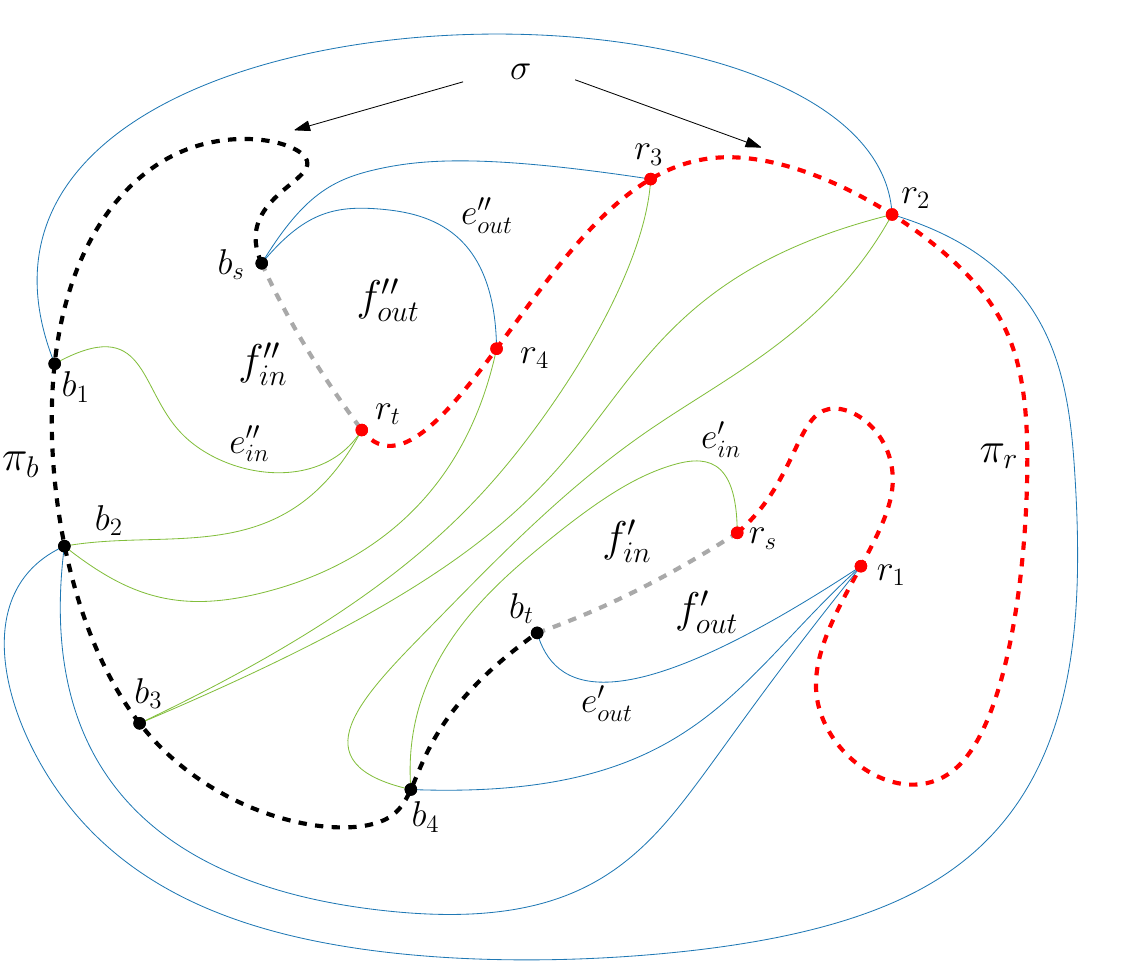}
        \subcaption{A planar embedding $\mathcal{E}_{\sigma}$ of the bipartite graph $G$ together with a connector $\sigma$ in which $b_t$ immediately precedes $r_s$ counter-clockwise along $\sigma$ and $r_t$ immediately precedes $b_s$ counter-clockwise along $\sigma$.
        The edges of $\sigma$ are dashed. The paths $\pi_b$ and $\pi_r$ are black and red, respectively. The edge $(b_t,r_s)$ belongs to both $G$ and $\sigma$.
        }
        \label{fig:connector}
    \end{minipage}
    \\
        \begin{minipage}{\textwidth}
    \centering
        \includegraphics[width=.543\textwidth,page=6]{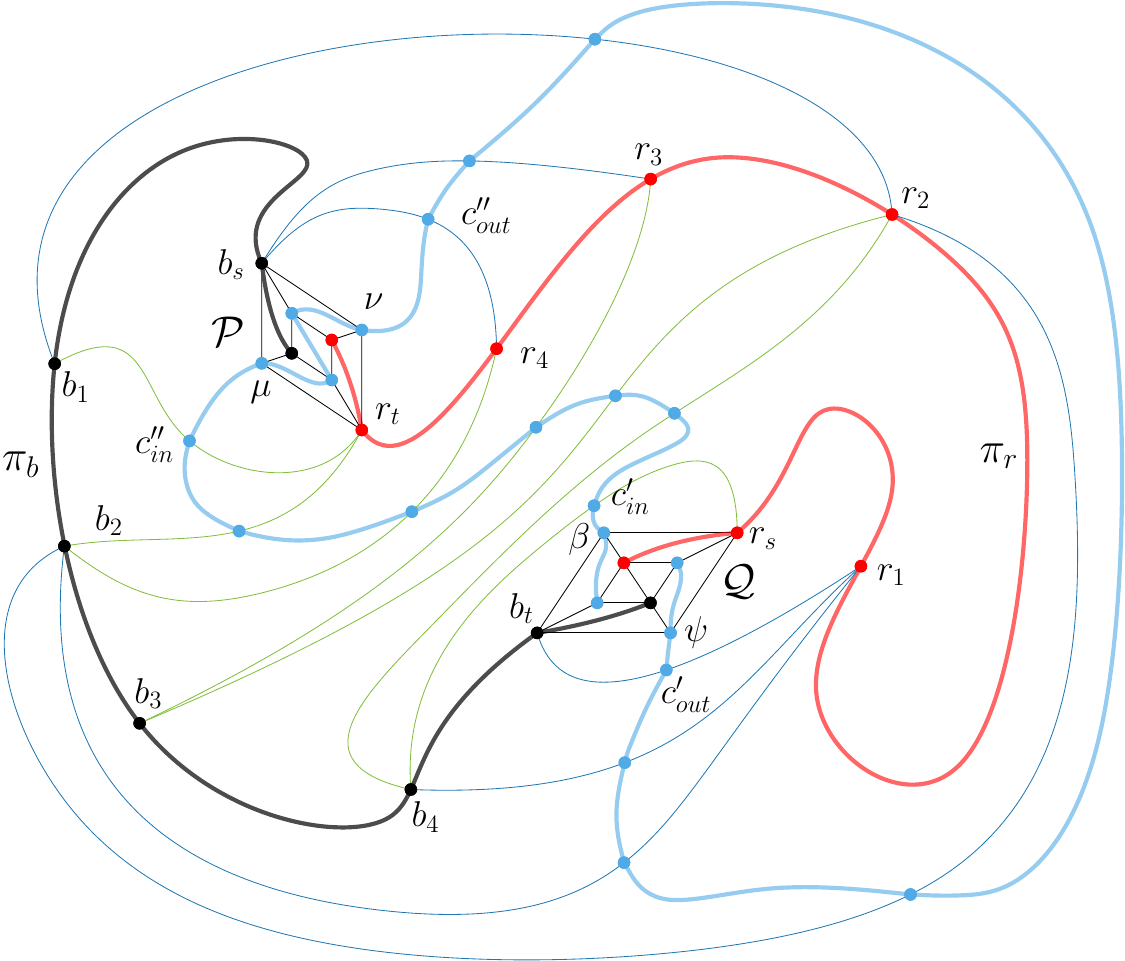}
        \subcaption{The planar embedding $\mathcal{E}_{H^*}$ of the linear saturation $H^*$ of the underlying graph of the independent c-graph $\cal C$, obtained from $\mathcal{E}_\sigma$.}
        \label{fig:saturator}
    \end{minipage}
    \caption{Illustrations for the proof of \cref{le:reduction}.}
    \label{fig:Q-plus-claim}
\end{figure}
\clearpage

We initialize $H^* = G_\sigma$, $\mathcal E_{H^*} = \mathcal E_\sigma$,  $Z_b = E(\pi_b)$, and $Z_r = E(\pi_r)$. 
Then, we draw each vertex $c_e$ of $V_c$ in the interior of the drawing of the edge $e=(b,r)$ of $G$ it stems from, and let the two parts of $e$ that connect $c_e$ with $b$ and with $r$ be the drawing of the edges $(b,c_e)$ and $(c_e,r)$ of $H$, respectively.
Second, by Properties {\bf P}\ref{p:init} and {\bf P}\ref{p:traverse}, we have that each transversal face contains exactly two vertices $c_{e}$ and $c_{g}$ of $V_c$. For each such face $f$, we add the edge $(c_{e},c_{g})$ to $Z_c$ and draw this edge inside $f$ in $\mathcal{E}_{H^*}$. For simplicity, we will denote the unique vertex of $V_c$ stemming from $e'_{in}$, $e'_{out}$, $e''_{in}$, and $e''_{out}$ as $c'_{in}$, $c'_{out}$, $c''_{in}$, and $c''_{out}$, respectively.
Also, we remove from $H$ and $\mathcal{E}_{H^*}$ the edge $(b_t,r_s)$ and draw the originating gadget $\mathcal Q$ arbitrarily close to the drawing of $(b_t,r_s)$ so that $b'_t$ coincides with $b_t$, so that $r'_s$ coincides with $r_s$, and that no crossing is introduced. 
W.l.o.g., we assume that the drawing of $\mathcal Q$ is such that the vertex $\beta$ shares a face with $c'_{in}$ and thus that the vertex $\psi$ shares a face with the vertex $c'_{out}$.
Similarly, we remove from $H$ and $\mathcal{E}_{H^*}$ the edge $(r_t,b_s)$ and draw the traversing gadget $\mathcal P$ arbitrarily close to the drawing of $(r_t,b_s)$ so that $b'_s$ coincides with $b_s$, so that $r'_t$ coincides with $r_t$, and that no crossing is introduced. 
W.l.o.g., we assume that the drawing of $\mathcal P$ is such that the vertex $\mu$ shares a face with $c''_{in}$ and thus that the vertex $\nu$ shares a face with the vertex $c''_{out}$.
Observe that the drawing of $\mathcal P$ and $\mathcal Q$ in $\mathcal E_{H^*}$ leaves the vertices of $H$ that do not belong to these graphs in the exterior of their outer cycles. This allows us to add 
the edges $(b'_t,b^*)$ and $(b'_s,b^\diamond)$ to $H^*$ and $Z_b$, 
the edges $(r'_s,r^*)$ and $(r'_t,r^\diamond)$ to $H^*$ and $Z_r$, and
the edges $(\alpha,\beta)$,  $(\psi,\omega)$, $(\mu,\lambda)$, $(\lambda,\chi)$, and $(\chi,\nu)$ to $H^*$ and $Z_c$, and 
to draw these edges planarly in $\mathcal{E}_{H^*}$. 
Notice that this concludes the construction of the sets $Z_b$ and $Z_r$.
To complete the construction of $H^*$ and $\mathcal{E}_{H^*}$, we add to $Z_c$ the edges $(\beta,c'_{in})$, $(\mu,c''_{in})$, $(\psi,c'_{out})$, and $(\nu,c''_{out})$, which can be done while preserving planarity as the endpoints of these edges are incident to distinct faces.\footnote{It may happen that $f'_{in} = f''_{in}$, in which case we add the edge $(\beta,\mu)$. Similarly if $f'_{out} = f''_{out}$, when we add the edge $(\psi,\nu)$.}
This concludes the construction of $Z_c$. Since $H^*$ is planar and the graphs $H[Z_b]$, $H[Z_r]$, and $H[Z_c]$ are paths, we have that $\mathcal C$ is a positive instance of \textsc{CPLS-3}.

$(\Leftarrow)$ Suppose now that $\mathcal C$ admits a planar saturation $H^*=(V(H), E(H) \cup \sum^3_{i=1} Z_i)$ via three sets $Z_1$, $Z_2$, and $Z_3$ of non-edges of $H$ such that $H[V_i]$ is a path, with $i \in [3]$. Let $\mathcal{E}_{H^*}$ be a planar embedding of $H^*$. 
We show that $(G,b_s,b_t,r_s,r_t)$ is a positive instance of {\sc B2PBE-PE}.
We will exploit the following property of a {$\mathcal{PQ}$-merge} of a $2$-connected graph.

\begin{longproperty}\label{prop:exterior-of-gadgets}
Let $W$ be a $2$-connected planar graph with $4$ distinct vertices $b_s$, $b_t$, $r_b$, and $r_t$, and let $W^*$ be a \emph{$\mathcal{PQ}$-merge} of $W$. If $W^*$ is planar, then in any planar embedding of $W^*$ it holds that:
\begin{itemize}
    \item 
the outer cycle of $\mathcal P$ separates the vertices of the inner cycle of $\mathcal P$ from the vertices in $V(W) \cup V({\mathcal{Q}})$, and
\item 
the outer cycle of $\mathcal Q$ separates the vertices of the inner cycle of $\mathcal Q$ from the vertices in $V(W) \cup V({\mathcal{P}})$.
\end{itemize}
\end{longproperty}

\begin{proof}
In the following, we assume that $W^*$ is planar and we let $\mathcal{E}_{W^*}$ be any of its planar embeddings.
Let $W^-_{\mathcal P} = W^*[V(W) \cup V(\mathcal{Q})]$. We show that (i) holds, i.e., we show that 
the outer cycle $\mathcal O$ of $\mathcal{P}$ separates the vertices of the inner cycle $\mathcal I$ of $\mathcal{P}$ from the vertices of $W^-_{\mathcal P}$; the proof that (ii) holds being analogous. First observe that, since $\mathcal P$ is $3$-connected, the vertices of $\mathcal I$ must lie on the same side of $\mathcal O$. For a contradiction, suppose that there exists a vertex $u \in V(W^-_\mathcal{P})$ that lies on the same side of $\mathcal O$ that hosts the vertices of $\mathcal I$. W.l.o.g, vertex $u$ and the vertices of $\mathcal I$ lie in the interior of $\mathcal O$.

We highlight some properties that will soon turn useful. First, observe that
$W^-_{\mathcal P}$ is $2$-connected, since $W$ and $\mathcal Q$ are $2$-connected and since $W^-_{\mathcal P}$ is obtained by identifying two distinct vertices of $W$ with two distinct vertices of $\mathcal Q$. By the same arguments, we also have that $W^*$ is $2$-connected. Moreover, by the definition of $\mathcal{PQ}$-merge, the vertices $b_t$ and $b_s$ form a separation pair of $W^*$. 

Let $\mathcal E_{\mathcal P}$ be the planar embedding of $\mathcal P$ induced by $\mathcal E_{W^*}$, which is unique (up to a flip and to the choice of the outer face) since $\mathcal P$ is $3$-connected.
Since $u$ lies in the interior of $\mathcal O$, then it must lie in the interior of one of the five faces of $\mathcal{E}_{\mathcal P}$ that are not bounded by the outer cycle of $\mathcal P$. Hereafter, we call these faces \emph{inward}. 
Since $W^*$ is $2$-connected, there must exist in $W^*$ two vertex-disjoint paths $P'$ and $P''$ connecting $u$ and any vertex of $\mathcal I$, say $\alpha$. In particular, 
one of these paths, say $P'$, must pass through $b_t$ and the other path must pass through $r_s$ (due to the fact that $b_t$ and $r_s$ form a separation pair of $W^*$). Let $P'_u$ be the subpath of $P'$ from $u$ to $b_t$ and $P''_u$ be the subpath of $P''$ from $u$ to $r_s$. Then, either $P'_u$, or $P''_u$, or both must intersect the boundary of the inward face of $\mathcal E_{\mathcal P}$ that contains $u$, which contradicts the planarity~of~$\mathcal{E}_{W^*}$.
\end{proof}

Recall that the graph $G$ is $2$-connected and that, by construction, $H$ is a $\mathcal{PQ}$-merge of a subdivision of $G$. Then, by \cref{prop:exterior-of-gadgets} and since $\mathcal{E}_{H^*}$ induces a planar embedding $\mathcal{E}_{H}$ of $H$, we have that in $\mathcal{E}_{H^*}$ it holds that {\bf (a)} the outer cycle of $\mathcal Q$ 
separates the vertices of the inner cycle of $\mathcal Q$ from the vertices of $V(H)$ not in $V(\mathcal Q)$. Similarly, we have that
the outer cycle of $\mathcal P$ 
separates the vertices of the inner cycle of $\mathcal P$ from the vertices of $V(H)$ not in $V(\mathcal P)$. We can exploit {\bf (a)} and {\bf (b)} to show the following.

\begin{longclaim}\label{cl:paths}
    The following statements about paths $P_b = H^*[V_b]$, $P_r = H^*[V_r]$, and $P_c = H^*[V_c]$ hold:
    \begin{itemize}
    \item $P_b$ contains the edges $(b^*,b'_t)$ and $(b^\diamond,b'_s)$ and has $b^*$ and $b^\diamond$ as endpoints; 
    \item $P_r$ contains the edges $(r^*,r'_s)$ and $(r^\diamond,r'_t)$ and has $r^*$ and $r^\diamond$ as endpoints; and
    \item $P_c$ contains the edges $(\alpha,\beta)$, 
    $(\nu,\chi)$, 
    $(\chi,\lambda)$, 
    $(\lambda,\mu)$, and
    $(\psi,\omega)$ and has $\alpha$ and $\omega$ as endpoints.
    \end{itemize}
\end{longclaim}

\begin{proof}
Refer to \cref{fig:gadget-Q-saturation,fig:gadget-P-saturation}. 
By {\bf (a)}, we have that the faces of the unique planar embedding of $\mathcal Q$ that are not bounded by the outer cycle of $\mathcal Q$ are also faces of $\mathcal E_{H*
}$. Thus, the only black vertex $b^*$ shares a face with is $b'_t$ and the only red vertex $r^*$ shares a face with is $r'_s$. This implies that $(b^*,b'_t) \in Z_b$ and that $(r^*,r'_s) \in Z_r$. In turn, the presence of these two edges in $H^*$ and the $3$-connectivity of $\mathcal Q$, imply that $\beta$ is the only cyan vertex that shares a face with vertex $\alpha$ and that $\psi$ is the only cyan vertex that shares a face with vertex~$\omega$. Therefore, $(\alpha,\beta), (\psi,\omega) \in Z_c$. In particular, $b^*$ is an endpoint of $P_b$, $r^*$ is an endpoint of $P_r$, and $\alpha$ and $\omega$ are the two endpoints of $P_c$. 
By {\bf (b)}, we have that the faces of the unique planar embedding of $\mathcal P$ that are not bounded by the outer cycle of $\mathcal P$ are also faces of $\mathcal E_{H*
}$. Thus, the only black vertex $b^\diamond$ shares a face with is $b'_s$ and the only red vertex $r^\diamond$ shares a face with is $r'_t$. This implies that $(b^\diamond,b'_s) \in Z_b$ and that $(r^\diamond,r'_t) \in Z_r$. In turn, the presence of these two edges in $H^*$, the $3$-connectivity of $\mathcal P$, and the fact that the endpoints of $P_c$ are $\alpha$ and $\omega$ 
imply that $(\nu,\chi), (\chi,\lambda), (\lambda,\mu) \in Z_c$.
\end{proof} 

By \cref{cl:paths}, the path $P_b=H[Z_b]$ consists of the edge 
$(b^\diamond,b'_s)$, followed by a path $\pi_b$ that traverses all the vertices of $U_b$ starting from $b'_s = b_s$ and ending at $b'_t = b_t$, followed by the edge $(b'_t,b^*)$.
Also, the path $P_r=H[Z_r]$ consists of the edge 
$(r^\diamond,r'_t)$, followed by a path $\pi_r$ that traverses all the vertices of $U_r$ starting from $r'_t = r_t$ and ending at $r'_s = r_s$, followed by the edge $(r'_s,r^*)$.

We are finally ready to show that $(G,b_s,b_t,r_s,r_t)$ is a positive instance of {\sc B2PBE-PE}. To this aim, by \cref{le:characterization-btpbe-pe}, starting from $\mathcal{E}_{H^*}$, we show how to construct a planar embedding of $G$ together with a connector $\sigma$ in which 
 $b_t$ immediately precedes $r_s$ counter-clockwise along $\sigma$ and
 $r_t$ immediately precedes $b_s$ counter-clockwise along $\sigma$; refer to \cref{fig:Q-saturation-to-book}.

\begin{figure}[h!]
\centering
        \includegraphics[width=.543\textwidth,page=7]{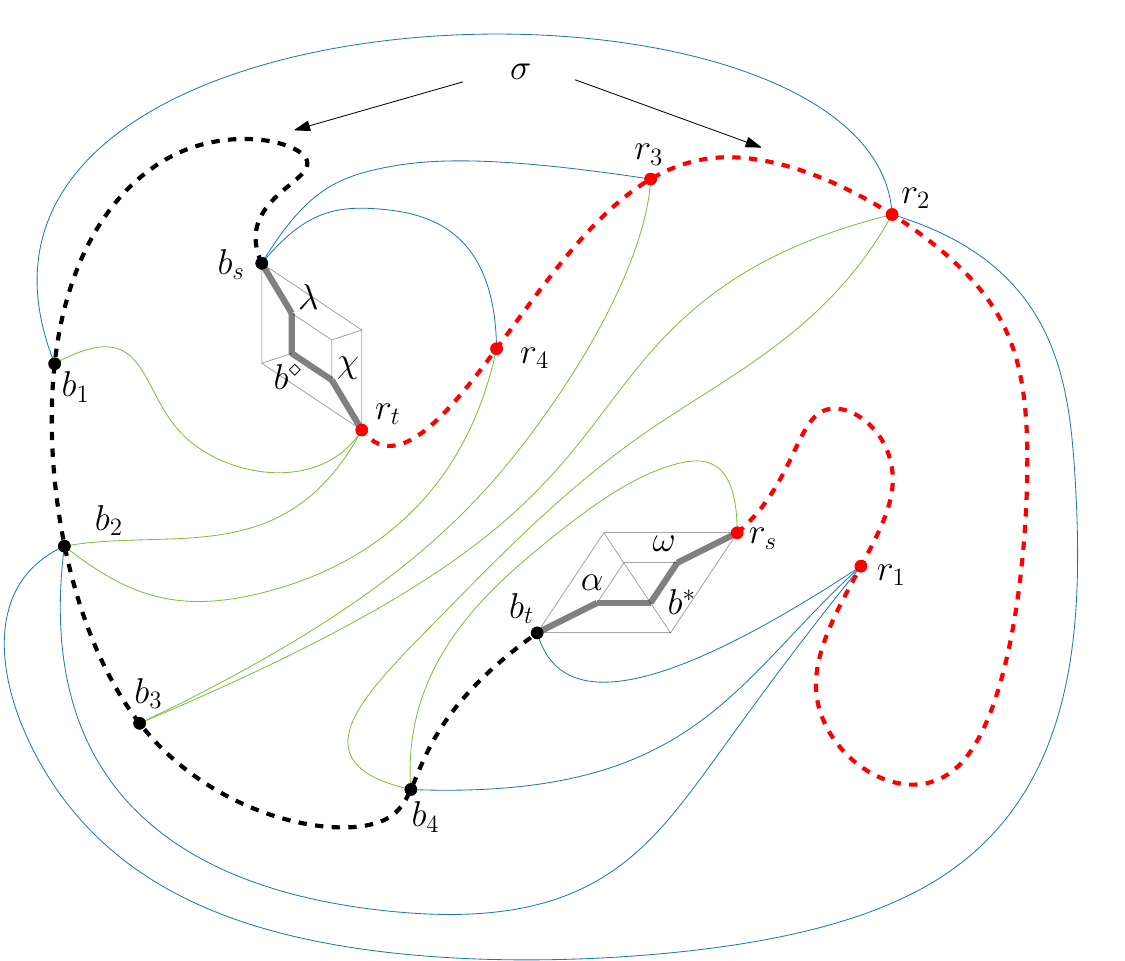}
    \caption{Illustration for the proof of \cref{cl:paths}. Construction a planar embedding of $G$ together with a connector $\sigma$ obtained starting from the planar embedding of $H^*$ in \cref{fig:connector}. The drawings of $\mathcal P$ and $\mathcal Q$ (which are not part of $G$) are shaded gray.}
    \label{fig:Q-saturation-to-book}
\end{figure}

First, we modify $H^*$ and $\mathcal{E}_{H^*}$ as follows. We remove the edges of $Z_c$. Then, for each edge $e=(u,v)$ in $G$, we smoothen the vertex $c_e$, i.e., we remove $c_e$ and let the union of the drawings of $(u,c_e)$ and $(c_e,v)$ be the drawing of the edge $(u,v)$. Finally, we remove the vertices of $\mathcal P$ and $\mathcal Q$ different from $b_s$, $b_t$, $r_s$, and $r_t$ as well as the drawing of the edges of $H^*$ whose endpoints are both in $V(\mathcal P)$ or both in $V(\mathcal Q)$.
Let $G^*$ be the resulting graph and $\mathcal{E}_{G^*}$ be the corresponding embedding obtained from $\mathcal{E}_{H^*}$.
Observe that, after the above steps, no vertex of $V_C$ belongs to $G^*$. Moreover, the following hold:
\begin{itemize}
\item $G^*$ is the union of $G$ and a subset $Z'_b$ and $Z'_r$ of the edges of $Z_b$ and $Z_r$, respectively;
\item $G^*[V_b]$ is the edge set of the path $\pi_b$; and 
\item $G^*[V_r]$ is the edge set of the path $\pi_r$.
\end{itemize}

We now show how to extend $\mathcal{E}^*$ to a planar embedding $\mathcal{E}_\sigma$ of $G \cup \sigma$, where $\sigma$ is a connector of $G$. We initialize $\sigma$ to the union of $\pi_b$ and $\pi_r$. Then, we add to $\sigma$ the edges $(b_t,r_s)$ and $(r_t,b_s)$. Clearly, $\sigma$ is a cycle that traverses all the vertices of $U_r$ and of $U_b$ consecutively. Moreover, in a traversal of $\sigma$ that starts at $b_t$ and traverses the edge $(b_t,r_s)$ from $b_t$ to $r_s$, we encounter $r_s$, $r_t$, and $b_s$ in this order. We thus only need to show that it is possible to draw in  
$\mathcal{E}_{G^*}$ the edges $(b_t,r_s)$ and $(b_t,b_s)$ while preserving planarity. To this aim, we let the drawing of $(b_t,r_s)$ in $\mathcal{E}$  be the (planar) one of the path $(b_t,\alpha,b^*,\omega,r_s)$ of $\mathcal{Q}$ in $\mathcal{E}_{H^*}$ and we let the drawing of  $(r_t,b_s)$ in $\mathcal{E}$ be the (planar) one of the path $(r_t,\chi,b^\diamond,\lambda,b_s)$ of $\mathcal{P}$ in $\mathcal{E}_{H^*}$. If $b_t$ immediately precedes $r_s$ counter-clockwise along $\sigma$, we are done; otherwise, we simply flip $\mathcal{E}$.
This concludes the proof of correctness of the reduction.
\end{proof} 

Since the \NP-completeness of {\sc B2PBE-PE} holds for $2$-connected instances with a fixed embedding, since the graphs resulting from a $\mathcal{PQ}$-merge of a $2$-connected graph are also $2$-connected (see \cref{prop:exterior-of-gadgets}), since $\mathcal P$ and $\mathcal Q$ are $3$-connected, and since the reduction of \cref{le:reduction} can be carried out while also enforcing a prescribed rotation system around the vertices $b'_t$ and $r'_s$ of $\mathcal Q$ and a prescribed rotation system around the vertices $b'_s$ and $r'_t$ of $\mathcal P$, we have shown that both {\CPLSF} and {\CPLS} are \NP-complete for $2$-connected instances with at most three clusters. This concludes the proof.
\end{proof} 

\section{Conclusions}

This paper established upper and lower bounds that significantly expand our understanding of the limits of tractability for finding linear saturators in the context of clustered planarity. 

We remark that, prior to this research, the problem was not known to be \NP-complete for instances with $\bigoh(1)$ clusters. Our \NP-hardness result for instances of {\sc CPLS} with three clusters narrows the complexity gap to its extreme (while also solving the open problem posed in~\cite[OP 4.3]{AngeliniL20} about the complexity of {\sc Clique Planarity} for instances with a bounded number of clusters), and animates the interest for the remaining two cluster case.

\clearpage
\bibliography{bibliography}

\end{document}